\documentclass{article}

\usepackage[T1]{fontenc}
\usepackage{graphicx}
\usepackage{tikz}
\usetikzlibrary{arrows.meta,positioning,matrix,calc}
\usepackage{quantikz}
\usepackage{amsmath}
\usepackage{amssymb}
\usepackage{amsthm}
\usepackage{mathtools}
\usepackage{booktabs}
\usepackage{array}
\usepackage{tabularx}
\usepackage{multirow}
\usepackage{makecell}
\usepackage{microtype}
\usepackage{needspace}
\usepackage{float}
\usepackage{placeins}
\usepackage{yhmath}
\usepackage[margin=1in]{geometry}
\usepackage[hidelinks]{hyperref}
\usepackage{xurl}
\usepackage{authblk}

\newtheorem{theorem}{Theorem}[section]
\newtheorem{lemma}[theorem]{Lemma}
\newtheorem{proposition}[theorem]{Proposition}
\newtheorem{corollary}[theorem]{Corollary}

\theoremstyle{definition}
\newtheorem{definition}[theorem]{Definition}

\newtheorem{example}{Example}[section]

\theoremstyle{remark}
\newtheorem{remark}[theorem]{Remark}
\numberwithin{equation}{section}
\newcolumntype{Y}{>{\raggedright\arraybackslash}X}
\newcolumntype{L}[1]{>{\raggedright\arraybackslash}p{#1}}

\usepackage[
  backend=biber,
  style=numeric-comp,
  sorting=none,
  sortcites=true
]{biblatex}
\AtBeginBibliography{\small\sloppy\setlength{\emergencystretch}{3em}}

\DeclareMathOperator{\Cl}{Cl}
\DeclareMathOperator{\GL}{GL}

\DeclareMathOperator{\Hom}{Hom}

\DeclareMathOperator{\rank}{rank}

\DeclareMathOperator{\Sp}{Sp}

\newcommand{\ZZ}{\mathbb{Z}}
\newcommand{\FF}{\mathbb{F}}
\newcommand{\CC}{\mathbb{C}}
\newcommand{\TT}{\mathbb{T}}
\newcommand{\UU}{\mathrm{U}}
\newcommand{\bH}{\mathcal{H}}
\newcommand{\Pauli}{\mathcal{P}}
\newcommand{\transpose}{\mathsf{T}}

\title{Hybrid Qubit--Rotor Quantum Systems: Clifford Structure, Universal Control, and Applications}
\author[1]{Dengyao Luo\thanks{dluo6@ncsu.edu}}
\author[1,2,3]{Arvin Kushwaha\thanks{research@arvinsk.org}}
\author[1,2]{Mastawal Tirfe}
\author[1]{Bojko N. Bakalov\thanks{bnbakalo@ncsu.edu}}
\date{August 20, 2026}
\affil[1]{Department of Mathematics, North Carolina State University,
Raleigh, North Carolina 27695, USA}
\affil[2]{Department of Physics, North Carolina State University,
Raleigh, North Carolina, 27695, USA}
\affil[3]{Department of Physics, Universit\"{a}t Stuttgart,
Stuttgart, Baden-W\"{u}rttemberg, 70569, Germany}
\begin{document}

\maketitle

\begin{abstract}
A $U(1)$ quantum rotor pairs a periodic angle with an integer-valued conjugate
momentum, and occurs in molecular rotation, superconducting phase--charge
circuits, and compact gauge fields. Coupling such a rotor coherently to qubits
gives a hybrid register whose control structure is not inherited from either
the oscillator--qubit or the qudit case. We develop a Clifford
theory, together with a universal-control result, for registers of $n$ qubits
and $r$ rotors. We classify all automorphisms of the hybrid phase space
$\FF_2^{2n}\times\ZZ^r\times\TT^r$ that preserve the Weyl commutation
relations, and give an explicit finite Clifford circuit for each one. The
classification is directional: rotor momentum parity may control qubit Pauli
operations within the Clifford group, while every nonzero qubit-controlled
rotor momentum shift is non-Clifford. It also yields normal forms for the
mixed qubit--rotor couplings and the exact minimum number of elementary mixed
gates needed to synthesize them. Adding a rotor cosine potential and one fixed
qubit--rotor conditional phase to the local Clifford operations gives
universal control on the full Hilbert space in the strong operator topology.
We then apply this structure in three settings: an exact controlled-shift
realization of gauge-covariant matter hopping, which is necessarily
non-Clifford; rotor phase estimation with direct angle readout and probe
optimization under momentum-support and energy constraints; and finite Fourier
transforms on rotor momentum codes, where the one-rotor transform for $d=2^s$
compiles into $O(s)$ momentum-selective and controlled-shift instructions and
each cross-register Fourier factor is implemented by one quadratic rotor
Clifford gate.
\end{abstract}

\begingroup
\small
\tableofcontents
\endgroup
\clearpage

\section{Introduction}

Quantum rotors arise in molecular rotation, rotational motion of trapped
particles, superconducting phase--charge degrees of freedom, and compact
$U(1)$ gauge fields
\cite{chou2017preparation,koch2019molecular,urban2019coherent,
glikin2025rotational,leibscher2025planar,martinetz2020electromechanics,
alcainecuervo2026compact}. Their Hilbert spaces also support
finite-dimensional logical encodings
\cite{raynal2010rotor,vuillot2024homological}, while planar-rotor models
can instead be digitized into qubit registers \cite{moeed2025qubit}. Here we
retain the rotor itself as a computational degree of freedom and allow it to
interact coherently with qubits.

Such hybrid qubit--rotor systems occur in several contexts. In compact $U(1)$
gauge--matter theory, gauge links carry integer electric flux while matter is
represented by finite-dimensional degrees of freedom \cite{kogut1975wilson}.
Recent quantum-simulation proposals likewise combine finite-dimensional matter
registers with infinite-dimensional gauge degrees of freedom
\cite{crane2024fermions,ale2026electrodynamics,alcainecuervo2026compact}, and
direct qubit--rotor couplings appear in hybrid quantum rotor devices
\cite{leitch2024thermodynamics}. Rotor quantum phase estimation uses a rotor
as the phase register for a qubit target
\cite{kemper2025continuousdiscrete}, while hybrid processor architectures
require coherent interfaces between finite- and infinite-dimensional registers
\cite{liu2026instruction,bierman2026statetransfer}. In each case, control is
required both within the individual registers and across the qubit--rotor
interface.

These examples share a common control problem. The two registers have
incompatible phase spaces, and once they are coupled, the transformations
preserving their joint commutation relations are no longer determined by
either one alone. Neither is the set of operations needed to go beyond them.

Much of hybrid quantum processing has been developed for oscillator modes
\cite{vanloock2008hybrid,andersen2015discrete,krastanov2015universal,
eickbusch2022fast,liu2026instruction}, and oscillator--qudit systems admit a
corresponding stabilizer and symplectic theory
\cite{chakraborty2026stabilizer}. Unlike an oscillator, which has two
continuous phase-space directions, a $U(1)$ rotor pairs a compact
configuration variable with a discrete momentum \cite{albert2017phasespaces}:
its angle $\theta\in\TT=\mathbb R/2\pi\ZZ$ and integer-valued angular momentum
$\ell\in\ZZ$ form the Pontryagin-dual pair
\begin{equation*}
  \ZZ\longleftrightarrow\TT .
\end{equation*}
For $n$ qubits and $r$ rotors, the Weyl operators are parametrized by the
Abelian group
\begin{equation*}
  K_{n,r}
  =
  \FF_2^{2n}\times\ZZ^r\times\TT^r .
\end{equation*}

The first consequence is an asymmetry in the mixed Clifford action. Rotor
momentum parity can control a qubit Pauli operation, but no Clifford operation
can make a qubit control a nonzero rotor momentum shift, since every
homomorphism from $\FF_2^{2n}$ to $\ZZ^r$ is trivial. This has no counterpart
in oscillator--qubit systems, where the mixed blocks vanish in both
directions, or in finite qudit systems, where they can be nonzero in both.

Normalizer circuits over Abelian groups provide the general setting for this
structure. Ref.~\cite{bermejovega2016normalizer} develops a generalized
stabilizer formalism for such groups, including the configuration group
$\FF_2^n\times\ZZ^r$ underlying $K_{n,r}$, and establishes the corresponding
Gottesman--Knill simulation result, which the hybrid Clifford circuits
constructed below inherit. Ref.~\cite{bauer2026quadratic} gives a broader
quadratic formulation of Clifford-type structures over Abelian groups, and the
pure-rotor Clifford structure is described in Ref.~\cite{xu2024multimode}. In
these treatments group automorphisms enter as normalizer-gate primitives, as
the simulation theorem requires.

We work out the qubit--rotor case at the level of circuits and cost. Every
automorphism preserving the Weyl commutator has the unique block form of
Theorem~\ref{thm:hybrid-phase-space-classification}, and each one is realized
by an explicit finite circuit assembled from qubit Clifford gates, rotor
Clifford gates, and elementary mixed gates
(Theorem~\ref{thm:explicit-hybrid-realization}). The qubit--rotor coupling is
carried entirely by one binary matrix, the mixed block $C$, which determines
both the parity-controlled Pauli action and the antisymmetric part of the
shear matrix $\Theta$. Under Clifford operations acting separately on the two
registers, $C$ is completely classified by the pair
$k=\rank_{\FF_2}C$ and
$h=\tfrac12\rank_{\FF_2}\!\bigl(C^{\mathsf T}J_nC\bigr)$
(Theorem~\ref{thm:classification-mixed-blocks}). Here $k$ counts the
independent qubit Pauli directions controlled by rotor momentum parity and
equals the minimum number of elementary mixed gates in the cost model of
Corollary~\ref{cor:minimum-mixed-gate-count}, while $h$ counts the independent
anticommuting pairs among those directions.

The same analysis applies directly to parameterized gates.
Sec.~\ref{sec:gate-families} determines the Clifford parameter values of the
momentum- and angle-dependent operations that arise on rotor platforms:
momentum-diagonal phases, momentum-dependent qubit gates, periodic angle
potentials, and qubit-controlled rotor displacements. In particular, every
nonzero qubit-controlled rotor momentum shift is non-Clifford. Beyond the
Clifford group, the local Clifford operations $\Cl_{n,r}^{\mathrm{loc}}$, the
cosine-potential family $K_1(t)$ on one rotor, and the fixed conditional phase
$C_{c_1,1}(\pi/4)$ generate a subgroup of $\UU(\bH_{n,r})$ that is dense in
the strong operator topology for $n,r\geq1$. The single-rotor step uses
spectral controllability results for infinite-dimensional bilinear
Schr\"odinger systems
\cite{chambrion2009discretespectrum,boscain2012spectral}. This is an existence
result: it gives no bound on circuit length or control time.

The first application is compact $U(1)$ gauge--matter dynamics. Rotor links
carry integer electric flux and qubits encode matter occupations
\cite{kogut1975wilson,crane2024fermions}; related implementations use
qubit--qumode or superconducting phase--charge degrees of freedom
\cite{ale2026electrodynamics,alcainecuervo2026compact}. Gauge-covariant
hopping transfers a matter excitation between neighboring sites together with
a unit change of the intervening link flux. We realize its time evolution
exactly on the full rotor Hilbert space by conjugating a qubit exchange with a
qubit-controlled rotor shift. The shift is non-Clifford, so gauge-covariant
hopping lies outside the Clifford group: every hopping term in a rotor-based
compact $U(1)$ simulation carries an irreducible non-Clifford cost. This
construction is exact only on the full rotor space; we therefore study
hard-wall electric-flux truncation through ground-state and real-time
calculations on an open two-plaquette model.

A second application uses the rotor as the phase register in quantum phase
estimation \cite{kemper2025continuousdiscrete}. The momentum label $\ell$
indexes the controlled powers $U^{-\ell}$ of the target unitary. Phase
kickback produces the character $e^{-i\ell\varphi}$ of $\ZZ$, which translates
the conjugate angle distribution under $\ZZ\longleftrightarrow\TT$. The
eigenphase can therefore be read directly from an angle measurement, without a
coherent inverse QFT on the phase register. The initial rotor probe determines
both the estimation accuracy and the required range of controlled powers. At
fixed momentum support, the optimum is the cosine-window probe
\cite{buzek1999optimal,imai2009fourier}; at fixed mean kinetic energy, the
optimum is a Mathieu probe \cite{hayashi2023special}. Under the
circular-variance criterion, both optimized probes give $O(E_R^{-1/2})$
angular-error scaling for $E_R=\langle\hat\ell^2\rangle$, compared with
$O(E_R^{-1/4})$ for the uniform finite-momentum probe.

The third application uses finite rotor codes for Fourier processing. We
construct the same finite Fourier matrix in two forms, related to hybrid
state-transfer and oscillator-assisted Fourier methods
\cite{liu2025mixed,bierman2026statetransfer}. One construction maps an angle
code to a momentum code using local angle-to-momentum Fourier operations and
mixed code-space phases. The other remains within a rotor momentum code and
compiles the transform using the hybrid gate set developed earlier. Under the
logical-instruction convention stated in
Sec.~\ref{subsec:qft-logical-resource-comparison}, the one-rotor compilation
for $d=2^s$ uses $O(s)$ instructions. Each cross-register base-$d$ Fourier
phase is implemented by the rotor--rotor quadratic Clifford gate
\begin{equation*}
  \mathrm{CPHS}_{jk}(\alpha)
  :=
  e^{i\alpha\hat\ell_j\hat\ell_k},
\end{equation*}
while the corresponding factor in the binary decomposition contains $s^2$
controlled-phase gates.

These constructions are formulated at the level of logical operations.
Sec.~\ref{sec:physical-implementation-considerations} connects them to
native-rotor and finite-code implementations. Superconducting phase--charge
circuits provide the quadratic momentum and cosine terms used in the control
construction
\cite{koch2007chargeinsensitive,alcainecuervo2026compact}. Finite-code
implementations can instead use spectrally selective qubit control,
conditional displacements, and nonlinear interactions on a chosen subspace of
another infinite-dimensional system
\cite{heeres2015cavity,eickbusch2022fast,diringer2024conditional,
rainaldi2026trigonometric}. For the compiled QFT, the main requirements are
resolution of individual binary digits of the momentum label and, for
$r\geq2$, a cross-register interaction angle $2\pi/d^r$.

\clearpage
\section{Clifford structure of hybrid qubit--rotor systems}
\label{sec:clifford-structure-hybrid-qubit-rotor-systems}
We begin by constructing the Weyl and Clifford structure of the hybrid
qubit--rotor register. The commutator pairing on $K_{n,r}=\FF_2^{2n}\times\ZZ^r\times\TT^r$ fixes the allowed phase-space transformations. 
Their block form leads to finite Clifford circuit realizations. 
Local changes of basis then reduce the mixed block to two rank invariants and give the elementary mixed-gate cost.

\subsection{Hybrid register and phase space}
\label{subsec:hybrid-register-phase-space}

Throughout this paper, a rotor means a $U(1)$ quantum rotor with configuration
group $\TT= \mathbb{R}/2\pi\ZZ \cong S^1$ and momentum group $\ZZ$. Consider
$n$ qubits and $r$ rotors,

\begin{equation}
  \bH_{n,r} = (\CC^2)^{\otimes n}\otimes\ell^2(\ZZ^r) \simeq (\CC^2)^{\otimes n}\otimes L^2(\TT^r).
\end{equation}

The rotor phase space pairs the discrete group $\ZZ^r$ with the compact group
$\TT^r$. We use the joint basis

\begin{equation}
  |u,\ell\rangle, \qquad u\in\FF_2^n, \quad \ell\in\ZZ^r.
\end{equation}

For one rotor, the generalized angle states are related to momentum states by
the Fourier relation

\begin{equation}
  |\theta\rangle = \frac{1}{\sqrt{2\pi}}\sum_{\ell \in \ZZ}e^{-i\ell\theta}|\ell\rangle.
\label{eq:rotor-angle-momentum-fourier-relation}
\end{equation}

All angle-dependent operators below are defined through periodic functions,
such as $e^{im\hat\theta}$ and $\cos\hat\theta$.

The joint qubit computational basis and rotor momentum basis are indexed by
the Abelian group
$
  A_{n,r}=\FF_2^n\times\ZZ^r.
$
Its Pontryagin dual,
$
  \widehat A_{n,r}=\FF_2^n\times\TT^r,
$
indexes the corresponding phase operators.

We define the hybrid qubit--rotor phase space by
\begin{equation}
  K_{n,r}
  =A_{n,r}\times\widehat A_{n,r}
  \cong
  \FF_2^{2n}\times\ZZ^r\times\TT^r.
\end{equation}
The Weyl commutator pairing introduced below equips $K_{n,r}$ with its
symplectic structure.

Write 

\begin{equation}
  x=(q,m,\phi),
  \qquad
  q=\begin{pmatrix}a\\ b\end{pmatrix}\in\FF_2^{2n},
  \quad
  m\in\ZZ^r,
  \quad
  \phi\in\TT^r,
\end{equation}

and define

\begin{equation}
  J_n=
  \begin{pmatrix}
    0&I_n\\
    I_n&0
  \end{pmatrix}.
\end{equation}

An overbar on an integer vector or matrix denotes entrywise reduction modulo two.
For a matrix $M$ over $\FF_2$, $\pi M$ denotes the corresponding matrix
with entries in $\TT$, under the entrywise identification
$0\mapsto0$ and $1\mapsto\pi$.

\subsection{Weyl multiplication and commutation relation}

Define the qubit Weyl operators by 

\begin{equation}
  X_Q(a)|u\rangle=|u+a\rangle,
  \qquad
  Z_Q(b)|u\rangle=(-1)^{b^{\transpose}u}|u\rangle,
\end{equation}

and define the rotor momentum shift $X_R(m)$ and rotor phase operator
$Z_R(\phi)$ by

\begin{equation}
  X_R(m)|\ell\rangle=|\ell+m\rangle,
  \qquad
  Z_R(\phi)|\ell\rangle=e^{i\phi^{\transpose}\ell}|\ell\rangle.
\end{equation}

We write the corresponding rotor Weyl displacement as
\begin{equation}
  D_R(m,\phi):=X_R(m)Z_R(\phi).
  \label{eq:rotor-weyl-displacement}
\end{equation}

For $x=(q,m,\phi)\in K_{n,r}$, define the associated hybrid Weyl operator by

\begin{equation}
  W(x)=W(q,m,\phi):=X_Q(a)Z_Q(b)\otimes X_R(m)Z_R(\phi).
\end{equation}

For $y=(q',m',\phi')\in K_{n,r}$, direct multiplication gives
\begin{equation}
  W(x)W(y)=\sigma(x,y)W(x+y),
  \qquad
  \sigma(x,y)=(-1)^{b^{\transpose}a'}e^{i\phi^{\transpose}m'}.
\end{equation}

Reversing the product order gives
\begin{equation}
  W(x)W(y)=\omega(x,y)W(y)W(x),
\end{equation}
where
\begin{equation}
\omega(x,y)
=\frac{\sigma(x,y)}{\sigma(y,x)}
=(-1)^{q^{\transpose}J_nq'}
\exp\!\left(i[\phi^{\transpose}m'-{\phi'}^{\transpose}m]\right).
  \label{eq:hybrid-commutator}
\end{equation}
The Weyl commutator pairing $\omega$ is continuous and multiplicative in
each phase-space argument.

\begin{lemma}[Perfectness of the commutator pairing]
\label{lemma:perfect}
The map
\begin{equation}
  D_\omega: K_{n,r} \longrightarrow \widehat K_{n,r}, \qquad y\longmapsto\omega(y,\,\cdot\,),
\end{equation}
is an isomorphism. Thus every continuous character of $K_{n,r}$ can be
written uniquely in the form $\omega(y,\,\cdot\,)$ for some
$y\in K_{n,r}$.
\end{lemma}

\begin{proof}
Set $\mathcal A=A_{n,r}$ and write
$K_{n,r}=\mathcal A\times\widehat{\mathcal A}$. For
$x=(u,\chi)$ and $y=(v,\psi)$,
Eq.~\eqref{eq:hybrid-commutator} becomes
\begin{equation}
  \omega(x,y)=\chi(v)\psi(u)^{-1}.
\end{equation}
Let
$
  \iota_{\mathcal A}:
  \mathcal A
  \longrightarrow
  \widehat{\widehat{\mathcal A}}
$
be the evaluation map. Under the natural identification
$
  \widehat{K_{n,r}}
  \cong
  \widehat{\mathcal A}
  \times
  \widehat{\widehat{\mathcal A}},
$
the map $D_\omega$ is
\begin{equation}
  D_\omega(u,\chi) = \bigl( \chi, -\iota_{\mathcal A}(u) \bigr).
\end{equation}
Pontryagin biduality identifies $\mathcal A$ with
$\widehat{\widehat{\mathcal A}}$. Hence $D_\omega$ is a topological group
isomorphism, which proves the claim
\cite[Example~1.1]{prasad2010symplectic}.
\end{proof}

\subsection{Hybrid Clifford group and induced phase-space action}

Adjoining scalar phases to the hybrid Weyl operators gives the hybrid
Pauli group
\begin{equation}
  \Pauli_{n,r}
  =\{zW(x):z\in\UU(1),\ x\in K_{n,r}\}.
\end{equation}
Since the commutator pairing is nondegenerate, the center of
$\Pauli_{n,r}$ consists of the scalar phases:
$
  Z(\Pauli_{n,r})=\UU(1)
$.
Moreover, the map $zW(x)\mapsto x$ induces a group isomorphism
$
  \Pauli_{n,r}/\UU(1)I\cong K_{n,r}.
$

The hybrid Clifford group is defined as the unitary normalizer of $\Pauli_{n,r}$:
\begin{equation}
  \Cl_{n,r}
  =N_{\UU(\bH_{n,r})}(\Pauli_{n,r}).
\end{equation}

\begin{lemma}[Clifford action on phase space]
For every $U\in\Cl_{n,r}$, conjugation induces a unique group automorphism
$S_U:K_{n,r}\to K_{n,r}$. For the chosen Weyl representatives, there is a
unique phase function $\alpha_U:K_{n,r}\to\UU(1)$ such that
\begin{equation}
  UW(x)U^\dagger
  =
  \alpha_U(x)W(S_Ux).
\end{equation}
Moreover, $S_U$ preserves the commutator pairing:
\begin{equation}
  \omega(S_Ux,S_Uy)=\omega(x,y),
  \qquad x,y\in K_{n,r}.
\end{equation}
\end{lemma}

\begin{proof}
Conjugation by $U$ induces an automorphism of
$\Pauli_{n,r}/\UU(1)I\cong K_{n,r}$, which defines $S_U$.
The remaining scalar defines $\alpha_U$. Since conjugation preserves
operator commutators, $S_U$ preserves $\omega$.
\end{proof}

Define
\begin{equation}
  \Sp(K_{n,r},\omega) = \left\{ S\in\operatorname{Aut}(K_{n,r}): \omega(Sx,Sy)=\omega(x,y)
  \text{ for all }x,y\in K_{n,r} \right\}.
\end{equation}
The preceding lemma therefore defines a homomorphism
$\pi:\Cl_{n,r}\longrightarrow\Sp(K_{n,r},\omega)$ by $U\longmapsto S_U$.

\subsection{Classification of hybrid phase-space transformations}

\begin{theorem}[Hybrid phase-space classification]
\label{thm:hybrid-phase-space-classification}
Every group automorphism of $K_{n,r}$ preserving $\omega$ is
automatically continuous and has the unique form
\begin{equation}
  S_{M,A,C,\Theta}(q,m,\phi) =\left( Mq+C\bar m, Am, A^{-\transpose} [\phi+\Theta m+\pi C^{\transpose}J_nMq]
  \right),
  \label{eq:general-hybrid-map}
\end{equation}

where
$M\in\Sp(2n,\FF_2),
  A\in\GL(r,\ZZ),
  C\in\operatorname{Mat}_{2n\times r}(\FF_2),
  \Theta\in\operatorname{Mat}_r(\TT),$
and

\begin{equation}
  \Theta-\Theta^{\transpose}
  =\pi C^{\transpose}J_nC.
  \label{eq:theta-constraint}
\end{equation}
Conversely, every such tuple defines an element of $\Sp(K_{n,r},\omega)$.
\end{theorem}

The matrix $M$ is the qubit symplectic action familiar from the stabilizer
formalism \cite{gottesman1999heisenberg}. The matrix $A$ preserves the rotor
momentum lattice, and $\Theta$ is the shear matrix. The binary matrix $C$ records
the mixed qubit--rotor action $ q\longmapsto q+C\bar m$.

Writing $c_j$ for the $j$th column of $C$, each
$c_j\in\FF_2^{2n}$ labels a qubit Pauli operator $P_{c_j}$ up to phase. We
call the nonzero columns the Pauli directions controlled by rotor momentum
parity. Their pairwise commutation relations are recorded by
$C^{\transpose}J_nC$: its $(i,j)$ entry is zero when $P_{c_i}$ and $P_{c_j}$
commute and one when they anticommute. These mixed commutators leave the
antisymmetric part of the shear no freedom. Eq.~\eqref{eq:theta-constraint}
fixes it as $\Theta-\Theta^{\transpose}=\pi C^{\transpose}J_nC$, so only the
symmetric part of $\Theta$ remains an independent parameter.

\begin{figure}[!ht]
  \centering
  \begin{tikzpicture}[
    formula/.style={
      draw,
      rounded corners=2pt,
      fill=black!2,
      align=left,
      inner xsep=12pt,
      inner ysep=8pt
    },
    block/.style={
      draw,
      rounded corners=2pt,
      fill=black!2,
      align=center,
      inner xsep=5pt,
      inner ysep=5pt,
      text width=29mm,
      minimum height=15mm,
      font=\small
    },
    subblock/.style={
      draw,
      rounded corners=2pt,
      align=center,
      inner xsep=4pt,
      inner ysep=4pt,
      text width=31mm,
      minimum height=13mm,
      font=\scriptsize
    },
    arr/.style={
      -{Latex[length=1.8mm]},
      line width=0.45pt
    }
  ]

    \node[formula] (S) {$
      \begin{array}{r@{}l}
        q'&=Mq+C\bar m,\\
        m'&=Am,\\
        \phi'&=A^{-\transpose}
        \bigl(\phi+\Theta m+\pi C^{\transpose}J_nMq\bigr).
      \end{array}$};

    \matrix (blocks) [
      below=10mm of S,
      matrix of nodes,
      nodes={block},
      column sep=3.5mm
    ] {
      |(M)| {$M$\\qubit symplectic action}
      &
      |(A)| {$A$\\rotor lattice action}
      &
      |(C)| {$C$\\mixed block}
      &
      |(Th)| {$\Theta$\\shear matrix}
      \\
    };

    \node[subblock] (Tsym)
      at ([xshift=-18mm,yshift=-19mm]Th.south)
      {symmetric additions to $\Theta$\\
       rotor quadratic gates};

    \node[subblock] (Tasym)
      at ([xshift=18mm,yshift=-19mm]Th.south)
      {$\Theta-\Theta^{\transpose}
        =\pi C^{\transpose}J_nC$\\
       fixed by the mixed block $C$};

    \draw[arr] ([xshift=-36mm]S.south) -- (M.north);
    \draw[arr] ([xshift=-12mm]S.south) -- (A.north);
    \draw[arr] ([xshift=12mm]S.south) -- (C.north);
    \draw[arr] ([xshift=36mm]S.south) -- (Th.north);

    \draw[arr] ([xshift=-5mm]Th.south) |- (Tsym.north);
    \draw[arr] ([xshift=5mm]Th.south) |- (Tasym.north);

  \end{tikzpicture}

  \caption{Block structure of the hybrid symplectic map. The matrices $M$ and
  $A$ describe the qubit symplectic and rotor lattice actions, and $C$ is the
  mixed block. The difference
  $\Theta-\Theta^{\transpose}=\pi C^{\transpose}J_nC$
  is fixed by the mixed block; symmetric additions to $\Theta$
  correspond to rotor quadratic gates.}
\end{figure}
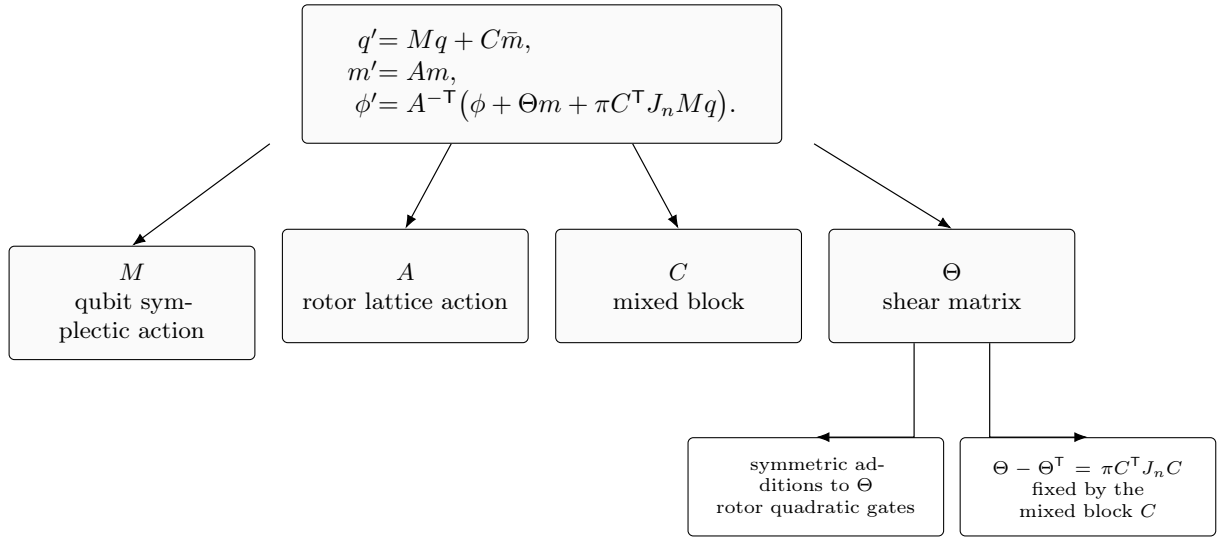

\begin{proof}[Proof sketch]
A homomorphism first has the block form
$
  q'=Mq+C\bar m,
  m'=Am,
  \phi'=\alpha(\phi)+Bm+\pi Gq.
$
The vanishing homomorphism groups and the uniqueness of the remaining
  blocks are established in
  Appendix~\ref{app:block-classification-hybrid-phase-space-transformations}.
The pure qubit commutators give $M^{\transpose}J_nM=J_n$.
The torus--integer commutators force
$\alpha(\phi)=A^{-\transpose}\phi$, which also removes discontinuous
torus automorphisms.  The finite--integer and integer--integer
commutators then determine $G$ and
Eq.~\eqref{eq:theta-constraint}.  Direct substitution proves the
converse.  The appendix also gives the composition law and inverse.
\end{proof}

As a first specialization of
Theorem~\ref{thm:hybrid-phase-space-classification}, set $n=0$. Then $q$ and
$M$ are absent, and the $0\times r$ mixed matrix $C$ vanishes. Hence
\begin{equation}
  S_{A,\Theta}(m,\phi)
  =\left(Am,A^{-\transpose}[\phi+\Theta m]\right),
  \qquad
  \Theta=\Theta^{\transpose}.
\end{equation}
Define the phase-space transformations
$
  L_A(m,\phi)=(Am,A^{-\transpose}\phi),
  N_\Theta(m,\phi)=(m,\phi+\Theta m).
$
Then $S_{A,\Theta}=L_AN_\Theta$, and a direct conjugation gives
\begin{equation}
  L_AN_\Theta L_A^{-1}
  =N_{A^{-\transpose}\Theta A^{-1}}.
\end{equation}
Hence the symmetric shear parameters form a normal subgroup on which
$\GL(r,\ZZ)$ acts by
$\Theta\mapsto A^{-\transpose}\Theta A^{-1}$.  Since a symmetric
$r\times r$ matrix with entries in $\TT$ has $r(r+1)/2$ independent entries, the
pure-rotor symplectic group is
\begin{equation}
  \operatorname{Sym}_r(\TT)\rtimes\GL(r,\ZZ)
  \simeq
  \UU(1)^{r(r+1)/2}\rtimes\GL(r,\ZZ).
\end{equation}
This recovers the known phase-space action of the pure-rotor Clifford group.
Representative rotor gates are described in Ref.~\cite{xu2024multimode}.
In the pure-qubit case, $r=0$ removes $A,C,$ and $\Theta$, leaving the
standard qubit symplectic group $\Sp(2n,\FF_2)$.

\subsection{Unitary realization of hybrid symplectic maps}

For $c=(a,b)\in\FF_2^{2n}$, define the Hermitian Pauli operator associated with $c$ by
$
  P_c=i^{a^{\transpose}b}X_Q(a)Z_Q(b),
  P_c=P_c^\dagger,
  P_c^2=I.
$
For rotor $j$, define
\begin{equation}
  \Gamma_{c,j}
  =\sum_{\ell\in\ZZ^r}
    P_c^{\,\ell_j\bmod2}\otimes|\ell\rangle\langle\ell|.
  \label{eq:elementary-mixed-gate}
\end{equation}
This gate applies $P_c$ when $\ell_j$ is odd and the identity when
$\ell_j$ is even.
Its induced mixed block is
$
  C=c e_j^{\transpose},
$
which has rank one when $c\neq0$.

When $n,r\geq1$, take $c=(0,e_1)$, so that $P_c=Z_1$, and set $j=1$.
Then
\begin{equation}
  \Gamma_{c,1} = C_\pi = e^{i\pi|1\rangle\langle1|_{q_1}\otimes\hat\ell_1}.
  \label{eq:Cpi-main}
\end{equation}

For $M\in\Sp(2n,\FF_2)$, choose a qubit Clifford $V_M$ inducing
$q\longmapsto Mq$.  For $A\in\GL(r,\ZZ)$, define the rotor
automorphism unitary by
\begin{equation}
  U_A|\ell\rangle=|A\ell\rangle.
  \label{eq:rotor-automorphism-unitary}
\end{equation}
For $R=R^{\transpose}\in\operatorname{Mat}_r(\TT)$, set
\begin{equation}
  f_R(\ell)
  =\sum_{j=1}^{r}R_{jj}\frac{\ell_j(\ell_j+1)}{2}
   +\sum_{1\leq i<j\leq r}R_{ij}\ell_i\ell_j,
  \label{eq:rotor-quadratic-phase}
\end{equation}
\begin{equation}
  Q_R|\ell\rangle
  =e^{if_R(\ell)}|\ell\rangle.
  \label{eq:rotor-quadratic-gate}
\end{equation}
Their induced actions on rotor phase space are
\begin{equation}
  \begin{split}
    U_A:\quad(m,\phi)
    &\longmapsto(Am,A^{-\transpose}\phi),
    \\
    Q_R:\quad(m,\phi)
    &\longmapsto(m,\phi+Rm).
  \end{split}
\end{equation}
Consequently, the combined unitary
$V_M\otimes(U_AQ_R)$ induces
\begin{equation}
  L_{M,A,R}(q,m,\phi) = \bigl( Mq,\, Am,\, A^{-\transpose}(\phi+Rm) \bigr).
  \label{eq:local-clifford-phase-space-action}
\end{equation}

The next theorem realizes every symplectic automorphism in
Theorem~\ref{thm:hybrid-phase-space-classification} by a finite circuit.
The chosen order of the mixed gates produces a triangular contribution to
the shear $\Theta$, and a quadratic gate acting only on the rotor register
supplies the remaining symmetric contribution.

\begin{theorem}[Explicit hybrid Clifford realization]
\label{thm:explicit-hybrid-realization}
Let $S=S_{M,A,C,\Theta}$ be a map classified in
Theorem~\ref{thm:hybrid-phase-space-classification}.  
Set
\begin{equation}
  D=M^{-1}C=(d_1,\ldots,d_r), \qquad \Gamma_D := \Gamma_{d_r,r}\cdots\Gamma_{d_1,1},
\end{equation}
so that the gates $\Gamma_{d_1,1},\ldots,\Gamma_{d_r,r}$ act in increasing order of the rotor index. 

Define
$\Theta_D^\triangle\in\operatorname{Mat}_r(\TT)$ by
\begin{equation}
  (\Theta_D^\triangle)_{ji}
  =
  \begin{cases}
    \pi d_j^{\transpose}J_nd_i,
      & i<j,\\
    0,
      & i\geq j,
  \end{cases}
\end{equation}
and set $\Theta_{\mathrm{loc}}:=\Theta-\Theta_D^\triangle$.
Then $\Theta_{\mathrm{loc}}$ is symmetric, and
\begin{equation}
  U_S
  :=
  \bigl[V_M\otimes(U_AQ_{\Theta_{\mathrm{loc}}})\bigr]\Gamma_D,
\end{equation}
\begin{equation}
  U_SW(x)U_S^\dagger
  \propto W(Sx), \qquad x\in K_{n,r}.
\end{equation}
\end{theorem}

\begin{proof}
First isolate the mixed part by setting $M=A=I$:
\begin{equation}
  S_{I,I,D,\Xi}(q,m,\phi) = \left( q+D\bar m,\, m,\, \phi+\Xi m+\pi D^{\transpose}J_nq \right).
\end{equation}
The elementary mixed gate satisfies
\begin{equation}
  \Gamma_{c,j}: (q,m,\phi) \longmapsto \left( q+c\bar m_j,\, m,\, \phi+\pi e_jc^{\transpose}J_nq \right).
  \label{eq:elementary-mixed-gate-action}
\end{equation}

Applying the factors in
$\Gamma_D=\Gamma_{d_r,r}\cdots\Gamma_{d_1,1}$ in increasing rotor order and
iterating Eq.~\eqref{eq:elementary-mixed-gate-action} gives
\begin{equation}
  \Gamma_D:(q,m,\phi)\longmapsto
  \left(q+D\bar m,\,m,\,
  \phi+\Theta_D^\triangle m+\pi D^{\transpose}J_nq\right).
\end{equation}
Thus $\Gamma_D$ implements $S_{I,I,D,\Theta_D^\triangle}$.
By construction,
\begin{equation}
  \Theta_D^\triangle-(\Theta_D^\triangle)^{\transpose}
  =\pi D^{\transpose}J_nD.
\end{equation}
Since $D=M^{-1}C$ and $M$ is symplectic,
$ 
  D^{\transpose}J_nD=C^{\transpose}J_nC.
$
Together with Eq.~\eqref{eq:theta-constraint}, this implies
$
  \Theta_{\mathrm{loc}}-\Theta_{\mathrm{loc}}^{\transpose}=0.
$
Hence $\Theta_{\mathrm{loc}}$ is symmetric, and
$Q_{\Theta_{\mathrm{loc}}}$ is a rotor quadratic gate.
The unitary
$V_M\otimes(U_AQ_{\Theta_{\mathrm{loc}}})$ implements
$L_{M,A,\Theta_{\mathrm{loc}}}$ from
Eq.~\eqref{eq:local-clifford-phase-space-action}. 

Using
$D=M^{-1}C$, $\Theta_D^\triangle+\Theta_{\mathrm{loc}}=\Theta$, and
$D^{\transpose}J_n=C^{\transpose}J_nM$ gives
\begin{equation}
  S_{M,A,C,\Theta}
  =L_{M,A,\Theta_{\mathrm{loc}}}\circ S_{I,I,D,\Theta_D^\triangle},
  \qquad
  U_SW(x)U_S^\dagger\propto W(Sx).
\end{equation}
This completes the construction of a unitary lift for every element of
$\Sp(K_{n,r},\omega)$.
\end{proof}

\begin{corollary}[Clifford--symplectic exact sequence]
\label{cor:clifford-symplectic-exact-sequence}
Conjugation on Weyl operators induces a surjective homomorphism
$
  \pi_{n,r}:\Cl_{n,r}\longrightarrow\Sp(K_{n,r},\omega),
$ by
$
  U\longmapsto S_U,
$
with kernel $\Pauli_{n,r}$.  Consequently,
\begin{equation}
  1 \longrightarrow \Pauli_{n,r} \longrightarrow \Cl_{n,r} \xrightarrow{\pi_{n,r}} \Sp(K_{n,r},\omega)
  \longrightarrow 1
\end{equation}
is a short exact sequence, and
$
  \Cl_{n,r}/\Pauli_{n,r}
  \simeq
  \Sp(K_{n,r},\omega).
$
\end{corollary}

\begin{proof}
The homomorphism property follows from composition of the induced phase-space
actions, and
surjectivity follows from Theorem~\ref{thm:explicit-hybrid-realization}.
It remains to identify the kernel.  If $S_U=\operatorname{id}$, then
$
  UW(x)U^\dagger=\chi(x)W(x)
$
for a phase function $\chi$ on $K_{n,r}$.  Comparison with the Weyl
multiplication law shows that $\chi$ is a character.  

Since $W(x)$ depends continuously on $x$, the identity
$
  UW(x)U^\dagger W(x)^\dagger=\chi(x)I
$
shows that $\chi$ is continuous.  
By Lemma~\ref{lemma:perfect}, there exists $y\in K_{n,r}$ such that
$\chi(x)=\omega(y,x)$ for all $x\in K_{n,r}$. Hence
$V:=W(y)^\dagger U$ commutes with every Weyl
operator.  Commutation with all qubit and rotor phase operators makes $V$
diagonal in the joint basis.  Commutation with all
qubit and rotor shifts makes its diagonal entries equal.  Thus $V$ is a
scalar and $U\in\Pauli_{n,r}$. The reverse inclusion follows directly from
the Weyl commutation relations. The quotient statement follows from the
first isomorphism theorem.
\end{proof}

\begin{corollary}[Generators of the hybrid Clifford group]
\label{cor:hybrid-clifford-generators}
For $n,r\geq1$,
\begin{equation}
  \Cl_{n,r} = \left\langle \Pauli_{n,r}, \Cl_{n,0}, \Cl_{0,r}, C_\pi \right\rangle .
  \label{eq:hybrid-generator-set}
\end{equation}
For $n=0$ or $r=0$, the mixed generator $C_\pi$ is omitted.
\end{corollary}

\begin{proof}
Every operator on the right-hand side of
Eq.~\eqref{eq:hybrid-generator-set} is a hybrid Clifford operator.

Conversely, let $U\in\Cl_{n,r}$, and let $S_U$ be its induced
phase-space action. By
Theorem~\ref{thm:explicit-hybrid-realization}, there is an explicitly
constructed unitary $U_{S_U}$ satisfying
$
  \pi_{n,r}(U_{S_U})=S_U.
$
The elementary mixed gates occurring in $U_{S_U}$ are generated by
$C_\pi$, Clifford operators acting only on one register, and Weyl operators,
while its remaining factors belong to $\Cl_{n,0}$ or $\Cl_{0,r}$. Hence
$U_{S_U}$ is generated
by the operators in Eq.~\eqref{eq:hybrid-generator-set}.

Since $U$ and $U_{S_U}$ induce the same phase-space action,
$
  UU_{S_U}^\dagger
  \in
  \ker\pi_{n,r}
  =
  \Pauli_{n,r}.
$
Therefore
$
  U=(UU_{S_U}^\dagger)U_{S_U}
$
is also generated by the operators in
Eq.~\eqref{eq:hybrid-generator-set}.

Fix $c\neq0$ and choose a qubit Clifford $V_c$ such that
$V_cZ_1V_c^\dagger=\varepsilon_cP_c$ with
$\varepsilon_c\in\{\pm1\}$. Choose a rotor permutation $U_{\pi_j}$ that
moves rotor $1$ to rotor $j$. Then
\begin{equation}
(V_c\otimes U_{\pi_j})C_\pi(V_c\otimes U_{\pi_j})^\dagger
=
\begin{cases}
  \Gamma_{c,j}, & \varepsilon_c=+1,\\
  Z_{R,j}(\pi)\Gamma_{c,j}, & \varepsilon_c=-1.
\end{cases}
\end{equation}
The extra factor in the second line is a rotor Weyl operator.
\end{proof}

\begin{example}[One qubit and one rotor]
For the smallest nontrivial hybrid register,
\begin{equation}
    \Sp(K_{1,1},\omega) \cong \TT\times S_4\times\ZZ_2.
\end{equation}
Here the factor $\TT$ is the one-rotor shear parameter, $\ZZ_2$ records
momentum reversal, and
$S_4\cong\operatorname{AGL}(2,\FF_2)$ combines the one-qubit symplectic
action with the mixed parity-controlled Pauli action. The group identification
is proved in Appendix~\ref{app:one-qubit-one-rotor-quotient}.
\end{example}

\paragraph{Classical simulation.}
The generators in Corollary~\ref{cor:hybrid-clifford-generators} are
normalizer operations over the Abelian group
$A_{n,r}=\FF_2^n\times\ZZ^r$. Hybrid Clifford circuits with the stabilizer
inputs and measurements are therefore covered by the generalized Gottesman--Knill simulation result of Ref.~\cite{bermejovega2016normalizer}.

\subsection{Classification of mixed Clifford couplings}

Let
$
  V=\FF_2^{2n},
  \beta_J(u,v)=u^{\transpose}J_nv,
$
and regard $C$ as a linear map from $\FF_2^r$ to $V$. Clifford operations
acting separately on the qubit and rotor registers transform $C$ as
\begin{equation}
  C\longmapsto MCR, \qquad M\in\Sp(2n,\FF_2), \quad R\in\GL(r,\FF_2).
  \label{eq:mixed-block-equivalence-action}
\end{equation}
Here $R=\overline A^{-1}$ for a rotor basis change
$A\in\GL(r,\ZZ)$. Every $R$ occurs because reduction modulo two maps
$\GL(r,\ZZ)$ onto $\GL(r,\FF_2)$. The right action changes the basis of
the domain; the left action sends $\operatorname{im}C$ to
$M(\operatorname{im}C)$.

Under this action, the orbit of $C$ is determined by the alternating space
$(\operatorname{im}C,\beta_J|_{\operatorname{im}C})$. Witt extension lifts
its normal form to the ambient
symplectic space.

\begin{theorem}[Classification of mixed blocks]
\label{thm:classification-mixed-blocks}
The pair
\begin{equation}
  (k,h) = \left( \rank_{\FF_2}C, \frac{1}{2} \rank_{\FF_2}(C^{\transpose}J_nC) \right)
\end{equation}
is a complete invariant of the action in
Eq.~\eqref{eq:mixed-block-equivalence-action}. Two mixed blocks are equivalent
under separate qubit and rotor Clifford operations if and only if they have
the same pair $(k,h)$.

Fix a symplectic basis
$e_1,f_1,\ldots,e_n,f_n$ of $V$. Every orbit contains a representative
whose nonzero columns are
\begin{equation}
  e_1,f_1,\ldots,e_h,f_h,
  e_{h+1},\ldots,e_{k-h},
  \label{eq:mixed-block-representative}
\end{equation}
followed by $r-k$ zero columns. The admissible values are
\begin{equation}
  0\leq k\leq\min(r,2n), \qquad \max(0,k-n) \leq h \leq \left\lfloor\frac{k}{2}\right\rfloor .
  \label{eq:mixed-block-range}
\end{equation}
\end{theorem}

\begin{proof}
\emph{Invariance.}
Suppose that
$
  C'=MCR,
  M\in\Sp(2n,\FF_2),
  R\in\GL(r,\FF_2).
$
Then
$
  \rank C'=\rank C
$
and
$
  (C')^{\transpose}J_nC'
  =
  R^{\transpose}(C^{\transpose}J_nC)R.
$
Hence both $k$ and $h$ are invariant under this action.

The quotient isomorphism
$
  \FF_2^r/\ker C
  \cong
  \operatorname{im}C
$
identifies the alternating form represented by
$C^{\transpose}J_nC$ with the restriction of $\beta_J$ to
$\operatorname{im}C$. Consequently, $k=\dim(\operatorname{im}C)$ and
$2h=\rank\!\left(\left.\beta_J\right|_{\operatorname{im}C}
  \right).$

\emph{Completeness.}
Now suppose that $C$ and $C'$ have the same pair $(k,h)$. The restricted
alternating spaces
$
  \left(\operatorname{im}C,\left.\beta_J\right|_{\operatorname{im}C}\right)$ and 
$
  \left(\operatorname{im}C',\left.\beta_J\right|_{\operatorname{im}C'}\right)
$
have the same dimension and rank, and are therefore isometric. By Witt
extension, this isometry extends to an element $M\in\Sp(2n,\FF_2)$
satisfying $M(\operatorname{im}C)=\operatorname{im}C'$
\cite[Theorem~3.4]{sprehn2020forms}.

The maps $C$ and $M^{-1}C'$ are surjections
$\FF_2^r\longrightarrow\operatorname{im}C$ of the same rank. Choosing bases of
$\FF_2^r$ adapted to their kernels gives $R\in\GL(r,\FF_2)$ such that
$
  CR=M^{-1}C'.
$
Hence $MCR=C'$, so the pair $(k,h)$ completely determines the orbit.

\emph{Normal form and admissible range.}
The space $\operatorname{im}C$ has a basis consisting of $h$ symplectic pairs
and $k-2h$ vectors in the radical of
$\left.\beta_J\right|_{\operatorname{im}C}$. A symplectic transformation of
$V$ sends
such a basis to
$
  e_1,f_1,\ldots,e_h,f_h,
  e_{h+1},\ldots,e_{k-h}.
$
A change of basis in $\FF_2^r$ then places these vectors in the first $k$
columns and makes the remaining $r-k$ columns zero. 

This gives
Eq.~\eqref{eq:mixed-block-representative}.
The inequalities
$
  k\leq r,
  k\leq2n,
  2h\leq k
$
give the upper bounds in
Eq.~\eqref{eq:mixed-block-range}. Moreover, the normal form contains
an isotropic subspace of dimension $k-h$. Since an isotropic subspace of
$V$ has dimension at most $n$, together with $h\geq0$, this gives the stated lower bound. Conversely,
the columns in
Eq.~\eqref{eq:mixed-block-representative} exist for every pair
$(k,h)$ in this range.
\end{proof}

In this normal form, $k$ counts the independent qubit Pauli directions
controlled by rotor momentum parity, and $h$ counts the independent
anticommuting pairs among them. The remaining $k-2h$ directions lie in the
radical of the restricted alternating form.
Through Eq.~\eqref{eq:theta-constraint}, the binary matrix
$C^{\transpose}J_nC$ also fixes $\Theta-\Theta^{\transpose}$, and
$\rank(C^{\transpose}J_nC)=2h$.

\paragraph{Mixed-gate cost model.}
Qubit-only and rotor-only Clifford operations are free, and each elementary
mixed gate $\Gamma_{c,j}$ counts once.

\begin{corollary}[Minimum number of elementary mixed gates]
\label{cor:minimum-mixed-gate-count}
In this model, the minimum number of elementary mixed gates required to
realize an admissible mixed block $C$ is
$N_{\Gamma}^{\min}(C)=\rank_{\FF_2}C$.
\end{corollary}

\begin{proof}
These free operations act on $C$ by invertible left and right
multiplication and therefore preserve its rank. Under composition,
Eq.~\eqref{eq:hybrid-symplectic-map-composition} gives
$C_{12}=M_1C_2+C_1\bar A_2$. Hence a circuit containing $t$ elementary
mixed gates has a total mixed block that is a sum of at most $t$ rank-one
contributions, each multiplied on the left and right by invertible matrices.
Rank subadditivity therefore gives
$
  \rank_{\FF_2}C\leq t.
$
Thus at least $\rank_{\FF_2}C$ elementary mixed gates are required.

Conversely, the representative in
Eq.~\eqref{eq:mixed-block-representative} has
$k=\rank C$ nonzero columns, and each column is realized by one elementary
mixed gate. The ordered product of these gates may also produce a rotor
shear $\Theta_*$. Since both $\Theta$ and $\Theta_*$ satisfy
$
  \Theta-\Theta^{\transpose}
  =
  \Theta_*-\Theta_*^{\transpose}
  =
  \pi C^{\transpose}J_nC,
$
their difference $\Theta-\Theta_*$ is symmetric and can be implemented by
a Clifford operation acting only on the rotor register. Thus $\rank C$
elementary mixed gates are sufficient.
\end{proof}

\section{Gate families}
\label{sec:gate-families}

The applications below use gates whose parameters depend on rotor momentum
or angle. This section determines when those gates are Clifford. We
also examine qubit-controlled rotor displacements, which reverse the
direction of control in the elementary mixed Clifford gate.

\subsection{Pure-rotor and mixed Clifford gates}
\label{subsec:pure-rotor-mixed-clifford-gates}

We regard the pure-qubit Clifford group as standard and recall only the
rotor and mixed gates used below. Up to a global phase, the pure-rotor
Clifford group is generated by Weyl operators, momentum-lattice
automorphisms, and quadratic phase gates. These operations are implemented
by $D_R(m,\phi)$, $U_A$, and $Q_R$, respectively, as defined in
Eqs.~\eqref{eq:rotor-weyl-displacement},
\eqref{eq:rotor-automorphism-unitary}, and
\eqref{eq:rotor-quadratic-gate} \cite{xu2024multimode}.

The following gates give concrete instances of the latter two types. For
$i\neq j$ and $\alpha\in\TT$, we use
\begin{equation}
  \begin{array}{r@{}l}
    \mathrm{PAR}_j|\ell_j\rangle
    &=|-\ell_j\rangle,
    \\
    \mathrm{SUM}_{i\to j}|\ell_i,\ell_j\rangle
    &=|\ell_i,\ell_j+\ell_i\rangle,
    \\
    \mathrm{QUAD}_j(\alpha)|\ell_j\rangle
    &=e^{i\alpha\ell_j(\ell_j+1)/2}|\ell_j\rangle,
    \\
    \mathrm{CPHS}_{ij}(\alpha)|\ell_i,\ell_j\rangle
    &=e^{i\alpha\ell_i\ell_j}|\ell_i,\ell_j\rangle.
  \end{array}
\end{equation}
PAR, SUM, and rotor SWAP implement elementary automorphisms of the momentum
lattice $\ZZ^r$. QUAD and CPHS implement the diagonal and off-diagonal
entries of the symmetric shear matrix $R$.

In the hybrid system, the elementary mixed Clifford gate $\Gamma_{c,j}$ in
Eq.~\eqref{eq:elementary-mixed-gate} acts as the identity for even $\ell_j$
and as $P_c$ for odd $\ell_j$. For $P_c=Z_1$ and $j=1$, this family reduces
to the one-qubit--one-rotor gate $C_\pi$ in Eq.~\eqref{eq:Cpi-main}. For
$n,r\geq1$,
Eq.~\eqref{eq:hybrid-generator-set} shows that $C_\pi$, together with
$\Pauli_{n,r}$, $\Cl_{n,0}$, and $\Cl_{0,r}$, generates $\Cl_{n,r}$.

\subsection{Momentum-dependent gates}

Momentum-dependent gates fall into two classes: scalar phase gates diagonal
in the momentum basis and qubit gates selected by the rotor momentum. The
phase in the first family must admit a quadratic representative on $\ZZ^r$.
In the second, apart
from a scalar rotor phase, the qubit blocks can vary only through
momentum-parity powers of a fixed Pauli operator.

\subsubsection{Momentum-diagonal phase gates}

Let $u:\ZZ^r\to\UU(1)$ and define
\begin{equation}
  G_u|\ell\rangle=u(\ell)|\ell\rangle.
\end{equation}

\begin{proposition}[Clifford criterion for momentum-diagonal phase gates]
\label{prop:momentum-diagonal-phases}
The gate $G_u$ belongs to $\Cl_{0,r}$ if and only if there are
$\gamma\in\TT$, $\lambda\in\TT^r$, and
$R=R^{\transpose}\in\operatorname{Mat}_r(\TT)$ such that
\begin{equation}
  u(\ell) = e^{i\gamma+i\lambda^{\transpose}\ell+if_R(\ell)} \text{for every }\ell\in\ZZ^r,
  \label{eq:momentum-diagonal-rotor-clifford-form}
\end{equation}
where $f_R$ is defined in Eq.~\eqref{eq:rotor-quadratic-phase}.
\end{proposition}

\begin{proof}[Proof sketch]
Conjugating $X_R(e_j)$ shows that, if $G_u$ is Clifford, the factor
$u(\ell+e_j)u(\ell)^{-1}$ must be a character of $\ZZ^r$, up to a scalar.
Compatibility of these finite differences gives a symmetric matrix $R$.
Appendix~\ref{app:momentum-diagonal-phase-gate-proof} solves the resulting
difference equations on $\ZZ^r$ and obtains
Eq.~\eqref{eq:momentum-diagonal-rotor-clifford-form}.
Conversely, this form gives $G_u=e^{i\gamma}Z_R(\lambda)Q_R$, which is a
pure-rotor Clifford gate.
\end{proof}

The degree of a polynomial expression for the phase does not by itself
determine whether the gate is Clifford, since different polynomials may agree
on the integer momentum spectrum. For one rotor,
\begin{equation}
  e^{i\pi\hat\ell^3} = e^{i\pi\hat\ell} = Z_R(\pi),
\end{equation}
because $\ell^3-\ell$ is even for every $\ell\in\ZZ$.

\subsubsection{Momentum-dependent qubit gates}

For one rotor, let $\{V_\ell\}_{\ell\in\ZZ}$ be a sequence of
$n$-qubit unitaries and consider the gate
\begin{equation}
  \mathcal V
  =
  \sum_{\ell\in\ZZ}
  V_\ell\otimes|\ell\rangle\langle\ell|.
\end{equation}
The qubit operation in each block depends on the rotor momentum $\ell$.

\begin{proposition}[Clifford criterion for momentum-dependent qubit gates]
\label{prop:momentum-dependent-qubit-gates}
The gate $\mathcal V$ belongs to $\Cl_{n,1}$ if and only if there
exist
$
  V_0\in\Cl_{n,0},
  c\in\FF_2^{2n},
  \alpha,\lambda\in\TT
$
such that
\begin{equation}
  V_\ell = e^{i[\lambda\ell+\alpha\ell(\ell+1)/2]} P_c^{\,\ell}V_0
  \qquad \text{for all }\ell\in\ZZ.
  \label{eq:momentum-dependent-qubit-clifford-form}
\end{equation}
\end{proposition}

\begin{proof}[Proof sketch]
Conjugating $I\otimes X_R(1)$ shows that Cliffordness forces
$V_{\ell+1}V_\ell^\dagger=e^{i(\rho+\alpha\ell)}P_c$ for a fixed Pauli
operator $P_c$. Appendix~\ref{app:momentum-dependent-qubit-gate-proof} solves
this recurrence for every $\ell\in\ZZ$ and gives
Eq.~\eqref{eq:momentum-dependent-qubit-clifford-form} with a fixed unitary
$V_0$. Conjugating the qubit Weyl operators shows that $V_0\in\Cl_{n,0}$.
Conversely, writing $\mathcal V=(I\otimes G_u)\Gamma_{c,1}(V_0\otimes I)$,
with $u(\ell)=e^{i[\lambda\ell+\alpha\ell(\ell+1)/2]}$, proves that
$\mathcal V$ is Clifford.
\end{proof}

A simple specialization is the conditional phase gate
\begin{equation}
  C_{c,j}(\varphi) = \exp\!\left[ i\varphi\frac{I-P_c}{2}\otimes\hat\ell_j \right].
  \label{eq:momentum-dependent-conditional-phase}
\end{equation}

\begin{corollary}[Clifford criterion for the conditional-phase gate]
\label{cor:conditional-phase-clifford-criterion}
For $n,r\geq1$, $c\in\FF_2^{2n}\setminus\{0\}$,
$1\leq j\leq r$, and $\varphi\in\TT$,
\begin{equation}
  C_{c,j}(\varphi)\in\Cl_{n,r}
  \quad\Longleftrightarrow\quad
  \varphi\equiv0\ \text{or}\ \pi\pmod{2\pi}.
\end{equation}
The case $\varphi\equiv0$ gives the identity, while
$
  C_{c,j}(\pi)=\Gamma_{c,j}.
$
\end{corollary}

\begin{proof}
The relative unitary between adjacent momentum blocks is
$
  \exp\!\left[
    i\varphi\frac{I-P_c}{2}
  \right].
$
Its eigenvalues are $1$ and $e^{i\varphi}$. For this operator to be
proportional to a Pauli operator, these eigenvalues must be either equal or
opposite. Hence $e^{i\varphi}=1$ or $e^{i\varphi}=-1,$
which gives
$\varphi\equiv0$ or $\pi\pmod{2\pi}$.
For $\varphi=\pi$, the relative unitary is $P_c$, and the gate is
$\Gamma_{c,j}$.
\end{proof}

A broader family is given by momentum-selective qubit rotations (MQRs).
For one qubit, let $X$ and $Y$ denote the Pauli matrices and define
\begin{equation}
  \mathrm{MQR} = \sum_{\ell\in\ZZ} R_{\phi_\ell}(\theta_\ell) \otimes|\ell\rangle\langle\ell|,
\end{equation}
\begin{equation}
  R_\phi(\theta) = \exp\!\left[ -\frac{i\theta}{2} (\cos\phi\,X+\sin\phi\,Y) \right],
\end{equation}
where $\phi_\ell\in\TT$ and $\theta_\ell\in\mathbb R$.
By Proposition~\ref{prop:momentum-dependent-qubit-gates}, such a gate is
Clifford if and only if its blocks have the form
\begin{equation}
  R_{\phi_\ell}(\theta_\ell) = e^{i[\lambda\ell+\alpha\ell(\ell+1)/2]} P_c^{\,\ell}V_0 \qquad
  \text{for all }\ell\in\ZZ,
  \label{eq:mqr-clifford-form}
\end{equation}
for some fixed $V_0\in\Cl_{1,0}$, fixed Pauli operator $P_c$, and
$\alpha,\lambda\in\TT$.
Thus, apart from a scalar phase depending quadratically on $\ell$, the
qubit operation alternates with momentum parity between $V_0$ and $P_cV_0$.
An MQR outside this form is non-Clifford. A simple Clifford example is
$
  R_0(\pi\ell)
  =
  (-iX)^\ell
  =
  e^{-i\pi\ell/2}X^\ell.
$
This is of the form in
Eq.~\eqref{eq:mqr-clifford-form}, with a linear scalar phase and
momentum-parity-dependent powers of $X$.

\subsection{Angle-dependent phase gates}

Angle-dependent phase gates include pure-rotor multiplication operators and
phases coupled to a qubit Pauli operator.

\subsubsection{Pure-rotor phase gates}

Let $g:\TT^r\to\UU(1)$ be continuous, and define the multiplication
operator
\begin{equation}
  [M_g\psi](\theta)
  =
  g(\theta)\psi(\theta).
\end{equation}
The gate $M_g$ commutes with every $X_R(m)$. Its conjugation of an angle
translation is
\begin{equation}
  M_gZ_R(\phi)M_g^\dagger
  =
  M_{g(\theta)\overline{g(\theta+\phi)}}Z_R(\phi).
  \label{eq:angle-dependent-phase-translation}
\end{equation}
If $M_g$ is Clifford, this factor must be proportional to a torus character for
every $\phi\in\TT^r$.

\begin{proposition}[Clifford criterion for pure-rotor phase gates]
\label{prop:pure-rotor-phase-gate-clifford-criterion}
The gate $M_g$ belongs to $\Cl_{0,r}$ if and only if there exist
$\gamma\in\TT$ and $s\in\ZZ^r$ such that
\begin{equation}
  g(\theta) = e^{i\gamma+i s^{\transpose}\theta} \qquad \text{for all }\theta\in\TT^r.
\end{equation}
Up to a global phase, $M_g$ is the rotor Weyl operator
$X_R(s)$.
\end{proposition}

\begin{proof}
Suppose that $M_g$ is Clifford. For each $\phi\in\TT^r$,
Eq.~\eqref{eq:angle-dependent-phase-translation} implies
\begin{equation}
  g(\theta)\overline{g(\theta+\phi)}
  =
  c_\phi e^{is_\phi^{\transpose}\theta}
  \label{eq:angle-phase-character}
\end{equation}
for some $c_\phi\in\UU(1)$ and $s_\phi\in\ZZ^r$.

Translation preserves the winding number of $g$ in every angle variable;
complex conjugation reverses it. Hence the left-hand side
of Eq.~\eqref{eq:angle-phase-character} has zero winding number in every
angle variable. The winding numbers of the character on the right-hand
side are the components of $s_\phi$, so $s_\phi=0$.

It follows that
$
  g(\theta+\phi)=\chi(\phi)g(\theta)
$
for some continuous function $\chi:\TT^r\to\UU(1)$. Applying two
successive translations gives
$
  \chi(\phi+\phi')
  =
  \chi(\phi)\chi(\phi'),
$
so $\chi$ is a continuous character of $\TT^r$. Therefore
$
  \chi(\phi)=e^{is^{\transpose}\phi}
$
for some $s\in\ZZ^r$. Taking $\theta=0$ gives
$
  g(\theta)=e^{i\gamma+is^{\transpose}\theta}
$
for some $\gamma\in\TT$.

Conversely, a function of this form gives
$
  M_g=e^{i\gamma}X_R(s),
$
which is a rotor Weyl operator and hence a Clifford gate.
\end{proof}

Let $f\in C(\TT^r,\mathbb R)$, and let $M_f$ be the multiplication
operator
$
  [M_f\psi](\theta)=f(\theta)\psi(\theta).
$
Define
\begin{equation}
  V_f(t)=e^{-itM_f},
  \qquad
  t\in\mathbb R.
\end{equation}

\begin{corollary}[Clifford criterion for real angle potentials]
\label{cor:real-angle-potential-clifford-criterion}
The gate $V_f(t)$ belongs to $\Cl_{0,r}$ if and only if
$t=0$ or $f$ is constant.
\end{corollary}

\begin{proof}
The case $t=0$ gives the identity. Suppose that $t\neq0$ and that
$V_f(t)$ is Clifford. By
Proposition~\ref{prop:pure-rotor-phase-gate-clifford-criterion}, there
exist $\gamma\in\TT$ and $s\in\ZZ^r$ such that
\begin{equation}
  e^{-itf(\theta)}
  =
  e^{i\gamma+is^{\transpose}\theta}.
\end{equation}
Because $f$ is a continuous real-valued function on $\TT^r$, the
left-hand side has zero winding number in every angle variable. The
winding numbers of the right-hand side are the components of $s$.
Therefore $s=0$, and $e^{-itf(\theta)}$ is constant.

It follows that $ t\bigl[f(\theta)-f(0)\bigr]\in2\pi\ZZ$ for all $\theta\in\TT^r$.
The left-hand side is continuous in $\theta$, while $2\pi\ZZ$ is
discrete. Since $\TT^r$ is connected, it must vanish identically.
Hence $f$ is constant. The converse is immediate, since a constant
$f$ gives only a global phase.
\end{proof}

For one rotor, the cosine-potential gate is
\begin{equation}
  K(t)=e^{-itM_{\cos}},
  \qquad
  [M_{\cos}\psi](\theta)=\cos\theta\,\psi(\theta).
  \label{eq:cosine-potential-gate}
\end{equation}
Since $\cos\theta$ is nonconstant,
Corollary~\ref{cor:real-angle-potential-clifford-criterion} shows that
$K(t)$ is non-Clifford for every $t\neq0$.

\subsubsection{Pauli-string--rotor phase gates}

For $n,r\geq1$ and $c\in\FF_2^{2n}\setminus\{0\}$, let $P_c$ be the
corresponding Hermitian Pauli string. Given
$f\in C(\TT^r,\mathbb R)$ and $\tau\in\mathbb R$, define the
Pauli-string--rotor phase gate (PSR gate)
\begin{equation}
  \mathrm{PSR}_{c,f}(\tau) = \exp\!\left[ -i\tau P_c\otimes M_f \right].
  \label{eq:pauli-string-rotor-phase-gate}
\end{equation}
The gate restricts to $V_f(\tau)$ and $V_f(-\tau)$ on the $+1$ and $-1$
eigenspaces of $P_c$, respectively.  All nonidentity Pauli operators are
Clifford-conjugate up to sign, so it suffices to analyze one representative.

When $f(\theta)=f_0$ is constant, the rotor dependence drops out and the gate
reduces to the Pauli rotation $e^{-i\tau f_0P_c}$.

\begin{proposition}[Clifford criterion for PSR gates]
For $\tau=0$, the gate $\mathrm{PSR}_{c,f}(\tau)$ is the identity.
Suppose that $\tau\neq0$. Then
$\mathrm{PSR}_{c,f}(\tau)$ is Clifford if and only if $f$ is
constant, say $f(\theta)=f_0$, and $\tau f_0\in\frac{\pi}{4}\ZZ.$
\end{proposition}

\begin{proof}
The case $\tau=0$ is immediate. Assume that $\tau\neq0$.
Since $P_c$ is nonidentity, choose a Hermitian Pauli string $P_d$ that
anticommutes with $P_c$. By
$
  P_c^2=I,
  P_cP_d=-P_dP_c,
$
we obtain
\begin{equation}
  \mathrm{PSR}_{c,f}(\tau) (P_d\otimes I_R) \mathrm{PSR}_{c,f}(\tau)^\dagger = P_d\otimes M_{\cos(2\tau f)}
  -iP_cP_d\otimes M_{\sin(2\tau f)},
  \label{eq:pauli-string-phase-pauli-conjugation}
\end{equation}
where $M_{\cos(2\tau f)}$ and $M_{\sin(2\tau f)}$ denote multiplication
by $\cos(2\tau f(\theta))$ and $\sin(2\tau f(\theta))$, respectively.

Suppose first that $\mathrm{PSR}_{c,f}(\tau)$ is Clifford. Then the
operator in
Eq.~\eqref{eq:pauli-string-phase-pauli-conjugation} must be proportional
to a single hybrid Weyl operator. The Pauli strings $P_d$ and $-iP_cP_d$ are
linearly independent, but a hybrid Weyl operator contains only one qubit
Pauli component. Therefore one of the two coefficient
functions must vanish identically:
\begin{equation}
  \sin\!\bigl(2\tau f(\theta)\bigr)=0
  \qquad
  \text{for all }\theta\in\TT^r,
\end{equation}
or
\begin{equation}
  \cos\!\bigl(2\tau f(\theta)\bigr)=0
  \qquad
  \text{for all }\theta\in\TT^r.
\end{equation}

In the first case, $2\tau f(\theta)\in\pi\ZZ$ for all $\theta$, while in the
second case, $2\tau f(\theta)\in\frac{\pi}{2}+\pi\ZZ$ for all $\theta.$

In either case, the continuous function $f$ takes values in a discrete
subset of $\mathbb R$. Since $\TT^r$ is connected, $f$ must be constant.
Write $f(\theta)=f_0, a=\tau f_0.$

The PSR gate then reduces to
\begin{equation}
  \mathrm{PSR}_{c,f}(\tau)
  =
  e^{-iaP_c}\otimes I_R.
\end{equation}

For a Pauli string $P_d$ anticommuting with $P_c$,
\begin{equation}
  e^{-iaP_c}P_de^{iaP_c}
  =
  \cos(2a)P_d-i\sin(2a)P_cP_d.
  \label{eq:pauli-rotation-conjugation}
\end{equation}
For this operator to be proportional to a single Pauli string, one of
the two coefficients must vanish. Hence $\sin(2a)=0$ or $\cos(2a)=0.$
Thus the gate can be Clifford only if $\tau f_0\in\frac{\pi}{4}\ZZ.$

Conversely, suppose that $f(\theta)=f_0$ and
$\tau f_0\in(\pi/4)\ZZ$. Any qubit Pauli string commuting with $P_c$
is unchanged under conjugation by $e^{-i\tau f_0P_c}$. If a qubit Pauli string instead anticommutes with $P_c$, Eq.~\eqref{eq:pauli-rotation-conjugation}
shows that it is mapped, up to a scalar phase, to either itself or
$P_cP_d$. Thus $e^{-i\tau f_0P_c}$ normalizes the qubit Pauli group.
Since the gate acts trivially on the rotor subsystem,
$\mathrm{PSR}_{c,f}(\tau) = e^{-i\tau f_0P_c}\otimes I_R \in\Cl_{n,r}.$
\end{proof}

For $1\leq j\leq n$ and $1\leq k\leq r$, choose $P_c=Z_j$ and
$f(\theta)=\cos\theta_k$. Then $e^{-i\tau Z_j\otimes\cos\hat\theta_k}$
is non-Clifford for every $\tau\neq0$. On the $Z_j=+1$ and $Z_j=-1$
eigenspaces, its rotor factors are $e^{-i\tau\cos\hat\theta_k}$ and
$e^{i\tau\cos\hat\theta_k}$, respectively.

\subsection{Qubit-controlled rotor displacements}
\label{subsec:qubit-controlled-rotor-displacements}

The mixed Clifford gate $\Gamma_{c,j}$ uses rotor momentum parity to
control a qubit Pauli operation. We now consider the opposite direction
of control, in which the qubit controls a rotor displacement. The
classification in
Sec.~\ref{sec:clifford-structure-hybrid-qubit-rotor-systems}
already indicates an obstruction: every homomorphism
$
  \FF_2^{2n}\longrightarrow\ZZ^r
$
is trivial. A Clifford transformation therefore cannot contain a
nonzero qubit-controlled rotor-momentum shift.

For one qubit and one rotor, define the conditional displacement
\begin{equation}
  \mathrm{CD}(m,\lambda) = |0\rangle\langle0|\otimes D_R(m,\lambda) +
  |1\rangle\langle1|\otimes D_R(-m,-\lambda), \qquad m\in\ZZ,\quad \lambda\in\TT.
\end{equation}

\begin{proposition}[Clifford criterion for conditional rotor displacements]
The gate $\mathrm{CD}(m,\lambda)$ is Clifford if and only if $m=0$ and $\lambda\in\frac{\pi}{2}\ZZ\subset\TT.$
\end{proposition}

\begin{proof}
Conjugating the rotor phase operator gives
\begin{equation}
  \mathrm{CD}(m,\lambda) (I_Q\otimes Z_R(\phi)) \mathrm{CD}(m,\lambda)^\dagger = e^{-im\phi Z}\otimes Z_R(\phi).
\end{equation}
If $m\neq0$, one can choose $\phi$ such that
$e^{-im\phi Z}$ is not proportional to a Pauli operator. The conjugated
operator is then not proportional to a hybrid Weyl operator. Hence
$m=0$ is necessary.

Let $c_Z$ denote the Pauli label for which $P_{c_Z}=Z$. When $m=0$,
\begin{equation}
  \mathrm{CD}(0,\lambda)
  =
  [I_Q\otimes Z_R(\lambda)]
  C_{c_Z,1}(-2\lambda).
\end{equation}
The first factor is a rotor Weyl operator. By
Corollary~\ref{cor:conditional-phase-clifford-criterion}, the second
factor is Clifford if and only if $-2\lambda\equiv0$ or $\pi \pmod{2\pi}.$ Thus $\lambda\in\frac{\pi}{2}\ZZ\subset\TT$, which gives the stated
condition.
\end{proof}

The following qubit-controlled rotor shift acts nontrivially only on the
$|1\rangle$ component:
\begin{equation}
  \mathrm{CShift}_{q\to R}(m) = |0\rangle\langle0|_q\otimes I_R + |1\rangle\langle1|_q\otimes X_R(m), \qquad
  m\in\ZZ.
\end{equation}

\begin{corollary}[Clifford criterion for qubit-controlled rotor shifts]
The gate $\mathrm{CShift}_{q\to R}(m)$ is Clifford if and only if
$m=0$.
\end{corollary}

\begin{proof}
Conjugating the rotor phase operator gives
\begin{equation}
  \mathrm{CShift}_{q\to R}(m) (I_Q\otimes Z_R(\phi)) \mathrm{CShift}_{q\to R}(m)^\dagger = \left(
  |0\rangle\langle0| + e^{-im\phi}|1\rangle\langle1| \right) \otimes Z_R(\phi).
\end{equation}
If $m\neq0$, choose $\phi$ such that
$e^{-im\phi}\neq\pm1$. The qubit factor is then not proportional to a
Pauli operator, so the conjugated operator is not proportional to a
hybrid Weyl operator. Hence the gate is non-Clifford. For $m=0$, it is
the identity.
\end{proof}

\subsection{Summary of Clifford conditions}

Table~\ref{tab:summary-clifford-conditions} summarizes the Clifford
conditions for the parameterized gate families studied above. The standard
Clifford gates recalled in
Sec.~\ref{subsec:pure-rotor-mixed-clifford-gates} are omitted.

\begin{table}[htbp]
  \centering
  \small
  \renewcommand{\arraystretch}{1.2}
  \begin{tabularx}{\textwidth}{@{}L{0.32\textwidth}Y@{}}
    \toprule
    Gate family & Clifford condition \\
    \midrule
    $G_u$
      & $u(\ell)=e^{i\gamma+i\lambda^{\transpose}\ell+if_R(\ell)}$ for some
        $\gamma\in\TT$, $\lambda\in\TT^r$, and
        $R=R^{\transpose}\in\operatorname{Mat}_r(\TT)$. \\
    $\mathcal V\quad(r=1)$
      & $V_\ell=e^{i[\lambda\ell+\alpha\ell(\ell+1)/2]}
        P_c^{\,\ell}V_0$ for some $V_0\in\Cl_{n,0}$,
        $c\in\FF_2^{2n}$, and $\alpha,\lambda\in\TT$. \\
    $C_{c,j}(\varphi)$
      & $\varphi\equiv0$ or $\pi\pmod{2\pi}$. \\
    $V_f(t)$
      & $t=0$ or $f$ is constant. \\
    $\mathrm{PSR}_{c,f}(\tau)$
      & $\tau=0$, or $f(\theta)=f_0$ and
        $\tau f_0\in\frac{\pi}{4}\ZZ$. \\
    $\mathrm{CD}(m,\lambda)$
      & $m=0$ and $\lambda\in\frac{\pi}{2}\ZZ\subset\TT$. \\
    $\mathrm{CShift}_{q\to R}(m)$
      & $m=0$. \\
    \bottomrule
  \end{tabularx}
  \caption{Summary of the Clifford conditions for the parameterized gate
  families considered in this section.}
  \label{tab:summary-clifford-conditions}
\end{table}

\section{Universal control from Clifford and non-Clifford gates}
\label{sec:universal-control-clifford-non-clifford-gates}
Sec.~\ref{sec:gate-families} shows that the cosine potential and the
conditional phase at angle \(\pi/4\) are non-Clifford. Together with Clifford
operations acting separately on the two registers, these two controls give
universal control in the strong operator topology.

Define the group of Clifford operations acting separately on the two
registers by
\begin{equation}
  \Cl_{n,r}^{\mathrm{loc}} := \left\{ U_Q\otimes U_R: U_Q\in\Cl_{n,0},\ U_R\in\Cl_{0,r} \right\}.
\end{equation}

For $n,r\geq1$, let $c_1$ be the Pauli label satisfying $P_{c_1}=Z_1$,
and let $K_1(t)$ denote the cosine-potential gate in
Eq.~\eqref{eq:cosine-potential-gate} acting on the first rotor. Let
\begin{equation}
  G_{n,r} := \left\langle \Cl_{n,r}^{\mathrm{loc}}, K_1(t):t\in\mathbb R, C_{c_1,1}(\pi/4) \right\rangle .
  \label{eq:universal-control-resource-group}
\end{equation}

By Corollary~\ref{cor:conditional-phase-clifford-criterion},
$C_{c_1,1}(\pi/4)$ is non-Clifford. Since $\cos\theta$ is nonconstant,
Corollary~\ref{cor:real-angle-potential-clifford-criterion} shows that
$K_1(t)$ is non-Clifford for every $t\neq0$.

The fixed conditional-phase gate satisfies
\begin{equation}
  \bigl[C_{c_1,1}(\pi/4)\bigr]^4 = C_{c_1,1}(\pi) = \Gamma_{c_1,1}.
\end{equation}
Since $\Cl_{n,r}^{\mathrm{loc}}$ contains the hybrid Weyl operators,
Corollary~\ref{cor:hybrid-clifford-generators} implies that
$\Cl_{n,r}^{\mathrm{loc}}$ together with
$C_{c_1,1}(\pi)$ generates the full hybrid Clifford group.
Lemma~\ref{lem:conditional-phase-qubit-t} further shows that
$C_{c_1,1}(\pi/4)$ yields a qubit $T$ gate.

\subsection{Strong operator topology and finite momentum windows}
\label{subsec:strong-topology-finite-window}

Throughout this section, $\overline{G}^{\,\mathrm{s}}$ denotes the closure
of $G$ in $\UU(\bH)$ with respect to the strong operator topology.

\begin{definition}[Universal control in the strong operator topology]
\label{def:universal-control-strong-topology}
A subgroup $G\subseteq\UU(\bH)$ provides universal control in the strong
operator topology if $\overline{G}^{\,\mathrm{s}}=\UU(\bH).$ Thus,
for every target unitary $V\in\UU(\bH)$, every finite set of
vectors $\psi_1,\ldots,\psi_s\in\bH$, and every $\epsilon>0$, there exists
$U\in G$ such that
\begin{equation}
  \|(U-V)\psi_j\|<\epsilon,
  \qquad
  j=1,\ldots,s.
\end{equation}
\end{definition}

Lemma~\ref{lem:finite-window-strong-density} reduces strong density on the full
Hilbert space to approximate control on each finite momentum window.

\begin{lemma}[Finite-window criterion for strong density]
\label{lem:finite-window-strong-density}
Let $\bH$ be separable, and let $\{P_L\}$ be an increasing sequence of
finite-rank orthogonal projections satisfying
$
  P_L\longrightarrow I
$ strongly.
Set $\bH_L=P_L\bH$. Suppose that for every $L$, every
$V_L\in\UU(\bH_L)$, and every $\eta>0$, there exists
$W\in\overline{G}^{\,\mathrm{s}}$ such that $W$ preserves $\bH_L$ and
\begin{equation}
  \left\|W|_{\bH_L}-V_L\right\|<\eta.
\end{equation}
Then $G$ is strongly dense in $\UU(\bH)$.
\end{lemma}

Appendix~\ref{app:proofs-universal-control} proves
Lemma~\ref{lem:finite-window-strong-density} by interpolating a target unitary
on a sufficiently large finite momentum window and controlling the projected
tails.
Fig.~\ref{fig:universal-control-roadmap} summarizes the proof strategy.

\begin{figure}[!htbp]
  \centering

  \resizebox{0.86\linewidth}{!}{%
  \begin{tikzpicture}[
    font=\small,
    panel/.style={
      draw,
      rounded corners=2pt,
      fill=black!2,
      line width=0.45pt,
      inner sep=0pt
    },
    arr/.style={
      -{Latex[length=1.8mm]},
      line width=0.45pt
    }
  ]

\node[
  panel,
  minimum width=68mm,
  minimum height=12mm
] at (0,6.30) {};

\node at (0,6.48) {
  $\mathrm{Cl}^{\mathrm{loc}}_{n,r},
   \qquad
   K_1(t),
   \qquad
   C_{\pi/4}$
};

\node at (0,6.08) {
  generating set $G_{n,r}$
};

\node[
  panel,
  minimum width=60mm,
  minimum height=40mm
] at (-5.50,3.20) {};

\node[font=\small\bfseries]
  at (-5.50,4.78)
  {rotor-local control};

\node[font=\scriptsize]
  at (-5.50,4.35)
  {$H_{\mathrm b}
    =
    \hat{\ell}^{\,2}
    +
    \xi\hat{\ell},
    \qquad
    \xi\in\mathbb{R}\setminus\mathbb{Z}$};

\node[font=\scriptsize]
  at (-5.50,3.92)
  {$M_{\cos}=\cos\hat{\theta},
    \qquad
    K(t)=e^{-itM_{\cos}}$};

\node[font=\scriptsize]
  at (-5.50,3.46)
  {$M_{\cos}$ couples adjacent momentum states};

\node[font=\small]
  at (-7.15,2.95)
  {$\lvert\ell-1\rangle$};

\node[font=\small]
  at (-5.50,2.95)
  {$\lvert\ell\rangle$};

\node[font=\small]
  at (-3.85,2.95)
  {$\lvert\ell+1\rangle$};

\draw[
  <->,
  line width=0.45pt,
  >=Latex
]
  (-6.58,2.95)
  --
  (-6.02,2.95);

\draw[
  <->,
  line width=0.45pt,
  >=Latex
]
  (-4.98,2.95)
  --
  (-4.42,2.95);

\node[font=\scriptsize]
  at (-5.50,2.42)
  {nonresonant connected chain};

\node[font=\small]
  at (-5.50,1.72)
  {$\overline{G_R}^{\,s}
    =
    \mathrm{U}\!\bigl(L^2(\mathbb{T})\bigr)$};

\node[
  panel,
  minimum width=46mm,
  minimum height=35mm
] at (0,3.25) {};

\node[font=\small\bfseries]
  at (0,4.45)
  {qubit-local control};

\node[font=\scriptsize]
  at (0,3.92)
  {the commutator of $C_{\pi/4}$ and $X_R(1)$};

\node[font=\scriptsize]
  at (0,3.53)
  {yields $T^\dagger\otimes I_R$};

\node[font=\scriptsize]
  at (0,2.94)
  {single-qubit Clifford gates and $T$};

\node[font=\scriptsize]
  at (0,2.57)
  {are norm dense in $\mathrm{U}(\mathbb{C}^2)$};

\node[font=\scriptsize]
  at (0,2.05)
  {$\overline{
      \langle\mathrm{Cl}_{1,0},T\rangle
    }^{\lVert\cdot\rVert}
    =
    \mathrm{U}(\mathbb{C}^2)$};

\node[
  panel,
  minimum width=40mm,
  minimum height=27mm
] at (5.00,3.42) {};

\node[font=\small\bfseries]
  at (5.00,4.25)
  {mixed Clifford gate $C_\pi$};

\node[font=\small]
  at (5.00,3.63)
  {$C_\pi=(C_{\pi/4})^4$};

\node[font=\scriptsize,align=center]
  at (5.00,2.91)
  {entangling on every\\
   finite window};

\node[
  panel,
  minimum width=78mm,
  minimum height=14mm
] at (0,-0.05) {};

\node at (0,0.12) {
  $\overline{G_{1,1}}^{\,s}
   =
   \mathrm{U}\!\left(
     \mathbb{C}^2\otimes L^2(\mathbb{T})
   \right)$
};

\node at (0,-0.32) {
  universal control of one qubit and one rotor
};

\node[
  panel,
  minimum width=86mm,
  minimum height=15mm
] at (0,-4.25) {};

\node at (0,-4.05) {
  $\overline{G_{n,r}}^{\,s}
   =
   \mathrm{U}(\mathcal{H}_{n,r}),
   \qquad
   n,r\geq1$
};

\node at (0,-4.50) {
  universal control of $n$-qubit, $r$-rotor registers
};

\draw[arr]
  (-2.10,5.70)
  --
  (-4.30,5.20);

\draw[arr]
  (0,5.70)
  --
  (0,5.02);

\draw[arr]
  (2.10,5.70)
  --
  (4.35,4.80);

\draw[arr]
  (-4.65,1.20)
  --
  (-2.70,0.68);

\draw[arr]
  (0,1.48)
  --
  (0,0.68);

\draw[arr]
  (5.00,2.05)
  --
  (2.70,0.68);

\draw[arr]
  (0,-0.76)
  --
  (0,-3.48);

\node[
  font=\scriptsize,
  align=right,
  anchor=east
] at (4.7,-1.71) {
  permutations distribute local control\\
  CPHS couples distinct rotor modes
};

\node[
  font=\scriptsize,
  align=left,
  anchor=west
] at (0.12,-2.74) {
  finite-window norm density\\
  $\Longrightarrow$ full-space strong density
};

  \end{tikzpicture}%
  }

  \caption{Proof structure for universal control. A biased rotor drift and
  cosine control yield universal control of a single rotor, while a commutator
  of $C_{\pi/4}$ with $X_R(1)$ produces $T^\dagger\otimes I_R$ and hence
  single-qubit control. Combining these local controls with the hybrid
  entangler $C_\pi=(C_{\pi/4})^4$ gives universal control of one qubit and one
  rotor. Permutations and CPHS couplings extend the construction to general
  hybrid registers, while finite-window norm density implies full-space
  strong density.}

  \label{fig:universal-control-roadmap}
\end{figure}
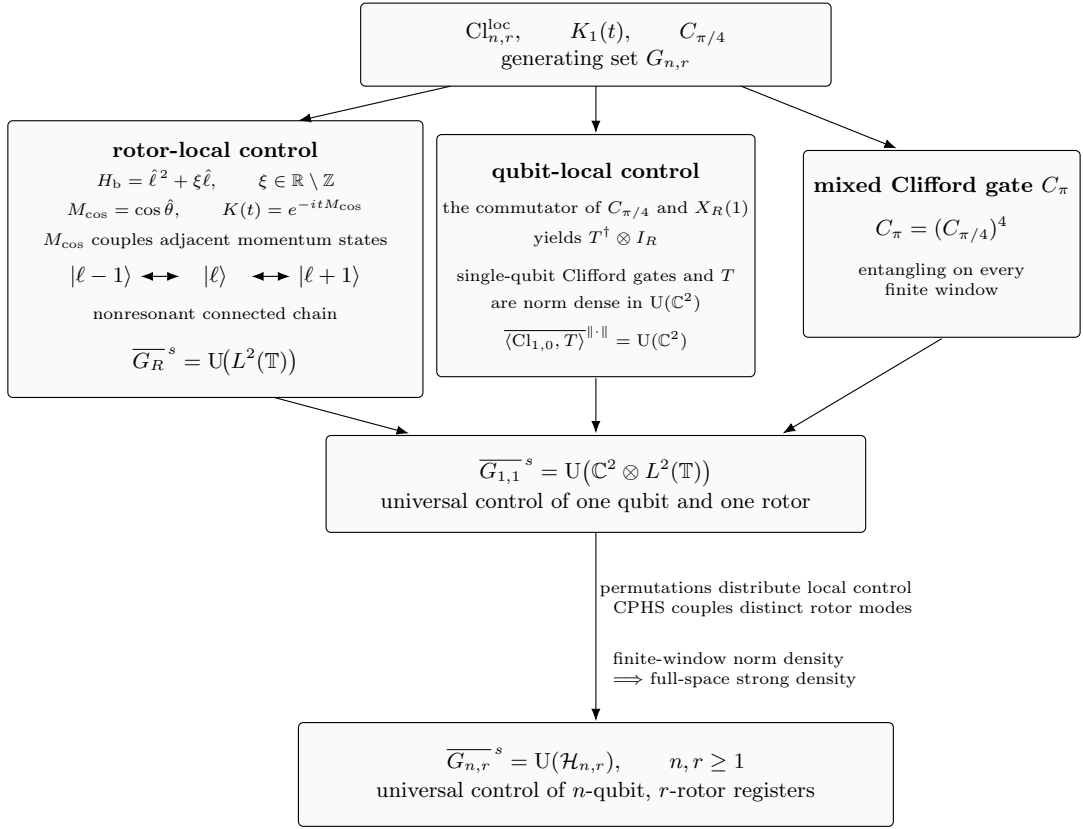

\FloatBarrier
\subsection{Single-rotor universal control}

Let $K(t)$ be the cosine-potential gate defined in
Eq.~\eqref{eq:cosine-potential-gate}, and set
\begin{equation}
  H_{\mathrm b} = \hat\ell^2+\xi\hat\ell, \quad \xi\in\mathbb R\setminus\ZZ.
  \label{eq:biased-rotor-control-hamiltonian}
\end{equation}
The drift evolution is a pure-rotor Clifford gate because
\begin{equation}
  e^{-isH_{\mathrm b}}
  =
  \mathrm{QUAD}(-2s)Z_R((1-\xi)s).
\end{equation}
In the momentum basis,
$
  H_{\mathrm b}|\ell\rangle
  =
  \varepsilon_\ell|\ell\rangle,
  \varepsilon_\ell=\ell^2+\xi\ell,
$
and the nearest-neighbor transition gaps are
\begin{equation}
  \varepsilon_{\ell+1}-\varepsilon_\ell
  =
  2\ell+1+\xi.
\end{equation}
Moreover,
$
  \langle k|\cos\hat\theta|\ell\rangle
  =
  \frac{1}{2}
  \left(
    \delta_{k,\ell+1}
    +
    \delta_{k,\ell-1}
  \right),
$
so the control couples only neighboring momentum levels. Since
$\xi\notin\ZZ$, the spectrum is nondegenerate and the corresponding
transition frequencies
$
  |\varepsilon_{\ell+1}-\varepsilon_\ell|
$
are pairwise distinct. The nearest-neighbor transitions connect the full
momentum lattice and therefore form a nonresonant connectedness chain.
Theorem~2.11 of
Ref.~\cite{boscain2012spectral} consequently applies. A related earlier
result on approximate state controllability under a stronger spectral
assumption is given in Ref.~\cite{chambrion2009discretespectrum}.

\begin{theorem}[Single-rotor universal control]
\label{thm:single-rotor-universal-control}
The group $
  G_R
  :=
  \left\langle
    e^{-isH_{\mathrm b}},
    K(t):
    s,t\in\mathbb R
  \right\rangle
$
is strongly dense in $\UU(L^2(\TT))$.
\end{theorem}

\begin{proof}
Appendix~\ref{app:proofs-universal-control} verifies the hypotheses of
Theorem~2.11 in Ref.~\cite{boscain2012spectral}. For a constant control
amplitude $u$ applied for time $\tau$, the Lie--Trotter formula gives
\begin{equation}
  e^{-i\tau(H_{\mathrm b}+u\cos\hat\theta)} = \operatorname{s\!-\!lim}\limits_{N\to\infty} \left[ e^{-i\tau H_{\mathrm b}/N} K\!\left(\frac{\tau u}{N}\right) \right]^N.
\end{equation}
Hence every evolution generated by a piecewise-constant control belongs
to $\overline{G_R}^{\,\mathrm{s}}$. The cited theorem states that these
evolutions can approximate any target unitary simultaneously on any finite
set of states. By
Definition~\ref{def:universal-control-strong-topology}, this is precisely
the strong density of $G_R$ in $\UU(L^2(\TT))$.
\end{proof}

\subsection{One qubit and one rotor}

In this subsection, write $C_\varphi:=C_{c_1,1}(\varphi).$

\begin{lemma}[A qubit \texorpdfstring{$T$}{T} gate from the conditional phase]
\label{lem:conditional-phase-qubit-t}
The conditional-phase gate and the unit rotor momentum shift satisfy
\begin{equation}
  (I_Q\otimes X_R(1)) C_\varphi (I_Q\otimes X_R(1)^\dagger) C_\varphi^\dagger = e^{-i\varphi|1\rangle\langle1|}
  \otimes I_R.
  \label{eq:conditional-phase-qubit-t}
\end{equation}
For $\varphi=\pi/4$, the right-hand side is
$T^\dagger\otimes I_R$, where $T=\operatorname{diag}(1,e^{i\pi/4}).$
\end{lemma}

\begin{proof}
Since
$
  X_R(1)\hat\ell X_R(1)^\dagger
  =
  \hat\ell-I_R,
$
conjugation by the momentum shift gives
$
  (I_Q\otimes X_R(1))
  C_\varphi
  (I_Q\otimes X_R(1)^\dagger)
  =
  \left(
    e^{-i\varphi|1\rangle\langle1|}
    \otimes I_R
  \right)
  C_\varphi.
$
Multiplication by $C_\varphi^\dagger$ proves
Eq.~\eqref{eq:conditional-phase-qubit-t}.
\end{proof}

The single-qubit Clifford group together with $T$ is dense, up to a global
phase, in the single-qubit unitary group
\cite{boykin1999faulttolerant,nielsen2000computation}.
Since scalar phases are contained in $\Cl_{1,0}$, the gates available here
generate a norm-dense subgroup of $\UU(\CC^2)$.

\begin{theorem}[One-qubit--one-rotor universal control]
The group $G_{1,1}$ defined in
Eq.~\eqref{eq:universal-control-resource-group} is strongly dense in $\UU\bigl(\CC^2\otimes L^2(\TT)\bigr).$
\end{theorem}

\begin{proof}
The group $G_{1,1}$ contains the single-rotor gates used in
Theorem~\ref{thm:single-rotor-universal-control}. Hence its strong closure
contains $I_Q\otimes U_R$ for every $U_R\in\UU(L^2(\TT)).$
By Lemma~\ref{lem:conditional-phase-qubit-t} and the single-qubit Clifford
gates, the same closure contains $U_Q\otimes I_R$ for every $U_Q\in\UU(\CC^2).$

For $L\geq1$, let
\begin{equation}
  \mathcal W_L = \operatorname{span} \left\{ |\ell\rangle: -L\leq\ell\leq L \right\}.
\end{equation}
Every unitary on $\mathcal W_L$, extended by the identity on its orthogonal
complement, is a rotor-local unitary and therefore belongs to
$\overline{G_{1,1}}^{\,\mathrm{s}}$. Such an extension preserves
$\mathcal W_L$.

Moreover,
\begin{equation}
  C_\pi = \bigl[C_{\pi/4}\bigr]^4 = |0\rangle\langle0|\otimes I_R + |1\rangle\langle1| \otimes(-1)^{\hat\ell}
\end{equation}
also preserves every momentum window. It is entangling on
$\CC^2\otimes \mathcal W_L$. Indeed, choose an even momentum state
$|e\rangle$ and an odd momentum state $|o\rangle$ in the window. Then
\begin{equation}
  C_\pi \left[ |+\rangle \otimes \frac{|e\rangle+|o\rangle}{\sqrt2} \right] = \frac{1}{2} \left[
  |0\rangle(|e\rangle+|o\rangle) + |1\rangle(|e\rangle-|o\rangle) \right],
  \label{eq:Cpi-entangling-action}
\end{equation}
which is an entangled state.

Arbitrary local unitaries together with an entangling bipartite gate generate
a norm-dense subgroup of $\UU\bigl(\CC^2\otimes \mathcal W_L\bigr)$
\cite[Corollary~2]{harrow2009universality}.
Thus every unitary on
$\CC^2\otimes \mathcal W_L$ can be approximated by elements of
$\overline{G_{1,1}}^{\,\mathrm{s}}$ that preserve this finite window.

The projections onto
$\CC^2\otimes \mathcal W_L$ increase strongly to the identity. Therefore
Lemma~\ref{lem:finite-window-strong-density} implies that
$G_{1,1}$ is strongly dense in
$\UU(\CC^2\otimes L^2(\TT))$.
\end{proof}

\subsection{General hybrid universal control}

\begin{theorem}[General hybrid universal control]
\label{thm:general-hybrid-universal-control}
For every $n,r\geq1$, $\overline{G_{n,r}}^{\,\mathrm{s}} = \UU(\bH_{n,r}),$
where $G_{n,r}$ is defined in Eq.~\eqref{eq:universal-control-resource-group}.
\end{theorem}

\begin{proof}
We first obtain arbitrary rotor operations on finite momentum windows.
Rotor permutations conjugate $K_1(t)$ to the corresponding cosine-potential
gate $K_j(t)$ on any rotor mode $j$. By
Theorem~\ref{thm:single-rotor-universal-control}, the strong closure therefore
contains every single-rotor unitary acting on any chosen mode.

For $L\geq1$, define the $r$-rotor momentum window
\begin{equation}
  \mathcal W_{L,r} = \operatorname{span} \left\{ |\ell_1,\ldots,\ell_r\rangle: -L\leq\ell_j\leq L \text{ for all }j
  \right\}.
\end{equation}
Every single-mode unitary on
$\operatorname{span}\{|\ell_j\rangle:-L\leq\ell_j\leq L\}$ can be
extended by the identity outside that window. Hence the restriction of
$\overline{G_{n,r}}^{\,\mathrm{s}}$ to $\mathcal W_{L,r}$ contains arbitrary
unitaries on each individual rotor mode.

The pure-rotor Clifford group contains the gates
$\mathrm{CPHS}_{ij}(\alpha)$. For
$\alpha\not\equiv0\pmod{2\pi}$, such a gate is entangling on the finite
window. Indeed, on the subspace spanned by $|0\rangle$ and $|1\rangle$
of each of the two modes, it maps
$
  \frac{|0\rangle+|1\rangle}{\sqrt2}
  \otimes
  \frac{|0\rangle+|1\rangle}{\sqrt2}
$
to a state whose coefficient matrix is
$
  \begin{pmatrix}
    1 & 1\\
    1 & e^{i\alpha}
  \end{pmatrix}.
$
Its determinant is $e^{i\alpha}-1\neq0$, so the resulting state is
entangled. Choose CPHS couplings between neighboring rotor modes. These
couplings form a connected interaction graph. Applying the bipartite
universality criterion of
Ref.~\cite[Corollary~2]{harrow2009universality} iteratively along a spanning
tree of this graph, together with arbitrary single-mode unitaries, gives a
norm-dense subgroup of $\UU(\mathcal W_{L,r})$. All these gates preserve the
finite momentum window.

We next consider the qubit register. By
Lemma~\ref{lem:conditional-phase-qubit-t}, the available gates produce a
$T$ gate on the first qubit. Qubit Clifford operations include qubit
permutations, so a $T$ gate is available on every qubit. The $n$-qubit
Clifford group together with these $T$ gates generates a norm-dense
subgroup of $\UU\bigl((\CC^2)^{\otimes n}\bigr).$

It remains to couple the two registers. Since
$C_{c_1,1}(\pi) = \bigl[C_{c_1,1}(\pi/4)\bigr]^4,$ the group $G_{n,r}$
contains the mixed Clifford gate $C_{c_1,1}(\pi)$.
This gate preserves every finite momentum window and is entangling across
the partition $(\CC^2)^{\otimes n} \otimes \mathcal W_{L,r}.$

For example, fix all qubits except the first and all rotors except the first,
and choose one even and one odd momentum state of the first rotor. The
restriction then contains the entangling action exhibited in
Eq.~\eqref{eq:Cpi-entangling-action}.

We have therefore obtained arbitrary unitaries on the qubit register,
arbitrary unitaries on the finite rotor window, and an entangling gate
between the two registers. The finite-dimensional bipartite universality
criterion consequently gives a norm-dense subgroup of
$\UU\left((\CC^2)^{\otimes n}\otimes \mathcal W_{L,r} \right)$
for every $L$. These operations preserve the finite-dimensional subspace
$(\CC^2)^{\otimes n}\otimes\mathcal W_{L,r}$.

The projections onto
$(\CC^2)^{\otimes n}\otimes \mathcal W_{L,r}$ increase strongly to the
identity. Lemma~\ref{lem:finite-window-strong-density} therefore implies
$\overline{G_{n,r}}^{\,\mathrm{s}} = \UU(\bH_{n,r}).$
\end{proof}

\begin{remark}[Boundary cases]
If $r=0$ and $n\geq1$, the $n$-qubit Clifford group together with a
single-qubit $T$ gate generates a norm-dense subgroup of
$\UU((\CC^2)^{\otimes n})$. If $n=0$ and $r\geq1$, the argument above
restricted to the rotor register gives
$
  \overline{
    \left\langle
      \Cl_{0,r},
      K_1(t):t\in\mathbb R
    \right\rangle
  }^{\,\mathrm{s}}
  =
  \UU(L^2(\TT^r)).
$
The case $n=r=0$ is one-dimensional.
\end{remark}

\paragraph{Logical operations in the applications.}
Secs.~\ref{sec:gauge-model-gate-realization}--\ref{sec:qft-rotor-registers}
also use CShift and $\mathcal O_U$ as logical operations.
Theorem~\ref{thm:general-hybrid-universal-control} places
these operations in the strong closure of $G_{n,r}$.

\section{Gate realization of a compact \texorpdfstring{$U(1)$}{U(1)} gauge--matter model}
\label{sec:gauge-model-gate-realization}

In compact $U(1)$ gauge--matter theory, each oriented link carries a compact
gauge angle and an integer-valued electric flux, and each matter site carries
a binary occupation. Links are therefore represented by rotors and matter
sites by qubits. For a lattice with $N_s$
matter sites and $N_\ell$ oriented links, the corresponding hybrid register is
\begin{equation}
  \mathcal H_{\mathrm{GM}} = \bigl(\mathbb C^2\bigr)^{\otimes N_s} \otimes L^2(\mathbb T)^{\otimes N_\ell}.
\end{equation}
Gauss's law couples the two registers at the level of the physical state
space. The four Hamiltonian terms fall into the gate families introduced in
Sec.~\ref{sec:gate-families}: the electric term is generated by rotor momentum quadratics, the
magnetic term by a periodic angle potential, and the mass term by qubit-local
operations. The gauge-covariant hopping term changes matter occupation
together with an integer shift of the link flux, placing it in the
qubit-to-rotor control direction studied in
Sec.~\ref{subsec:qubit-controlled-rotor-displacements}.

We use the standard Hamiltonian formulation of compact \(U(1)\) lattice
gauge theory introduced by Kogut and Susskind
\cite{kogut1975wilson}. Its separation into matter degrees of freedom and
integer-flux gauge links is also the starting point of recent hybrid quantum
simulation proposals, including oscillator--qubit, qubit--qumode, and
superconducting phase--charge implementations
\cite{crane2024fermions,ale2026electrodynamics,
alcainecuervo2026compact}. Broader reviews of quantum simulation methods for
lattice gauge theories can be found in
Refs.~\cite{zohar2016lattice,banuls2020lattice}.

\subsection{Rotor links, matter qubits, and Gauss's law}

Let $e=(a,b)$ denote a link oriented from site $a$ to site $b$.  Its gauge
field is a rotor with
\begin{equation}
  U_e=e^{i\hat\theta_e}=X_R(1), \qquad E_e=\hat\ell_e, \qquad [E_e,U_{e'}]=\delta_{ee'}U_{e'}.
\end{equation}
Thus $U_e|\ell_e\rangle=|\ell_e+1\rangle$.  For the oriented plaquettes, let
$B=(B_{pe})$ be the boundary matrix, with $B_{pe}\in\{0,\pm1\}$, and set
\begin{equation}
  U_p=\prod_eU_e^{B_{pe}}=e^{i\Phi_p},
  \qquad
  \Phi_p=\sum_eB_{pe}\hat\theta_e.
\end{equation}

One qubit at site $n$ stores the occupation
\begin{equation}
  \hat n_n=\frac{I+Z_n}{2}.
  \label{eq:gauge-matter-occupation}
\end{equation}
With this Pauli convention, $|0\rangle$ is occupied and $|1\rangle$ is empty.
On a two-dimensional square lattice, let 
$
\eta_n=(-1)^{n_x+n_y}, b_n=\frac{1-\eta_n}{2}
$
and
\begin{equation}
  \rho_n=\hat n_n-b_n+q_n^{\mathrm{ext}},
  \qquad
  q_n^{\mathrm{ext}}\in\ZZ.
\end{equation}
Here $q_n^{\mathrm{ext}}$ is a fixed external charge.

Let $D\in\ZZ^{N_s\times N_\ell}$ be the oriented site--link incidence
matrix, with
\begin{equation}
  D_{ne}
  =
  \begin{cases}
    +1, & \text{if site $n$ is the tail of $e$},\\
    -1, & \text{if site $n$ is the head of $e$},\\
    0,  & \text{otherwise}.
  \end{cases}
\end{equation}
The Gauss-law generators are
\begin{equation}
  G_n=\sum_eD_{ne}E_e-\rho_n,
  \qquad
  \bH_{\mathrm{phys}}=\bigcap_n\ker G_n.
\end{equation}

Because each oriented link has one tail and one head,
$\mathbf 1^{\transpose}D=0$. Every physical state therefore satisfies the
global Gauss constraint
\begin{equation}
  \sum_n\left(\hat n_n-b_n+q_n^{\mathrm{ext}}\right)|\psi\rangle=0,
  \qquad
  |\psi\rangle\in\bH_{\mathrm{phys}}.
  \label{eq:gauge-global-gauss-constraint}
\end{equation}

The Gauss-law constraint couples the flux divergence on the rotor links
incident on site $n$ to the matter-qubit charge $\rho_n$. The physical
subspace consists of matter occupations and link-flux configurations
satisfying $G_n=0$ at every site.

Let $g>0$ denote the gauge coupling, $a_{\mathrm{lat}}>0$ the lattice
spacing, $m_0\in\mathbb R$ the staggered mass parameter, and
$\kappa\in\mathbb R$ the hopping amplitude. The gauge--matter Hamiltonian is
\begin{equation}
  H=H_E+H_B+H_M+H_K,
\end{equation}
where
\begin{equation}
  \begin{array}{r@{}l@{\qquad}r@{}l}
    H_E&=\frac{g^2}{2}\sum_eE_e^2,
    &
    H_B&=\frac{1}{g^2a_{\mathrm{lat}}^2}\sum_p(1-\cos\Phi_p),
  \end{array}
\end{equation}
\begin{equation}
  \begin{array}{r@{}l@{\qquad}r@{}l}
    H_M&=m_0\sum_n\eta_n\hat n_n,
    &
    H_K&=\sum_{e=(a,b)}H_{K,ab}.
  \end{array}
\end{equation}
Let $\Psi_n$ and $\Psi_n^\dagger$ denote the fermionic annihilation and
creation operators at site $n$. For a link $e=(a,b)$ oriented from site $a$
to site $b$, let $\zeta_{ab}=e^{i\delta_{ab}}$ denote the fixed hopping
phase. The gauge-covariant hopping term is
\begin{equation}
  H_{K,ab} =\frac{\kappa}{2}\left( \zeta_{ab}\Psi_a^\dagger U_{ab}\Psi_b
  +\zeta_{ab}^*\Psi_b^\dagger U_{ab}^\dagger\Psi_a \right).
\end{equation}
The direction-dependent phases in a Kogut--Susskind discretization are
particular choices of $\zeta_{ab}$
\cite{susskind1977lattice,crippa2026confinement}. 

Choose a snake ordering $s(n)$ of the matter sites.  The Jordan--Wigner
transformation~\cite{jordan1928paulische} introduces
\begin{equation}
  \Pi_m^{\mathrm{JW}}=(-1)^{\hat n_m}=-Z_m,
  \qquad
  S_{ab}=\prod_{m:\,\min(s(a),s(b))<s(m)<\max(s(a),s(b))}
  \Pi_m^{\mathrm{JW}},
\end{equation}
and gives
\begin{equation}
  H_{K,ab} =\frac{\kappa}{2}\left( \zeta_{ab}\sigma_a^+S_{ab}U_{ab}\sigma_b^-
  +\zeta_{ab}^*\sigma_b^+S_{ab}U_{ab}^\dagger\sigma_a^- \right).
  \label{eq:gauge-qubit-rotor-hopping}
\end{equation}
For adjacent sites in the chosen ordering, $S_{ab}=I$, so
Eq.~\eqref{eq:gauge-qubit-rotor-hopping} reduces to a two-qubit exchange
coupled to the link momentum shift $U_{ab}$. Expanding the ladder operators
gives the form used below:
\begin{equation}
  H_{K,ab}=\frac{\kappa}{4}\Big[{} \cos(\hat\theta_{ab}+\delta_{ab}) (X_aS_{ab}X_b+Y_aS_{ab}Y_b)
  +\sin(\hat\theta_{ab}+\delta_{ab}) (X_aS_{ab}Y_b-Y_aS_{ab}X_b) \Big].
  \label{eq:gauge-hopping-pauli-form}
\end{equation}

Appendix~\ref{app:gauge-invariance-pauli-hopping} verifies that each of
$H_E$, $H_B$, $H_M$, and $H_K$ commutes with every $G_n$ and derives the
Pauli decomposition in Eq.~\eqref{eq:gauge-hopping-pauli-form}. For each hopping term,
the changes in matter charge and electric flux cancel in the Gauss-law
generators at both endpoints.

\subsection{Gate realization of the Hamiltonian}

\paragraph{Electric term.}
The electric evolution on one link is a pure-rotor Clifford gate:
\begin{equation}
  e^{-i\Delta t\,g^2\hat\ell_e^2/2}
  =
  Z_{R,e}(g^2\Delta t/2)\,
  \mathrm{QUAD}_e(-g^2\Delta t).
\end{equation}
The Weyl phase $Z_{R,e}(g^2\Delta t/2)$ compensates for the
$\ell_e(\ell_e+1)/2$ convention in the definition of QUAD.

\paragraph{Magnetic term.}
For a plaquette $p$, the boundary vector
$b_p=(B_{pe})_e$ is a primitive integer vector. Choose
$A_p\in\GL(N_\ell,\ZZ)$ such that
$
  A_pb_p=e_1.
$
Here $K_1(t)$ denotes the cosine-potential gate in
Eq.~\eqref{eq:cosine-potential-gate} acting on the first rotor mode.
The plaquette evolution is
\begin{equation}
  e^{-i\Delta t(1-\cos\Phi_p)/(g^2a_{\mathrm{lat}}^2)} = e^{-i\Delta t/(g^2a_{\mathrm{lat}}^2)} U_{A_p}^\dagger
  K_1\!\left(-\frac{\Delta t}{g^2a_{\mathrm{lat}}^2}\right) U_{A_p}.
\end{equation}
Indeed, $U_{A_p}^\dagger X_R(e_1)U_{A_p}=X_R(b_p)$, so conjugation maps the
cosine acting on the first rotor mode to $\cos\Phi_p$. This realizes the
plaquette evolution using a rotor automorphism and the cosine-potential gate.
The plaquette boundary vector satisfies $Db_p=0$, the lattice identity that
the boundary of a boundary vanishes, so the evolution preserves every
Gauss-law constraint.

\paragraph{Mass term.}
Using $\hat n_n=(I+Z_n)/2$, the mass evolution is
\begin{equation}
  e^{-i\Delta tH_M} = \exp\!\left( -\frac{i\Delta t\,m_0}{2}\sum_n\eta_n \right) \prod_n
  e^{-i\Delta t\,m_0\eta_nZ_n/2}.
\end{equation}
Apart from the displayed global phase, this term requires only
qubit-local rotations.

\paragraph{Hopping term: controlled-shift realization.}
On the full rotor Hilbert space, define
\begin{equation}
  H_{ab}^{(q)}(\delta_{ab}) = \frac{\kappa}{2} \left( e^{i\delta_{ab}}\sigma_a^+S_{ab}\sigma_b^- +
  e^{-i\delta_{ab}}\sigma_b^+S_{ab}\sigma_a^- \right).
  \label{eq:gauge-qubit-exchange-hamiltonian}
\end{equation}
Consider either of the following controlled shifts:
\begin{equation}
  V_{ab}^{(a)} = \mathrm{CShift}_{a\to e}(-1), \qquad V_{ab}^{(b)} = \mathrm{CShift}_{b\to e}(1).
\end{equation}
Let $V_{ab}$ denote either $V_{ab}^{(a)}$ or $V_{ab}^{(b)}$. Then
$
  H_{K,ab}
  =
  V_{ab}
  H_{ab}^{(q)}(\delta_{ab})
  V_{ab}^\dagger,
$
and therefore
\begin{equation}
  e^{-i\tau H_{K,ab}} = V_{ab} e^{-i\tau H_{ab}^{(q)}(\delta_{ab})} V_{ab}^\dagger.
  \label{eq:gauge-full-space-hopping-conjugation}
\end{equation}

For $\delta_{ab}=0$, define
$
  P_1=X_aS_{ab}X_b,
  \qquad
  P_2=Y_aS_{ab}Y_b.
$
These Pauli strings commute, and
$
  H_{ab}^{(q)}(0)
  =
  \frac{\kappa}{4}(P_1+P_2).
$
Moreover, with
$
  R_a(\delta)=e^{i\delta Z_a/2},
$
we have
$
  H_{ab}^{(q)}(\delta)
  =
  R_a(\delta)
  H_{ab}^{(q)}(0)
  R_a(\delta)^\dagger.
$
Consequently,
\begin{equation}
  e^{-i\tau H_{ab}^{(q)}(\delta_{ab})} = R_a(\delta_{ab}) e^{-i\tau\kappa P_1/4} e^{-i\tau\kappa P_2/4}
  R_a(\delta_{ab})^\dagger.
  \label{eq:gauge-qubit-exchange-decomposition}
\end{equation}
Together, Eqs.~\eqref{eq:gauge-full-space-hopping-conjugation}
and~\eqref{eq:gauge-qubit-exchange-decomposition} implement the
gauge-invariant hopping evolution.
Appendix~\ref{app:gauge-exact-full-space-hopping-realization} derives the
controlled-shift conjugation identities and the commuting-Pauli decomposition
underlying this realization.

\paragraph{Hopping term: PSR realization.}
Eq.~\eqref{eq:gauge-hopping-pauli-form} writes $H_{K,ab}$ as a sum
of four Pauli-string--rotor angle-potential terms. The exponential of
each term is a PSR gate of the form defined in
Eq.~\eqref{eq:pauli-string-rotor-phase-gate}. The sine potentials can be
obtained from the cosine potential by a rotor phase translation.

The controlled-shift realization preserves the Gauss-law subspace.
Finite-order product formulas built from the four PSR factors in
Eq.~\eqref{eq:gauge-hopping-pauli-form} do not generally share this
invariance~\cite{suzuki1990fractal}.

\subsection{Finite-flux truncation}

Numerical calculations replace each link Hilbert space by
\begin{equation}
  \bH_{e,L} = \operatorname{span}\{|\ell\rangle:-L\leq\ell\leq L\}, \qquad \Pi_L =
  \sum_{\ell=-L}^{L}|\ell\rangle\langle\ell|.
  \label{eq:gauge-flux-projector}
\end{equation}
The projected shift
\[
  U_{e,L}
  :=
  \Pi_LU_e\Pi_L
\]
acts as the ordinary link shift on every nonboundary momentum state but
annihilates the upper-boundary state:
\[
  U_{e,L}|\ell\rangle
  =
  \begin{cases}
    |\ell+1\rangle, & -L\leq\ell<L,\\
    0, & \ell=L.
  \end{cases}
\]
Consequently,
\begin{equation}
  U_{e,L}^\dagger U_{e,L} = I_{e,L}-|L\rangle\langle L|, \qquad U_{e,L}U_{e,L}^\dagger =
  I_{e,L}-|-L\rangle\langle-L|.
  \label{eq:gauge-projected-shift-defect}
\end{equation}
Thus $U_{e,L}$ is a nonunitary hard-wall truncation in which transitions
leaving the selected electric-flux window are omitted.

The discrepancy from the full rotor shift is supported entirely at the
cutoff boundary. For a normalized state
$|\psi\rangle\in\bH_{e,L}$,
\begin{equation}
  \begin{aligned}
    \left\|
      (U_e-U_{e,L})|\psi\rangle
    \right\|^2
    &=
    |\langle L|\psi\rangle|^2,
    \\
    \left\|
      (U_e^\dagger-U_{e,L}^\dagger)|\psi\rangle
    \right\|^2
    &=
    |\langle-L|\psi\rangle|^2.
  \end{aligned}
  \label{eq:gauge-projected-shift-boundary-weight}
\end{equation}
Eq.~\eqref{eq:gauge-projected-shift-boundary-weight} identifies the
occupations of $|L\rangle$ and $|-L\rangle$ with the squared norm errors of
the projected forward and backward shifts, respectively.

Adding the boundary transition $|L\rangle\mapsto|-L\rangle$ restores
unitarity. Define
\begin{equation}
  \widetilde U_{e,L} = \sum_{\ell=-L}^{L-1} |\ell+1\rangle\langle\ell| + |-L\rangle\langle L|, \qquad d=2L+1.
  \label{eq:gauge-cyclic-completion}
\end{equation}
The operator $\widetilde U_{e,L}$ cyclically permutes the
$d=2L+1$ momentum states and is therefore unitary. For
\[
  E_{e,L}
  =
  \Pi_LE_e\Pi_L,
\]
direct action on the momentum basis gives
\begin{equation}
  [E_{e,L},\widetilde U_{e,L}] = \widetilde U_{e,L} - d|-L\rangle\langle L|.
  \label{eq:gauge-cyclic-commutator-defect}
\end{equation}
The operators $U_{e,L}$ and $\widetilde U_{e,L}$ agree on every
transition with $\ell<L$. Their only difference is the added wraparound
transition $|L\rangle\mapsto|-L\rangle$, which produces the rank-one
correction in
Eq.~\eqref{eq:gauge-cyclic-commutator-defect}.

The projected shift preserves the ordinary integer-flux translation away from
the boundary and is nonunitary; the cyclic completion is unitary. In the full
hopping term, a unit change
of the endpoint matter charges is accompanied by a unit change of the link
flux. The cyclic wraparound changes the link flux by $-2L$ while the endpoint
matter charges still change by one unit, so it violates the original
integer-valued Gauss's law at the cutoff boundary.

In the cutoff calculation below, we use the hard-wall projected
Hamiltonian: hopping and plaquette transitions whose final flux lies
outside the window are omitted. The resulting Hamiltonian remains
Hermitian and preserves the finite Gauss-law basis. Its agreement with the
full rotor model is controlled by the weight on $|E_e|=L$, as quantified by
Eq.~\eqref{eq:gauge-projected-shift-boundary-weight}.
Appendix~\ref{app:gauge-projected-and-cyclic-shifts} derives this
boundary-supported shift error and compares the projected shift with its
cyclic completion. We monitor this boundary
occupation together with the physical observables in
Sec.~\ref{subsec:finite-flux-benchmarks}.

\subsection{Finite-flux benchmarks}
\label{subsec:finite-flux-benchmarks}

Flux-cutoff dependence is computed by diagonalizing the Hamiltonian
in the joint kernel of the Gauss-law generators for an open $2\times1$
plaquette strip with six sites, seven links, and two plaquettes. The
$2\times1$ strip is the smallest open geometry containing two plaquettes and
therefore supports the connected correlator $C_{12}$. At $L=6$, the
corresponding physical basis has dimension $2350$.

\begin{figure}[!htbp]
  \centering
  \resizebox{0.70\linewidth}{!}{%
\begin{tikzpicture}[
  x=1cm,
  y=1cm,
  font=\small,
  qubit/.style={
    circle,
    minimum size=2.6mm,
    inner sep=0pt,
    fill=black,
    draw=black
  },
  link/.style={
    line width=0.75pt
  },
  orientation/.style={
    -{Latex[length=2.2mm,width=1.5mm]},
    line width=0.75pt
  },
  legendtext/.style={
    align=left
  }
]

\coordinate (b0) at (0,0);
\coordinate (b1) at (4,0);
\coordinate (b2) at (8,0);
\coordinate (t0) at (0,3.6);
\coordinate (t1) at (4,3.6);
\coordinate (t2) at (8,3.6);

\draw[link] (b0) -- (b1);
\draw[link] (b1) -- (b2);
\draw[link] (t0) -- (t1);
\draw[link] (t1) -- (t2);

\draw[link] (b0) -- (t0);
\draw[link] (b1) -- (t1);
\draw[link] (b2) -- (t2);

\draw[orientation] (1.70,0.00) -- (2.30,0.00);
\draw[orientation] (5.70,0.00) -- (6.30,0.00);
\draw[orientation] (1.70,3.60) -- (2.30,3.60);
\draw[orientation] (5.70,3.60) -- (6.30,3.60);

\draw[orientation] (0.00,1.50) -- (0.00,2.10);
\draw[orientation] (4.00,1.50) -- (4.00,2.10);
\draw[orientation] (8.00,1.50) -- (8.00,2.10);

\foreach \P in {b0,b1,b2,t0,t1,t2}
  \node[qubit] at (\P) {};

\node[below=2pt] at (b0) {$1$};
\node[below=2pt] at (b1) {$2$};
\node[below=2pt] at (b2) {$3$};
\node[above=2pt] at (t0) {$4$};
\node[above=2pt] at (t1) {$5$};
\node[above=2pt] at (t2) {$6$};

\node[below=7pt,fill=white,inner sep=1pt] at (2.00,0.00) {$e_1$};
\node[below=7pt,fill=white,inner sep=1pt] at (6.00,0.00) {$e_2$};
\node[above=7pt,fill=white,inner sep=1pt] at (2.00,3.60) {$e_3$};
\node[above=7pt,fill=white,inner sep=1pt] at (6.00,3.60) {$e_4$};
\node[left=7pt,fill=white,inner sep=1pt] at (0.00,1.80) {$e_5$};
\node[left=7pt,fill=white,inner sep=1pt] at (4.00,1.80) {$e_6$};
\node[left=7pt,fill=white,inner sep=1pt] at (8.00,1.80) {$e_7$};

\node at (2.00,1.80) {$p_1$};
\node at (6.00,1.80) {$p_2$};

\node[qubit] at (8.5,2.65) {};
\node[legendtext,anchor=west]
  at (8.88,2.64)
  {matter qubit};

\draw[link] (8.34,1.65) -- (9.24,1.65);
\draw[orientation] (9,1.65) -- (9.5,1.65);

\node[legendtext,anchor=west]
  at (9.61,1.7)
  {oriented rotor link};

\end{tikzpicture}%
  }
  \caption{Open $2\times1$ gauge--matter strip used in the finite-flux
  benchmarks. Matter qubits occupy the six sites, $U(1)$ rotors occupy the seven
  oriented links, and $p_1$ and $p_2$ denote the two plaquettes. Horizontal
  links are oriented to the right and vertical links upward, with the site and
  link ordering used in the calculation listed in
  Appendix~\ref{app:gauge-finite-gauss-law-basis-numerical-details}.}
  \label{fig:gauge-two-plaquette-geometry}
\end{figure}
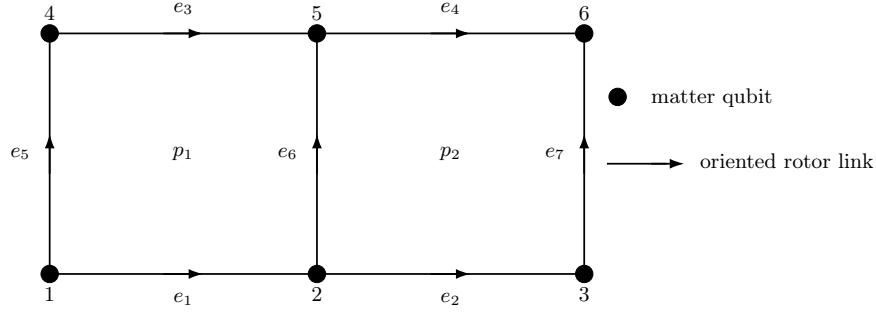
\FloatBarrier

We set
$q_n^{\mathrm{ext}}=0$, $a_{\mathrm{lat}}=1$, $\kappa=1$, $m_0=1.5$, and
$\zeta_{ab}=1$. The flux cutoffs are $L=1,2,3,4,6$, and the scan uses
$55$ logarithmically spaced values of $g^{-2}$ from $10^{-2}$ to $10$.
Here $\sum_n b_n=3$, so Eq.~\eqref{eq:gauge-global-gauss-constraint} fixes
the total matter occupation to $\sum_n\hat n_n=3$.
Using the finite basis, matrix elements, and sparse diagonalization specified
in Appendix~\ref{app:gauge-finite-gauss-law-basis-numerical-details}, we
evaluate
\begin{equation}
  \overline W = \frac{1}{N_p}\sum_p\langle\cos\Phi_p\rangle, \qquad
  \Delta=\mathcal E_1-\mathcal E_0, \qquad
  \langle E^2\rangle_{\mathrm{link}} = \frac{1}{N_\ell}\sum_e\langle E_e^2\rangle.
\end{equation}

\begin{figure}[H]
  \centering
  \includegraphics[width=0.90\textwidth]{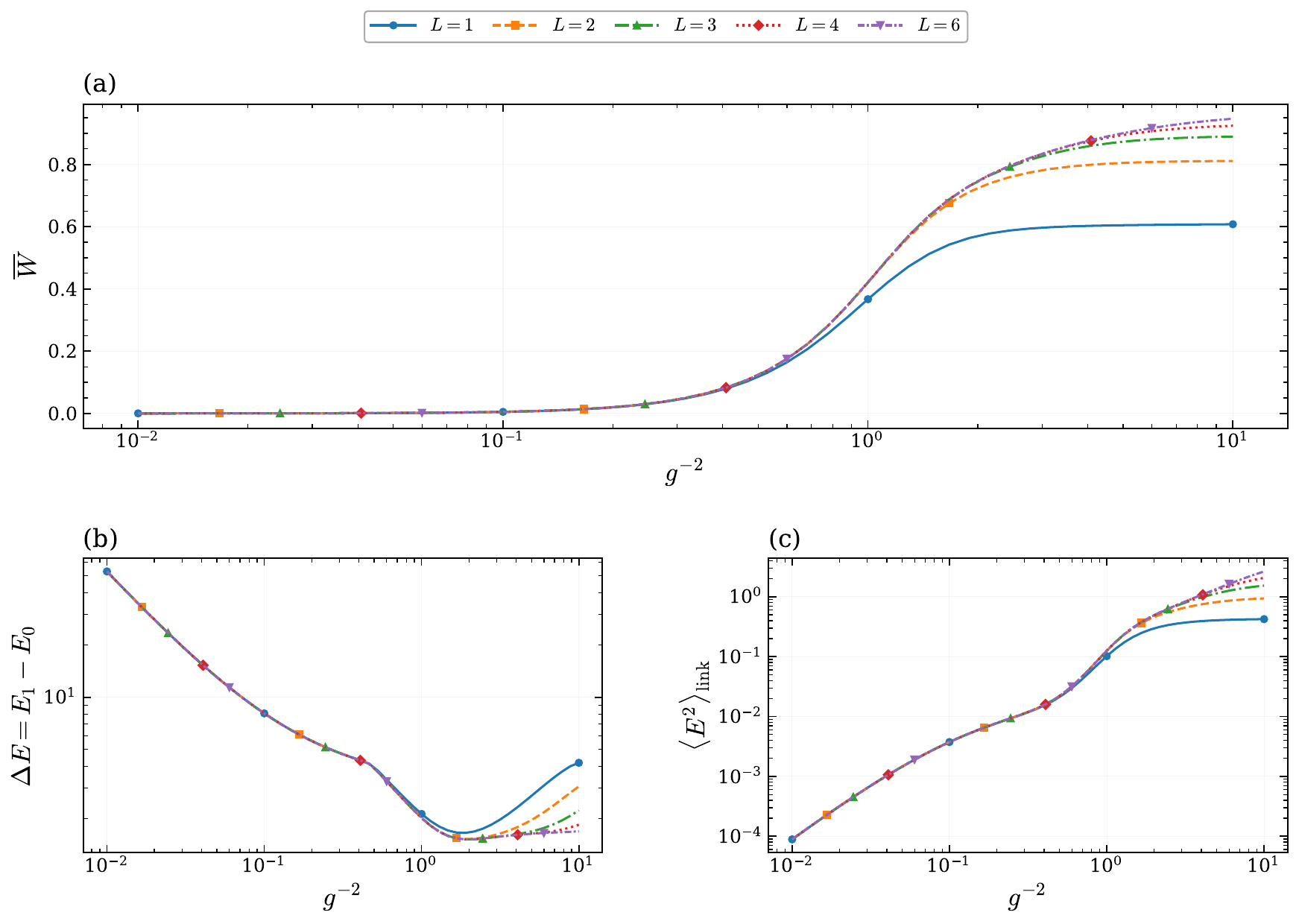}
  \caption{Ground-state results for the open $2\times1$ gauge--matter model.
(a) Average plaquette cosine $\overline W$.
(b) Spectral gap $\Delta=\mathcal E_1-\mathcal E_0$.
(c) Link-averaged squared electric flux
$\langle E^2\rangle_{\mathrm{link}}$.
The curves correspond to flux cutoffs $L=1,2,3,4,6$.
The vertical dotted line marks $g^{-2}=1$.}
  \label{fig:gauge-ground-state-cutoff}
\end{figure}
\vspace{-0.5em}

At small $g^{-2}$, the electric term suppresses nonzero link flux, and
$\overline W$ remains close to zero. As $g^{-2}$ increases, the plaquette
term becomes more important and the ground-state flux distribution broadens.
The spectral gap $\Delta$ develops a minimum near $g^{-2}\sim2$, in
the same crossover region where $\overline W$ rises rapidly and the
link-flux distribution broadens.
Over the sampled interval $g^{-2}\leq1$, the maximum absolute $L=4$ versus
$L=6$ differences in $\overline W$, $\Delta$, and
$\langle E^2\rangle_{\mathrm{link}}$ are $1.34\times10^{-8}$,
$4.69\times10^{-10}$, and $7.23\times10^{-9}$, respectively. The separation between the $L=4$ and $L=6$ curves increases toward $g^{-2}=10$ as more weight approaches the cutoff
boundary.

To probe the coupled qubit--rotor dynamics, we evolve the zero-flux
strong-coupling state
\begin{equation}
  |\Omega_{\mathrm{sc}}\rangle = |\mathbf o=\mathbf b;E=\mathbf0\rangle, \qquad |\psi_L(t)\rangle =
  e^{-itH_L}|\Omega_{\mathrm{sc}}\rangle.
\end{equation}
This state satisfies $DE=\mathbf o-\mathbf b=\mathbf0$. We take
$g^{-2}=1$, $m_0=1.5$, $\kappa=1$, $0\leq t\leq8$, and
$L=1,2,3,4$.

We monitor
\begin{equation}
  \begin{aligned}
    \overline W(t)
    &=
    \frac{1}{N_p}\sum_p\langle\cos\Phi_p\rangle_t,
    \\
    C_{12}(t)
    &=
    \langle\cos\Phi_1\cos\Phi_2\rangle_t
    -
    \langle\cos\Phi_1\rangle_t
    \langle\cos\Phi_2\rangle_t,
    \\
    I_M(t)
    &=
    -\sum_n\eta_n\langle\hat n_n\rangle_t,
    \\
    \overline p_{\partial L}(t)
    &=
    \frac{1}{N_\ell}
    \sum_e
    \left\langle
      \mathbf 1_{\{|E_e|=L\}}
    \right\rangle_t.
  \end{aligned}
\end{equation}

\begin{figure}[H]
  \centering
  \includegraphics[width=0.90\textwidth]
  {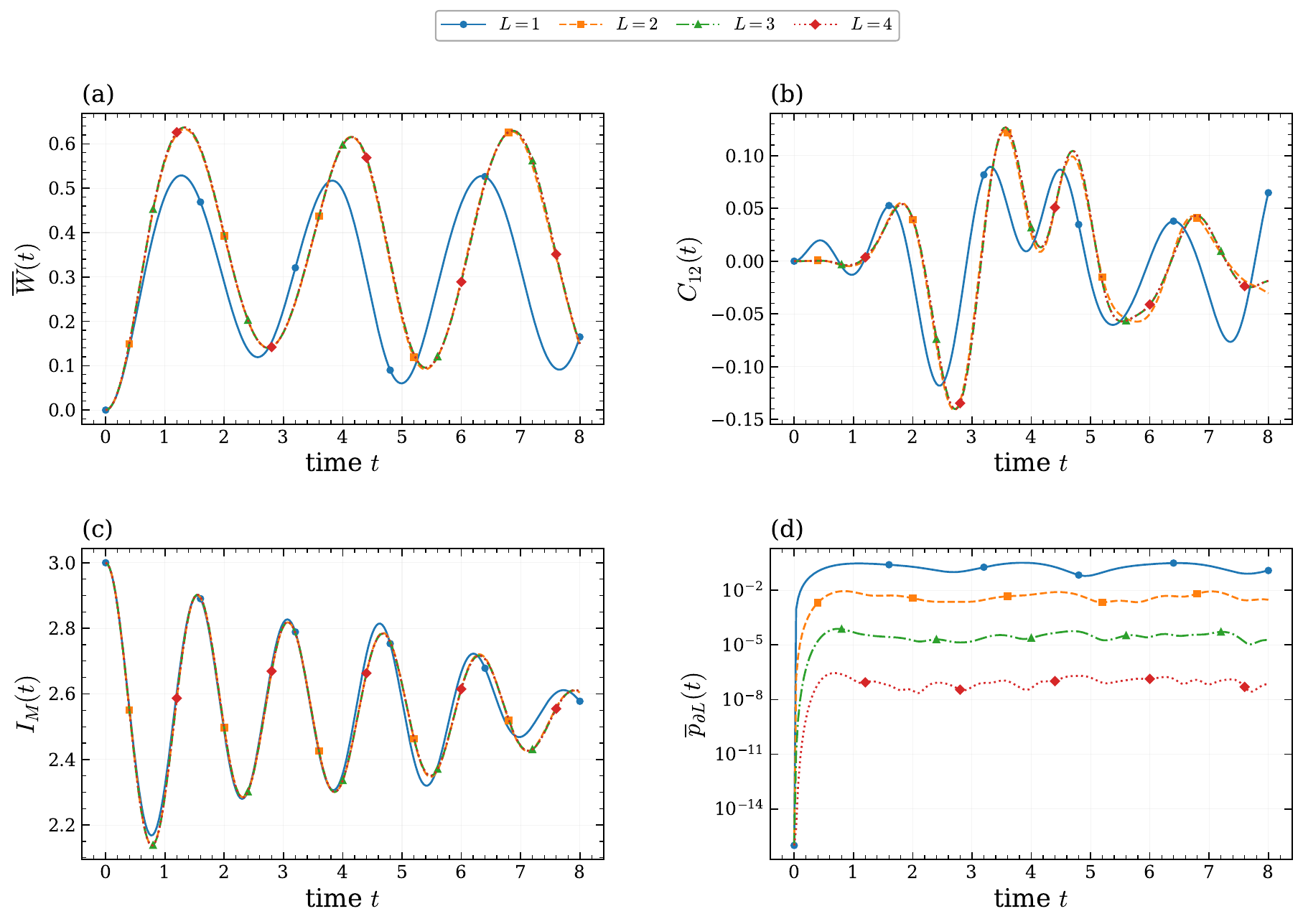}
  \caption{
  Gauge-invariant real-time dynamics of the open $2\times1$
  gauge--matter model following evolution from the zero-flux
  strong-coupling reference state. The parameters are
  $g^{-2}=1$, $m_0=1.5$, and $\kappa=1$.
  (a) Average plaquette cosine $\overline W(t)$.
  (b) Connected two-plaquette correlation $C_{12}(t)$.
  (c) Staggered matter imbalance $I_M(t)$.
  (d) Link-averaged cutoff-boundary occupation
  $\overline p_{\partial L}(t)$.
  The first two panels probe the rotor gauge observables, panel (c) probes
  the matter-qubit dynamics, and panel (d) diagnoses the finite-flux
  truncation.
  }
  \label{fig:gauge-matter-dynamics}
\end{figure}
\vspace{-0.5em}

Starting from the zero-flux state, the plaquette expectation
$\overline W(t)$ develops an oscillatory response, while the
connected correlation $C_{12}(t)$ becomes nonzero. The evolution
therefore generates both local magnetic coherence and correlations
between the two rotor plaquettes. At the same time, the matter imbalance
$I_M(t)$ departs from its initial value and undergoes oscillations. Since
the electric and magnetic terms act only on the rotor register and the
mass term is diagonal in the occupation basis, this redistribution of
matter occupation is generated by the gauge-covariant hopping term
$H_K$.

On the $241$ sampled times, the maximum $L=3$ versus $L=4$ differences in
$\overline W(t)$, $C_{12}(t)$, and $I_M(t)$ are $9.99\times10^{-5}$,
$2.28\times10^{-4}$, and $3.54\times10^{-5}$, respectively. The maximum
$\overline p_{\partial L}(t)$ decreases from $7.60\times10^{-5}$ at $L=3$
to $2.84\times10^{-7}$ at $L=4$.

%%\paragraph{Code and data availability.}
%The supplementary archive accompanying this manuscript contains \path{reproducibility/gauge_model}, including the source code, parameter file, processed data, metadata, software environment, and instructions used to generate Figs.~\ref{fig:gauge-ground-state-cutoff} and \ref{fig:gauge-matter-dynamics}.

\FloatBarrier

\clearpage
\section{Rotor quantum phase estimation and probe-state design}

In rotor quantum phase estimation, a single rotor serves as the phase
register for an \(n\)-qubit target. Let
$U\in\UU(\bH_{n,0})$ have an eigenstate satisfying
\begin{equation}
  U|\psi_\varphi\rangle
  =
  e^{i\varphi}|\psi_\varphi\rangle,
  \qquad
  |\psi_\varphi\rangle\in\bH_{n,0}.
\end{equation}
Standard QPE stores the phase in a finite register and applies an inverse
quantum Fourier transform before measurement
\cite{kitaev1996measurements,cleve1998algorithms,nielsen2000computation}.
With a rotor phase register and direct angle measurement, the inverse
transform can be omitted. The rotor momentum label
$\ell\in\ZZ$ indexes the controlled powers $U^{-\ell}$, and phase kickback
produces the character $e^{-i\ell\varphi}$. Under
$\ZZ\longleftrightarrow\TT$, this character translates the conjugate angle
distribution by $\varphi$, so the eigenphase can be read directly from an
angle measurement.

Bosonic modes have previously been used as phase registers
\cite{travaglione2002phase,liu2016qumode}, and the rotor version with
momentum-controlled powers and direct angle readout was proposed in
Ref.~\cite{kemper2025continuousdiscrete}. Here we take this rotor protocol as
the starting point and allow an arbitrary initial probe state. Its momentum
amplitudes determine the shape of the translated angle distribution, and
therefore control both the estimation error and the range of powers
\(U^{-\ell}\) required by the protocol.

This makes probe preparation the main resource question. At fixed momentum
support, the optimum is the cosine-window probe used below
\cite{buzek1999optimal,imai2009fourier}. At fixed mean kinetic energy, the
optimum is a Mathieu probe
\cite{hayashi2023special}. Sec.~\ref{subsec:qpe-probe-optimization-momentum-energy}
compares these probes with the uniform finite-momentum state and relates
their estimation error directly to the corresponding rotor resources.

\subsection{Phase kickback and direct-angle readout}

For a normalized rotor probe
\begin{equation}
  |\eta\rangle_R
  =
  \sum_{\ell\in\mathbb Z}a_\ell|\ell\rangle_R,
  \qquad
  \sum_{\ell\in\mathbb Z}|a_\ell|^2=1,
  \label{eq:qpe-general-rotor-probe}
\end{equation}
define the momentum-controlled power operation
\begin{equation}
  \mathcal O_U
  :=
  \sum_{\ell\in\mathbb Z}
  U^{-\ell}\otimes|\ell\rangle\langle\ell|.
  \label{eq:qpe-full-rotor-power-oracle}
\end{equation}
The sign convention is chosen so that, with the angle-state convention of
Sec.~\ref{subsec:hybrid-register-phase-space}, the measured angle
distribution is centered at $\theta=\varphi$.

Phase kickback gives
\begin{equation}
  \mathcal O_U
  \bigl(|\psi_\varphi\rangle\otimes|\eta\rangle_R\bigr)
  =
  |\psi_\varphi\rangle\otimes|\eta_\varphi\rangle_R,
  \qquad
  |\eta_\varphi\rangle_R
  =
  \sum_{\ell\in\mathbb Z}
  a_\ell e^{-i\ell\varphi}|\ell\rangle_R.
\end{equation}
The generalized angle states define the periodic angle measurement
\begin{equation}
  E_\theta(d\theta)
  =
  |\theta\rangle\langle\theta|\,d\theta,
  \qquad
  \int_0^{2\pi}E_\theta(d\theta)=I_R.
\end{equation}
Its probability density after phase kickback is
\begin{equation}
  p_\eta(\theta\mid\varphi)
  =
  \frac{1}{2\pi}
  \left|
    \sum_{\ell\in\mathbb Z}
    a_\ell e^{i\ell(\theta-\varphi)}
  \right|^2
  =
  p_\eta(\theta-\varphi\mid0).
  \label{eq:qpe-general-angle-translation}
\end{equation}
Eq.~\eqref{eq:qpe-general-angle-translation} shows that phase kickback
translates the angle profile without changing its shape. The amplitudes
$a_\ell$ determine the localization and tails of
the direct-angle distribution, and the largest occupied $|\ell|$ determines
the required range of controlled powers.

The uniform momentum superposition provides the baseline. Set $d=2^s$ and use
the finite rotor momentum code
\begin{equation}
  \mathcal M_d^{(0)}
  =
  \operatorname{span}
  \{|0\rangle_R,\ldots,|d-1\rangle_R\}.
\end{equation}
Finite rotor momentum encodings have also been considered in
Ref.~\cite{raynal2010rotor}.

The rotor analogue of the uniform qubit phase-register state is
\begin{equation}
  |+_d\rangle_R
  =
  \frac{1}{\sqrt d}
  \sum_{\ell=0}^{d-1}
  |\ell\rangle_R.
\end{equation}
Applying Eq.~\eqref{eq:qpe-full-rotor-power-oracle} gives
\begin{equation}
  \mathcal O_U
  \bigl(
    |\psi_\varphi\rangle\otimes|+_d\rangle_R
  \bigr)
  =
  |\psi_\varphi\rangle
  \otimes
  |\chi_d(\varphi)\rangle_R,
\end{equation}
where
\begin{equation}
  |\chi_d(\varphi)\rangle_R
  =
  \frac{1}{\sqrt d}
  \sum_{\ell=0}^{d-1}
  e^{-i\ell\varphi}
  |\ell\rangle_R.
\end{equation}
Indeed, the momentum component $|\ell\rangle_R$ applies $U^{-\ell}$ to the
target eigenstate and therefore acquires the phase $e^{-i\ell\varphi}$.

Using the angle--momentum Fourier relation in Eq.~\eqref{eq:rotor-angle-momentum-fourier-relation}, the angle-space amplitude after phase kickback is
\begin{equation}
  \langle\theta|\chi_d(\varphi)\rangle_R
  =
  \frac{1}{\sqrt{2\pi d}}
  \sum_{\ell=0}^{d-1}
  e^{i\ell(\theta-\varphi)}.
\end{equation}
The corresponding direct-angle probability density is
\begin{equation}
  p_d(\theta\mid\varphi)
  =
  \frac{1}{2\pi d}
  \left|
    \sum_{\ell=0}^{d-1}
    e^{i\ell(\theta-\varphi)}
  \right|^2
  =
  \frac{1}{2\pi d}
  \frac{
    \sin^2[d(\theta-\varphi)/2]
  }{
    \sin^2[(\theta-\varphi)/2]
  }.
  \label{eq:qpe-direct-angle-readout-distribution}
\end{equation}
Eq.~\eqref{eq:qpe-direct-angle-readout-distribution} is a shifted periodic
Fej\'er kernel. Its maximum is centered at $\theta=\varphi$, and its first
zeros occur at circular distance $2\pi/d$. As $d\to\infty$, the probability
measures $p_d(\theta\mid\varphi)d\theta$ converge weakly to the Dirac measure
  at $\varphi$. Appendix~\ref{app:qpe-fejer-kernel-weak-convergence-direct-angle-readout}
  gives the calculation.

The eigenphase can therefore be recovered by measuring the rotor angle
directly, without applying a coherent Fourier transform to the phase
register. The corresponding rotor protocol is shown in
Fig.~\ref{fig:rotor-qpe-direct-angle}.

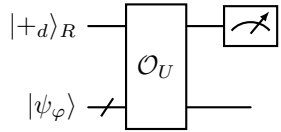
\begin{figure}[ht]
  \centering
  \begin{quantikz}
    \lstick{$|+_d\rangle_R$}
      & \gate[2]{\mathcal O_U}
      & \meter{} \\
    \lstick{$|\psi_\varphi\rangle$}
      & \qwbundle{}
      &
  \end{quantikz}
  \caption{Rotor quantum phase estimation with direct angle readout.
  The momentum-controlled power operation translates the angle
  distribution of the rotor phase register by the target eigenphase
  $\varphi$, which is then read out by an angle measurement.}
  \label{fig:rotor-qpe-direct-angle}
\end{figure}

For the uniform probe above, Fourier processing followed by momentum
measurement recovers the standard discrete QPE distribution. Direct angle
measurement requires no such coherent Fourier stage.

\subsection{Probe optimization under momentum and energy constraints}
\label{subsec:qpe-probe-optimization-momentum-energy}

Two resource constraints determine the probe optimization: finite momentum
support and fixed mean rotor energy.

For the general probe in Eq.~\eqref{eq:qpe-general-rotor-probe}, define its
mean resultant length by
\begin{equation}
  R(\eta)
  :=
  \left|
    \langle\eta|X_R(1)|\eta\rangle
  \right|
  =
  \left|
    \sum_{\ell\in\mathbb Z}
    a_\ell a_{\ell+1}^*
  \right|,
\end{equation}
and use the circular variance
\begin{equation}
  V(\eta)
  :=
  1-R(\eta).
\end{equation}
The modulus makes $R(\eta)$ independent of the angular origin. After centering
the probe at $\theta=0$,
$V(\eta)=\int_0^{2\pi}[1-\cos(\theta-\varphi)]
p_\eta(\theta\mid\varphi)\,d\theta$.
For a narrowly localized distribution,
\begin{equation}
  1-\cos(\theta-\varphi)
  \simeq
  \frac{1}{2}(\theta-\varphi)^2,
\end{equation}
so the angular error scale is
\begin{equation}
  \Delta_\phi
  \simeq
  \sqrt{2V}.
\end{equation}
The Holevo variance $R(\eta)^{-2}-1$ and $V(\eta)=1-R(\eta)$ are both
strictly decreasing functions of $R(\eta)$, so they select the same optimal
probe states.

We also use the dimensionless rotor kinetic energy
\begin{equation}
  E_R(\eta)
  :=
  \langle\eta|\hat\ell^2|\eta\rangle
  =
  \sum_{\ell\in\mathbb Z}
  \ell^2|a_\ell|^2.
\end{equation}
For probes with finite momentum support, the largest occupied $|\ell|$
determines the required range of controlled powers; $E_R$ measures the rotor
kinetic-energy resource.

\paragraph{Uniform momentum superposition.}
For comparison, consider the centered uniform state
\begin{equation}
  |\eta_{\mathrm{unif}}^{(L)}\rangle_R
  =
  \frac{1}{\sqrt{2L+1}}
  \sum_{\ell=-L}^{L}|\ell\rangle_R.
  \label{eq:qpe-centered-uniform-probe}
\end{equation}

For Eq.~\eqref{eq:qpe-centered-uniform-probe},
\begin{equation}
  R_{\mathrm{unif}}
  =
  \frac{2L}{2L+1},
  \qquad
  V_{\mathrm{unif}}
  =
  \frac{1}{2L+1}.
\end{equation}
The corresponding energy is
\begin{equation}
  E_{\mathrm{unif}}
  =
  \frac{1}{2L+1}
  \sum_{\ell=-L}^{L}\ell^2
  =
  \frac{L(L+1)}{3}.
\end{equation}
Since
\begin{equation}
  2L+1
  =
  \sqrt{12E_{\mathrm{unif}}+1},
\end{equation}
the circular variance obeys
\begin{equation}
  V_{\mathrm{unif}}
  =
  \frac{1}{\sqrt{12E_{\mathrm{unif}}+1}}
  =
  O(E_{\mathrm{unif}}^{-1/2}).
  \label{eq:qpe-uniform-energy-scaling}
\end{equation}
The corresponding error scale is therefore
$\Delta_{\mathrm{unif}}=O(E_{\mathrm{unif}}^{-1/4})$.

This scaling reflects the slowly decaying sidelobes of the uniform
finite-momentum profile. Its central lobe still has angular width
$O(1/L)=O(E_{\mathrm{unif}}^{-1/2})$. The weaker scaling in
Eq.~\eqref{eq:qpe-uniform-energy-scaling} arises because the circular variance
is sensitive to probability away from the central peak.

\paragraph{Finite-support optimum.}

Among all normalized rotor states supported on $\{-L,\ldots,L\}$, the
cosine-window probe is optimal. This is the standard finite-support
sine-state optimizer of Refs.~\cite{buzek1999optimal,imai2009fourier},
written in centered momentum coordinates:
\begin{equation}
  |\eta_{\mathrm{cos}}^{(L)}\rangle_R
  =
  \frac{1}{\sqrt{L+1}}
  \sum_{\ell=-L}^{L}
  \cos\!\left(
    \frac{\pi\ell}{2(L+1)}
  \right)
  |\ell\rangle_R.
  \label{eq:qpe-cosine-window-probe}
\end{equation}
With $n=\ell+L$, the cosine coefficient becomes
$\sin[\pi(n+1)/(2L+2)]$.
Appendix~\ref{app:qpe-finite-support-probe-optimization} recovers its
finite-support optimality from the numerical radius of the compressed rotor
shift.

In particular,
\begin{equation}
  R_{\mathrm{cos}}
  =
  \cos\!\left(
    \frac{\pi}{2(L+1)}
  \right).
\end{equation}

Its circular variance is therefore
\begin{equation}
  V_{\mathrm{cos}}
  =
  1-
  \cos\!\left(
    \frac{\pi}{2(L+1)}
  \right)
  =
  \frac{\pi^2}{8(L+1)^2}
  +
  O(L^{-4}).
  \label{eq:qpe-cosine-circular-variance}
\end{equation}
The corresponding error scale is
\begin{equation}
  \Delta_{\cos}
  \sim
  \frac{\pi}{2(L+1)}.
\end{equation}

The kinetic energy of this state is
\begin{equation}
  E_{\mathrm{cos}}
  =
  \frac{2L^2+4L+3}{6}
  -
  \frac{1}{2}
  \csc^2\!\left(
    \frac{\pi}{2(L+1)}
  \right),
  \label{eq:qpe-cosine-energy-exact}
\end{equation}
with the large-$L$ behavior
\begin{equation}
  E_{\mathrm{cos}}
  =
  \left(
    \frac{1}{3}
    -
    \frac{2}{\pi^2}
  \right)(L+1)^2
  +
  O(L^{-2}).
  \label{eq:qpe-cosine-energy-asymptotic}
\end{equation}
Combining Eqs.~\eqref{eq:qpe-cosine-circular-variance} and
\eqref{eq:qpe-cosine-energy-asymptotic} gives
\begin{equation}
  V_{\mathrm{cos}}
  \sim
  \frac{\pi^2-6}{24E_{\mathrm{cos}}},
  \qquad
  \Delta_{\cos}
  \sim
  \sqrt{
    \frac{\pi^2-6}{12E_{\mathrm{cos}}}
  }.
  \label{eq:qpe-cosine-energy-scaling}
\end{equation}

The cosine taper suppresses the angular tails responsible for the larger
circular variance of the uniform probe. It achieves the
$E_{\mathrm{cos}}^{-1/2}$ error scaling while retaining strictly finite
momentum support.

\paragraph{Fixed-energy Mathieu optimum.}
The cosine-window probe is optimal when the momentum support is fixed.
A second optimization removes the support restriction and fixes the mean
rotor energy $E_R=\langle\hat\ell^2\rangle$. Under this constraint, the
minimum-circular-variance probe is a Mathieu state
\cite{hayashi2023special}.

For a prescribed energy $E_R$, consider
\begin{equation}
  V_\star(E_R)
  :=
  \min_{\substack{
    \langle\eta|\eta\rangle=1\\
    \langle\eta|\hat\ell^2|\eta\rangle=E_R
  }}
  \left[
    1-
    \left|
      \langle\eta|X_R(1)|\eta\rangle
    \right|
  \right].
  \label{eq:qpe-fixed-energy-optimization}
\end{equation}
After centering the angular distribution at $\theta=0$, the minimizing
state can be chosen even, and minimizing $V$ reduces to maximizing
$\langle\cos\hat\theta\rangle$ at fixed
$\langle\hat\ell^2\rangle$.

Introducing a Lagrange multiplier $\lambda>0$ for the energy constraint
leads to $H_\lambda:=\hat\ell^2-\lambda\cos\hat\theta$. Its ground state
minimizes the corresponding variational functional and satisfies
\begin{equation}
  \left(
    \hat\ell^2
    -
    \lambda\cos\hat\theta
  \right)
  |\eta_\lambda\rangle
  =
  \varepsilon_\lambda
  |\eta_\lambda\rangle.
  \label{eq:qpe-mathieu-hamiltonian}
\end{equation}
Choosing $\lambda$ so that
$\langle\eta_\lambda|\hat\ell^2|\eta_\lambda\rangle=E_R$ gives the constrained
optimum in Eq.~\eqref{eq:qpe-fixed-energy-optimization}.

In the angle representation, $\hat\ell=-i\,d/d\theta$. With
$z=\theta/2$, Eq.~\eqref{eq:qpe-mathieu-hamiltonian} takes the standard
Mathieu form
\begin{equation*}
  \frac{d^2\eta_\lambda}{dz^2}
  +
  \left(a_\lambda-2q_\lambda\cos 2z\right)\eta_\lambda
  =0,
  \qquad
  a_\lambda=4\varepsilon_\lambda,
  \qquad
  q_\lambda=-2\lambda.
\end{equation*}
The normalized even ground-state wavefunction is therefore
\begin{equation}
  \eta_\lambda(\theta)
  =
  \langle\theta|\eta_\lambda\rangle
  =
  \mathcal N_\lambda
  \operatorname{ce}_0
  \left(
    \frac{\theta}{2},
    -2\lambda
  \right).
\end{equation}

For large $E_R$, the optimal circular variance satisfies
\begin{equation}
  V_\star(E_R)
  \sim
  \frac{1}{8E_R},
  \label{eq:qpe-mathieu-large-energy}
\end{equation}
so the corresponding error scale obeys
\begin{equation}
  \Delta_\star
  \sim
  \frac{1}{2\sqrt{E_R}}.
  \label{eq:qpe-mathieu-error-scaling}
\end{equation}
The Mathieu probe therefore achieves the same
$E_R^{-1/2}$ error scaling as the cosine-window probe, with the
optimal prefactor under the fixed-energy constraint.
Appendix~\ref{app:qpe-fixed-energy-probe-asymptotics} derives
Eqs.~\eqref{eq:qpe-mathieu-large-energy}
and~\eqref{eq:qpe-mathieu-error-scaling} from the local quadratic approximation
to the Mathieu ground state.

The Hamiltonian in Eq.~\eqref{eq:qpe-mathieu-hamiltonian} contains the same
cosine potential used in the universal-control construction of Sec.~\ref{sec:universal-control-clifford-non-clifford-gates}.

\paragraph{Asymptotic comparison.}

Table~\ref{tab:qpe-probe-comparison} collects the three direct-angle probe
families and their resource constraints.

\begin{table}[H]
  \centering
  \begin{tabular}{lccc}
    \toprule
    Probe
    & Constraint
    & $V$
    & Error scale $\Delta_\phi$ \\
    \midrule
    Uniform momentum
    & $|\ell|\le L$
    & $\displaystyle
       \sim \frac{1}{\sqrt{12E_R}}$
    & $\displaystyle O(E_R^{-1/4})$ \\[1.5ex]
    Cosine-window
    & $|\ell|\le L$
    & $\displaystyle
       \sim \frac{\pi^2-6}{24E_R}$
    & $\displaystyle
       \sim
       \sqrt{\frac{\pi^2-6}{12E_R}}$ \\[1.5ex]
    Mathieu
    & $E_R$ fixed
    & $\displaystyle \sim\frac{1}{8E_R}$
    & $\displaystyle \sim\frac{1}{2\sqrt{E_R}}$ \\
    \bottomrule
  \end{tabular}
  \caption{Asymptotic comparison of direct-angle rotor probes. The cosine-window
  probe minimizes $V$ at fixed momentum support. The Mathieu probe minimizes
  $V$ at fixed mean rotor energy.}
  \label{tab:qpe-probe-comparison}
\end{table}

Both optimized probes have $E_R^{-1/2}$ error scaling under the circular
variance $V$, compared with $E_R^{-1/4}$ for the uniform momentum profile.
The Mathieu probe has the smaller leading variance prefactor at fixed energy,
$1/8<(\pi^2-6)/24$, and the cosine-window probe has strictly finite momentum
support.

\clearpage

\section{Quantum Fourier transforms with rotor registers}
\label{sec:qft-rotor-registers}
A finite rotor code can itself serve as a computational register. 
A \(d\)-dimensional logical register can be represented either by states
localized at \(d\) angles or by \(d\) classes of rotor momentum. We use
these two encodings to construct the same \(d^r\)-point Fourier transform in
two ways.

The first construction maps an angle-code input to a momentum-code output.
It follows the standard base-\(d\) Fourier factorization
\cite{nielsen2000computation,pavlidis2021fourier}, with one local
angle-to-momentum transform on each rotor and mixed angle--momentum phases
between different registers. The second construction stays within a finite
momentum code and compiles the same Fourier matrix using the gates introduced
earlier. In this construction, each cross-register base-\(d\) phase is
implemented by the quadratic rotor Clifford gate $\mathrm{CPHS}_{jk}(\alpha) =
e^{i\alpha\hat\ell_j\hat\ell_k}$.

The two constructions therefore differ mainly in the local Fourier step:
the first takes the angle-to-momentum code map as a logical operation, while
the second resolves the Fourier transform into momentum-code gates.

When the input or output is stored in qubits, a qubit--rotor interface can
be added before or after the Fourier circuit. The construction in
Sec.~\ref{subsec:qft-rotor-codes-qubit-rotor-interfaces} adapts the displacement-and-reset state-transfer method of
Ref.~\cite{bierman2026statetransfer}, and related mixed analog--digital Fourier
methods are developed in Ref.~\cite{liu2025mixed}.

\subsection{Rotor codes and qubit--rotor interfaces}
\label{subsec:qft-rotor-codes-qubit-rotor-interfaces}

The constructions use two complementary \(d\)-dimensional rotor codes. The
angle code represents \(x\in\{0,\ldots,d-1\}\) by a localized wavepacket centered
at an angular grid point, and the momentum code represents
\(p\in\{0,\ldots,d-1\}\) by a superposition over one residue class of the
integer rotor momentum.  The generalized angle states introduced in
Sec.~\ref{subsec:hybrid-register-phase-space} are not normalizable and
therefore do not themselves define codewords.

For the angle code, set
\begin{equation}
  \theta_x := \frac{2\pi x}{d}, \qquad x=0,\ldots,d-1 .
\end{equation}
We begin with the von Mises packets
\begin{equation}
  |x;\kappa\rangle := \mathcal N_\kappa \int_0^{2\pi} e^{\kappa\cos(\theta-\theta_x)} |\theta\rangle_R\,d\theta,
  \qquad \mathcal N_\kappa := \frac{1}{\sqrt{2\pi I_0(2\kappa)}} ,
\end{equation}
where \(\kappa>0\) controls the angular localization and \(I_\ell\) is the
modified Bessel function of the first kind.  These states are translations
of a single packet,
\begin{equation}
  |x;\kappa\rangle
  =
  Z_R(-\theta_x)|0;\kappa\rangle,
\end{equation}
and have the momentum expansion
\begin{equation}
  |x;\kappa\rangle = \sum_{\ell\in\mathbb Z} c_\ell^{(\kappa)} e^{-i\ell\theta_x} |\ell\rangle_R, \qquad
  c_\ell^{(\kappa)} := \frac{I_\ell(\kappa)} {\sqrt{I_0(2\kappa)}} .
\end{equation}

The von Mises packets are normalized and nonorthogonal.  Let
\begin{equation}
  G_{xy}^{(\kappa)} := \langle x;\kappa|y;\kappa\rangle = \frac{ I_0\!\left( 2\kappa
  \cos\!\frac{\theta_x-\theta_y}{2} \right) }{ I_0(2\kappa) } .
\end{equation}
The Gram matrix \(G^{(\kappa)}\) is circulant and positive definite.  We
define the orthonormal angle-code states by
\begin{equation}
  |\overline x\rangle := \sum_{y=0}^{d-1} |y;\kappa\rangle \bigl(G^{(\kappa)}\bigr)^{-1/2}_{yx}, \qquad
  x=0,\ldots,d-1 .
\end{equation}
They satisfy
\begin{equation}
  \langle\overline x|\overline y\rangle = \delta_{xy}, \qquad |\overline x\rangle =
  Z_R(-\theta_x)|\overline 0\rangle .
\end{equation}
We write
\begin{equation}
  \mathcal A_d := \operatorname{span} \bigl\{ |\overline x\rangle: 0\leq x<d \bigr\}.
\end{equation}
The off-diagonal overlaps obey
$ 
  \lVert G^{(\kappa)}-I\rVert_2
  \leq
  (d-1)\max_{x\neq y}|G_{xy}^{(\kappa)}|.
$

For the momentum code, choose a normalized sequence
\(f=(f_n)_{n\in\mathbb Z}\in\ell^2(\mathbb Z)\),
\begin{equation}
  \sum_{n\in\mathbb Z}|f_n|^2=1,
\end{equation}
and define
\begin{equation}
  |\overline p\rangle := \sum_{n\in\mathbb Z} f_n\,|p+nd\rangle_R, \qquad p=0,\ldots,d-1 .
\end{equation}
Thus \( |\overline p\rangle \) is supported on the momentum values satisfying
\begin{equation}
  \ell\equiv p\pmod d .
\end{equation}
Different values of \(p\) correspond to disjoint residue classes, and hence
\begin{equation}
  \langle\overline p|\overline q\rangle
  =
  \delta_{pq}.
\end{equation}
The momentum-code space is
\begin{equation}
  \mathcal M_d(f) := \operatorname{span} \bigl\{ |\overline p\rangle: 0\leq p<d \bigr\}.
\end{equation}
The projector onto the momentum code is
\begin{equation}
  \Pi_M := \sum_{p=0}^{d-1} |\overline p\rangle \langle\overline p|.
\end{equation}
For the CPHS realization and gate-set compilation below, we choose
\begin{equation}
  f_n=\delta_{n0}, \qquad |\overline p\rangle=|p\rangle_R,
  \qquad p=0,\ldots,d-1.
  \label{eq:gate-set-momentum-codewords}
\end{equation}
The corresponding code space is
\begin{equation}
  \mathcal M_d^{(0)}:=\mathcal M_d(\delta_0)
  =\operatorname{span}\{|0\rangle_R,\ldots,|d-1\rangle_R\}.
  \label{eq:gate-set-momentum-code}
\end{equation}

Set
\begin{equation}
  \omega_d:=e^{2\pi i/d}.
\end{equation}
The angle-to-momentum Fourier operation
\(F_d^{A\to M}\) is defined on the angle code by
\begin{equation}
  F_d^{A\to M} |\overline x\rangle = \frac{1}{\sqrt d} \sum_{p=0}^{d-1} \omega_d^{px}
  |\overline p\rangle, \qquad x=0,\ldots,d-1 .
\end{equation}
The states on each side form orthonormal bases, so
$F_d^{A\to M}$ is an active code-space unitary from \(\mathcal A_d\) to
\(\mathcal M_d(f)\); the Fourier relation between the angle and momentum
wavefunctions of a fixed rotor state is a change of representation. The
code-space map extends to a unitary on
\(L^2(\mathbb T)\).

The momentum-code Fourier transform is defined by
\begin{equation}
  F_d^M |\overline p\rangle = \frac{1}{\sqrt d} \sum_{q=0}^{d-1} \omega_d^{qp} |\overline q\rangle, \qquad
  p=0,\ldots,d-1 .
\end{equation}
The two operations use the same \(d\times d\) Fourier matrix.  Their domains
are
\begin{equation}
  F_d^{A\to M} : \mathcal A_d \longrightarrow \mathcal M_d(f), \qquad
  F_d^M : \mathcal M_d(f) \longrightarrow \mathcal M_d(f).
\end{equation}
The first operation is used in the angle--momentum construction.  For the
specialization $f=\delta_0$, the second operation is compiled from the hybrid gate set in
Sec.~\ref{subsec:qft-gate-set-compilation-rotor-momentum-codes}.

For \(r\) rotors, the input and output code spaces are
\begin{equation}
  \mathcal A_d^{\otimes r}
  \qquad\text{and}\qquad
  \mathcal M_d(f)^{\otimes r}.
\end{equation}

If the input amplitudes are initially stored in qubits, they may be
transferred coherently to the rotor codes before the Fourier operation.
For the angle code, one may adapt the displacement-and-reset construction
of Ref.~\cite{bierman2026statetransfer}: qubit-controlled angular
translations write the binary input label into the center of a localized
rotor packet, and rotor-conditioned qubit operations return the source
qubits to $|0\rangle$. The periodic rotor angle removes the padding
and anti-padding steps needed to isolate a finite cyclic region of a
noncompact oscillator.

The inverse sequence transfers the rotor-code amplitudes back to qubits
when a qubit output is required.

\subsection{Angle--momentum construction of the base-\texorpdfstring{$d$}{d} QFT}

The angle and momentum codes implement the base-\(d\) quantum Fourier
transform across \(r\) rotor registers. The
construction treats the active code-space unitary
\(F_d^{A\to M}\) and the mixed code-space phase as
logical operations. The explicit gate-set construction in
Sec.~\ref{subsec:qft-gate-set-compilation-rotor-momentum-codes}
resolves the momentum-code Fourier transform into the hybrid gates
introduced earlier.

To define the mixed phase, introduce the finite-code digit operators
\begin{equation}
  D_A := \sum_{x=0}^{d-1} x |\overline x\rangle \langle\overline x|, \qquad
  D_M := \sum_{p=0}^{d-1} p |\overline p\rangle \langle\overline p|.
\end{equation}
The first operator is defined on the angle code
\(\mathcal A_d\), and the second is defined on the momentum code
\(\mathcal M_d(f)\). These are finite-rank spectral operators on their
respective code spaces, extended by zero on the orthogonal complements.

For two distinct rotor registers \(R_j\) and \(R_k\), define the
mixed angle--momentum phase
\begin{equation}
  \Lambda_{jk}(\beta) := \exp\!\left[ i\beta\, D_{M,j} \otimes
  D_{A,k} \right].
\end{equation}
Its action on the relevant code spaces is
\begin{equation}
  \Lambda_{jk}(\beta) |\overline p\rangle_{R_j} |\overline x\rangle_{R_k} = e^{i\beta px}
  |\overline p\rangle_{R_j} |\overline x\rangle_{R_k}, \qquad 0\leq p,x<d.
\end{equation}
For a fixed \(j\), the operations
\(\Lambda_{jk}(\beta)\) with \(k>j\) commute.

Set \(N=d^r\). Write an input integer in base \(d\) as
\begin{equation}
  X(\mathbf x) := \sum_{k=1}^{r} x_kd^{r-k}, \qquad \mathbf x = \bigl( x_1,\ldots,x_r \bigr), \qquad
  0\leq x_k<d,
\end{equation}
with the most significant digit stored in \(R_1\). The corresponding
angle-code basis state is
\begin{equation}
  |\overline{\mathbf x}\rangle := \bigotimes_{k=1}^{r} |\overline{x_k}\rangle_{R_k} \in
  \mathcal A_d^{\otimes r}.
\end{equation}

For the momentum-code output, write
\begin{equation}
  P(\mathbf p) := \sum_{j=1}^{r} p_jd^{r-j}, \qquad \mathbf p = \bigl( p_1,\ldots,p_r \bigr), \qquad
  0\leq p_j<d,
\end{equation}
and define
\begin{equation}
  |\overline{\mathbf p}\rangle := \bigotimes_{j=1}^{r} |\overline{p_j}\rangle_{R_j} \in
  \mathcal M_d(f)^{\otimes r}.
\end{equation}

The angle-to-momentum \(N\)-point Fourier transform is defined by
\begin{equation}
  F_N^{A\to M} |\overline{\mathbf x}\rangle = \frac{1}{\sqrt{d^r}}
  \sum_{p_1,\ldots,p_r=0}^{d-1} \exp\!\left[ \frac{2\pi i}{d^r} X(\mathbf x)P(\mathbf p) \right]
  |\overline{\mathbf p}\rangle .
\end{equation}
Its matrix from
\(\mathcal A_d^{\otimes r}\) to
\(\mathcal M_d(f)^{\otimes r}\) is the standard positive-exponent
\(d^r\)-point discrete Fourier matrix.

The circuit processes the rotor registers in the order
\[
  R_1,R_2,\ldots,R_r.
\]
When register $R_j$ is processed, it still contains the angle-code state
\(\lvert\overline{x_j}\rangle\). We first apply the local transform
\(F_{d,j}^{A\to M}\), whose action is
\begin{equation}
  F_{d,j}^{A\to M} \lvert\overline{x_j}\rangle_{R_j} = \frac{1}{\sqrt d}
  \sum_{p_j=0}^{d-1} \exp\!\left( \frac{2\pi i}{d} p_jx_j \right)
  \lvert\overline{p_j}\rangle_{R_j}.
\end{equation}
The register \(R_j\) is now in the momentum code, and each register
\(R_k\) with \(k>j\) remains in the angle code. For every \(k>j\), we then
apply
\begin{equation}
  \Lambda_{jk}
  \left(
    \frac{2\pi}{d^{k-j+1}}
  \right).
\end{equation}
On the basis state
\(
\lvert\overline{p_j}\rangle_{R_j}
\lvert\overline{x_k}\rangle_{R_k}
\),
this operation contributes the phase
\begin{equation}
  \exp\!\left(
    \frac{2\pi i}{d^{k-j+1}}
    p_jx_k
  \right).
\end{equation}
The total phase associated with \(p_j\) after processing register
\(R_j\) is therefore
\begin{equation}
  \exp\!\left[ 2\pi i p_j \left( \frac{x_j}{d} + \frac{x_{j+1}}{d^2} + \cdots +
  \frac{x_r}{d^{r-j+1}} \right) \right].
\end{equation}

Define the \(j\)th register step by
\begin{equation}
  V_j := \left[ \prod_{k=j+1}^{r} \Lambda_{jk} \left( \frac{2\pi}{d^{k-j+1}} \right)
  \right] F_{d,j}^{A\to M}, \qquad 1\leq j\leq r.
\end{equation}
The product is the identity for \(j=r\). For fixed \(j\), all mixed
angle--momentum phase operations in the product commute.

The complete circuit is
\begin{equation}
  U_{\mathrm{AM}} := V_r V_{r-1} \cdots V_1,
\end{equation}
with \(V_1\) acting first.

The circuit produces the Fourier output digits in reverse register order.
We therefore define the register-reversal unitary
\(\mathrm{REV}_r\) by
\begin{equation}
  \mathrm{REV}_r \left( \lvert\overline{p_1}\rangle_{R_1} \otimes\cdots\otimes
  \lvert\overline{p_r}\rangle_{R_r} \right) := \lvert\overline{p_r}\rangle_{R_1} \otimes\cdots\otimes
  \lvert\overline{p_1}\rangle_{R_r}.
\end{equation}
This operation only reverses the order of the rotor registers. It may be
implemented physically by rotor SWAP gates or absorbed into the output
register convention.

With this output convention, the complete circuit acts on
\(\mathcal A_d^{\otimes r}\) as
\begin{equation}
  U_{\mathrm{AM}}
  =
  \mathrm{REV}_r
  F_N^{A\to M}.
  \label{eq:angle-momentum-base-d-qft}
\end{equation}
Thus the circuit implements the standard positive-exponent
\(d^r\)-point Fourier transform, with the output rotor registers appearing
in reverse order. Appendix~\ref{app:qft-angle-momentum-factorization} proves
this factorization by matching the accumulated digit phase to the base-$d$
Fourier kernel.

The factorization contains $r$ local angle-to-momentum transforms and
$\binom{r}{2}$ mixed code-space phases, up to output-register relabeling.

We next relate the two logical operations to rotor-space operations. The local
transform \(F_d^{A\to M}\) changes the momentum envelope and therefore requires
controls that modify momentum-amplitude magnitudes.

For the mixed phase \(\Lambda_{jk}\), a periodic multiplication operator
provides an approximate realization on localized angle codes, with the
code-space error and leakage bounds of
Appendix~\ref{app:qft-mixed-angle-momentum-realization}.

\Needspace{10\baselineskip}
For the specialization $f_n=\delta_{n0}$, the mixed phase also admits a
reduction to CPHS. Define the label-preserving code isometry
$T_d:\mathcal A_d\to\mathcal M_d^{(0)}$ by
\begin{equation}
  T_d|\overline x\rangle=|x\rangle_R, \qquad 0\leq x<d,
\end{equation}
and choose a unitary extension to $L^2(\mathbb T)$, also denoted by $T_d$.
Let $T_k$ act as $T_d$ on register $R_k$ and as the identity on the
other registers. On the tensor product of $\mathcal M_d^{(0)}$ in register
$R_j$ and $\mathcal A_d$ in register $R_k$,
\begin{equation}
  \Lambda_{jk}(\beta)
  =T_k^{\dagger}\,\mathrm{CPHS}_{jk}(\beta)\,T_k.
  \label{eq:qft-mixed-angle-momentum-cphs-reduction}
\end{equation}
The map $T_d$ preserves the logical digit and changes its encoding.
The map $F_d^{A\to M}$ applies the finite Fourier matrix. For
the momentum-code specialization in
Eq.~\eqref{eq:gate-set-momentum-codewords},
Eq.~\eqref{eq:qft-mixed-angle-momentum-cphs-reduction} gives a CPHS
realization of the mixed logical phase.

\subsection{Gate-set compilation on rotor momentum codes}
\label{subsec:qft-gate-set-compilation-rotor-momentum-codes}

The angle--momentum construction treats the local
angle-to-momentum Fourier operation and the mixed angle--momentum phase as
logical operations. The momentum-code construction uses MQR, CShift,
conditional-phase, CPHS, and Hadamard gates as logical instructions. We use
the specialization $\mathcal M_d^{(0)}$ in
Eq.~\eqref{eq:gate-set-momentum-code}, with $d=2^s$.
Throughout this subsection and
Appendix~\ref{app:qft-gate-set-compilation}, the momentum code is
$\mathcal M_d^{(0)}$. When a qubit--rotor interface is required, the
binary label can be loaded by qubit-controlled shifts of distances
$1,2,\ldots,2^{s-1}$, after which momentum-conditioned qubit rotations
reset the source register. This gives the code-space action
\begin{equation*}
  \sum_{p=0}^{d-1}
  c_p|p\rangle_Q|0\rangle_R
  \longmapsto
  |0\rangle_Q^{\otimes s}
  \sum_{p=0}^{d-1}
  c_p|p\rangle_R .
\end{equation*}
For one $s$-qubit block, the transfer uses $s$ CShift and $s$ MQR
instructions. The
inverse transfer has the same count. For $r$ registers, a one-way transfer
therefore uses $rs$ CShift and $rs$ MQR instructions.

\subsubsection{Compiled QFT on one rotor momentum code}

The positive-exponent \(d\)-point Fourier transform is compiled on the
momentum code \(\mathcal M_d^{(0)}\). The qubit--rotor interface described
above is used only when the input amplitudes are initially stored in qubits
or a qubit output is required. The Fourier circuit acts entirely within
\(\mathcal M_d^{(0)}\).

Recall that
\begin{equation}
  F_d^M |\overline p\rangle_R = \frac{1}{\sqrt d} \sum_{q=0}^{d-1} \omega_d^{qp} |\overline q\rangle_R,
  \qquad 0\leq p<d,
\end{equation}
with \(d=2^s\) and
\(\omega_d=e^{2\pi i/d}\).

Let \(a\) be a reusable ancilla qubit initialized in
\(|0\rangle_a\). For \(j=0,\ldots,s-1\), define the momentum-conditioned
ancilla rotation
\begin{equation}
  B_j := \sum_{m=0}^{d-1} (iY_a)^{m_j} \otimes |\overline m\rangle \langle\overline m| + I_a\otimes \bigl(
  I_R-\Pi_M \bigr),
\end{equation}
where
\begin{equation}
  m_j := \left\lfloor \frac{m}{2^j} \right\rfloor \bmod 2
\end{equation}
is the \(j\)th binary digit of the current rotor momentum. The gate
\(B_j\) is an MQR operation. In particular,
\begin{equation}
  B_j^\dagger |0\rangle_a |\overline m\rangle_R = |m_j\rangle_a |\overline m\rangle_R .
\end{equation}

Define
\begin{equation}
  S_j := \mathrm{CShift}_{a\to R}(-2^j), \qquad \alpha_j := \frac{2\pi}{2^{j+1}} .
\end{equation}
The conditional phase \(C_j\) is the gate introduced in
Eq.~\eqref{eq:momentum-dependent-conditional-phase}, with the ancilla as
the qubit register and \(P_{c_a}=Z_a\). Its action is
\begin{equation}
  C_j |c\rangle_a |\ell\rangle_R = e^{i\alpha_j c\ell} |c\rangle_a |\ell\rangle_R, \qquad c\in\{0,1\}.
\end{equation}

The \(j\)th Fourier stage is
\begin{equation}
  F_j := B_j S_j^\dagger C_j (H_a\otimes I_R) S_j B_j^\dagger .
\end{equation}
The factors act from right to left in the order
\[
  B_j^\dagger,\quad
  S_j,\quad
  H_a\otimes I_R,\quad
  C_j,\quad
  S_j^\dagger,\quad
  B_j .
\]

\begin{figure}[t]
\centering

\begin{tikzpicture}[
  >=Latex,
  x=1cm,
  y=1cm
]

\coordinate (ain)  at (0.25,1.20);
\coordinate (rin)  at (0.25,0.00);
\coordinate (aout) at (11.25,1.20);
\coordinate (rout) at (11.25,0.00);

\node[
  draw,
  rounded corners=1.5pt,
  fill=white,
  minimum width=10mm,
  minimum height=17mm,
  inner sep=1pt,
  font=\small
] (Bjdag) at (1.50,0.60)
  {$B_j^\dagger$};

\node[
  font=\scriptsize,
  align=center
] at (1.50,1.80)
  {extract\\$p_j$};

\node[
  draw,
  rounded corners=1.5pt,
  fill=white,
  minimum width=16mm,
  minimum height=8mm,
  inner sep=1pt,
  font=\scriptsize
] (Tj) at (3.25,0.00)
  {$X_R(-2^j)$};

\node[font=\scriptsize] at (3.25,-0.62)
  {$S_j$};

\node[
  draw,
  fill=white,
  minimum width=8mm,
  minimum height=8mm,
  inner sep=1pt,
  font=\small
] (Ha) at (4.75,1.20)
  {$H$};

\node[
  font=\scriptsize,
  align=center
] at (4.75,1.88)
  {create\\$c_j$};

\node[
  draw,
  rounded corners=1.5pt,
  fill=white,
  minimum width=17mm,
  minimum height=8mm,
  inner sep=1pt,
  font=\scriptsize
] (Cj) at (6.35,0.00)
  {$Z_R(\alpha_j)$};

\node[font=\scriptsize] at (6.35,-0.62)
  {$C_j$};

\node[
  draw,
  rounded corners=1.5pt,
  fill=white,
  minimum width=16mm,
  minimum height=8mm,
  inner sep=1pt,
  font=\scriptsize
] (Tjdag) at (8.10,0.00)
  {$X_R(2^j)$};

\node[font=\scriptsize] at (8.10,-0.62)
  {$S_j^\dagger$};

\node[
  draw,
  rounded corners=1.5pt,
  fill=white,
  minimum width=10mm,
  minimum height=17mm,
  inner sep=1pt,
  font=\small
] (Bj) at (9.85,0.60)
  {$B_j$};

\node[
  font=\scriptsize,
  align=center
] at (9.85,1.80)
  {uncompute\\ancilla};

\draw[line width=0.45pt]
  (ain)
  --
  (Bjdag.west |- ain);

\draw[line width=0.45pt]
  (Bjdag.east |- ain)
  --
  (Ha.west);

\draw[line width=0.45pt]
  (Ha.east)
  --
  (Bj.west |- ain);

\draw[line width=0.45pt]
  (Bj.east |- ain)
  --
  (aout);

\draw[line width=0.45pt]
  (rin)
  --
  (Bjdag.west |- rin);

\draw[line width=0.45pt]
  (Bjdag.east |- rin)
  --
  (Tj.west);

\draw[line width=0.45pt]
  (Tj.east)
  --
  (Cj.west);

\draw[line width=0.45pt]
  (Cj.east)
  --
  (Tjdag.west);

\draw[line width=0.45pt]
  (Tjdag.east)
  --
  (Bj.west |- rin);

\draw[line width=0.45pt]
  (Bj.east |- rin)
  --
  (rout);

\draw[line width=0.45pt]
  (3.25,1.20)
  --
  (Tj.north);

\fill (3.25,1.20) circle (1.6pt);

\draw[line width=0.45pt]
  (6.35,1.20)
  --
  (Cj.north);

\fill (6.35,1.20) circle (1.6pt);

\draw[line width=0.45pt]
  (8.10,1.20)
  --
  (Tjdag.north);

\fill (8.10,1.20) circle (1.6pt);

\node[
  font=\small,
  anchor=east
] at (0.10,1.20)
  {$|0\rangle_a$};

\node[
  font=\small,
  anchor=east
] at (0.10,0.00)
  {$|\overline{\ell}\rangle_R$};

\node[
  font=\small,
  anchor=west
] at (11.40,1.20)
  {$|0\rangle_a$};

\node[
  font=\small,
  anchor=west
] at (11.40,0.00)
  {$|\psi_j(\ell)\rangle_R$};

\draw[
  decorate,
  decoration={
    brace,
    mirror,
    amplitude=5pt
  },
  line width=0.45pt
]
  (0.97,-1.13)
  --
  (10.35,-1.13)
  node[
    midway,
    below=7pt,
    font=\small
  ]
  {$F_j
    =
    B_jS_j^\dagger C_j
    (H_a\otimes I_R)
    S_jB_j^\dagger$};

\end{tikzpicture}

\caption{
One stage of the gate-set compilation of the rotor momentum-code QFT.
The momentum-dependent gate $B_j^\dagger$ extracts the $j$th input bit
$p_j$ into the ancilla. The controlled shift $S_j$ removes this bit from
the rotor momentum, the Hadamard introduces the output bit $c_j$, and
$C_j$ supplies the corresponding Fourier phase. The inverse shift writes
$c_j$ into the rotor momentum, and $B_j$ returns the ancilla to
$|0\rangle_a$.
}
\label{fig:qft-one-fourier-stage}
\end{figure}
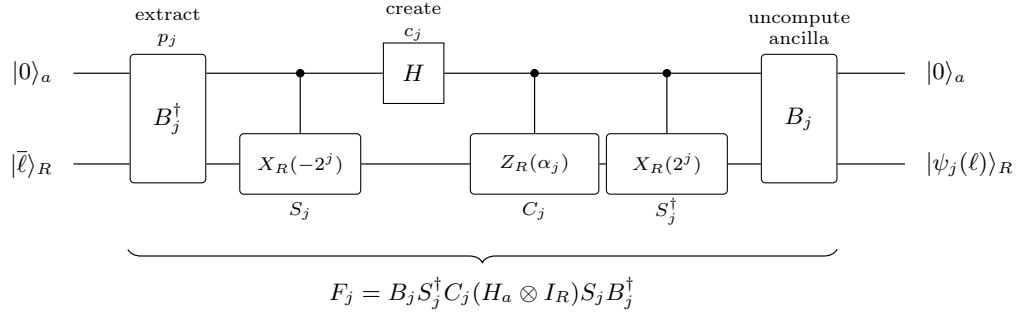

The stages are applied in the order \(j=s-1,s-2,\ldots,0\), with the
rightmost factor acting first. Fig.~\ref{fig:qft-one-fourier-stage}
summarizes their action: stage \(j\) replaces the input bit \(p_j\) by the
Fourier-output bit \(c_j\) and returns the ancilla to \(|0\rangle_a\).
The stages produce the output bits in reversed binary order.

Define the binary-reversal permutation on \(\mathcal M_d^{(0)}\) by
\begin{equation}
  \mathrm{BREV}_s \left| \overline{ \sum_{j=0}^{s-1}2^jc_j } \right\rangle_R := \left| \overline{
  \sum_{j=0}^{s-1} 2^{s-1-j}c_j } \right\rangle_R .
\end{equation}
$\mathrm{BREV}_s$ reverses the \(s\) binary positions within one rotor code.
The operation \(\mathrm{REV}_r\) reverses the order of the \(r\) rotor
registers.

Let
\(\widetilde{\mathrm{BREV}}_s\) denote the ancilla-assisted
implementation satisfying
\begin{equation}
  \widetilde{\mathrm{BREV}}_s |0\rangle_a |\overline m\rangle_R = |0\rangle_a \mathrm{BREV}_s
  |\overline m\rangle_R, \qquad 0\leq m<d .
\end{equation}
Its construction from MQR and CShift operations is given in
Appendix~\ref{app:binary-bit-reversal-one-rotor-momentum-code}.

Combining the \(s\) Fourier stages with the ancilla-assisted binary
reversal, define
\begin{equation}
  \widetilde F_d^M
  :=
  \widetilde{\mathrm{BREV}}_s
  F_0F_1\cdots F_{s-1}.
\end{equation}
Its action on the momentum code is
\begin{equation}
  \widetilde F_d^M |0\rangle_a |\overline p\rangle_R = |0\rangle_a F_d^M |\overline p\rangle_R,
  \qquad 0\leq p<d.
  \label{eq:compiled-one-rotor-qft-action}
\end{equation}
The ancilla returns to \(|0\rangle_a\) after every Fourier stage and after
the binary-reversal circuit. Every rotor momentum visited by the compiled
circuit lies in
\(\{0,\ldots,d-1\}\).

The Fourier-stage calculation and induction are given in
Appendix~\ref{app:verification-one-rotor-momentum-code-qft}.

\subsubsection{Base-\texorpdfstring{$d$}{d} QFT compilation on rotor
momentum codes}

The one-rotor compilation extends to $r$ rotor registers. This construction
uses compiled one-rotor Fourier transforms and rotor--rotor CPHS gates on
$(\mathcal M_d^{(0)})^{\otimes r}$, with both the input and the Fourier output
encoded in this momentum code.

Let $N=d^r$, and write the input integer in base $d$ as
\begin{equation}
  X = \sum_{k=1}^{r} x_kd^{r-k}, \qquad x_k\in\{0,\ldots,d-1\}.
\end{equation}
Here $x_k$ denotes the input digit and $p_j$ the Fourier-output
digit; both are encoded in the momentum code $\mathcal M_d^{(0)}$.
The digit $x_k$ is encoded in register $R_k$ as
$|\overline{x_k}\rangle_{R_k}$.

The registers are processed sequentially in the order
$R_1,R_2,\ldots,R_r$. On register $R_j$, first apply the compiled
implementation of $F_{d,j}^M$. The resulting state is expanded in
momentum-code output labels $p_j$, while every unprocessed register
$R_k$, with $k>j$, still carries its input label $x_k$. Immediately
after processing $R_j$, apply
\begin{equation}
  \mathrm{CPHS}_{jk} \left( \frac{2\pi}{d^{k-j+1}} \right) |\ell_j,\ell_k\rangle = \exp\!\left(
  \frac{2\pi i\ell_j\ell_k}{d^{k-j+1}} \right) |\ell_j,\ell_k\rangle
\end{equation}
for every $k>j$. For fixed $j$, these CPHS gates commute, so their order is
immaterial. On a component labeled by $p_j$ and $x_k$, the gate
contributes
\begin{equation}
  \exp\!\left(
    \frac{2\pi i p_jx_k}{d^{k-j+1}}
  \right),
  \label{eq:base-d-qft-cross-digit-phase}
\end{equation}
which is the corresponding cross-digit factor in the base-$d$ Fourier
kernel. Because $\mathrm{CPHS}_{jk}(\alpha)$ is a pure-rotor quadratic
Clifford gate, each cross-digit factor in
Eq.~\eqref{eq:base-d-qft-cross-digit-phase} is realized by one
rotor--rotor Clifford operation \cite{xu2024multimode}.

Let $F_N^M$ denote the standard positive-exponent $d^r$-point
Fourier transform in the standard base-$d$ basis of
$(\mathcal M_d^{(0)})^{\otimes r}$.

Let $U_{\mathrm{QFT}}$ denote the operation induced on the
rotor registers by the circuit above, with the reusable ancilla initialized
in $|0\rangle_a$ and returned to $|0\rangle_a$ after each local transform.
Its action on $(\mathcal M_d^{(0)})^{\otimes r}$ is
\begin{equation}
  U_{\mathrm{QFT}}
  =
  \mathrm{REV}_r F_N^M.
  \label{eq:compiled-base-d-qft-reversal}
\end{equation}
Here the register-order reversal is defined by
\begin{equation}
  \mathrm{REV}_r \bigotimes_{j=1}^{r} |\overline{p_j}\rangle_{R_j} = \bigotimes_{j=1}^{r}
  |\overline{p_{r-j+1}}\rangle_{R_j}.
\end{equation}
Appendix~\ref{app:verification-compiled-base-d-qft} proves
Eq.~\eqref{eq:compiled-base-d-qft-reversal} from the base-$d$ phase relation
in Eq.~\eqref{eq:app-base-d-phase-identity}.

Interpreting the output registers in reversed order requires no physical
permutation. Restoring the original physical order requires
$\lfloor r/2\rfloor$ rotor SWAP gates.

The qubit--rotor interface described at the beginning of
Sec.~\ref{subsec:qft-gate-set-compilation-rotor-momentum-codes}
may be applied separately to each rotor register when qubit input or output
is required.

\subsection{Logical resource comparison}
\label{subsec:qft-logical-resource-comparison}

To compare the two rotor constructions with the standard qubit circuit, we
specialize to
\begin{equation}
  N=d^r=2^{rs},
  \qquad
  d=2^s.
\end{equation}
All three constructions realize the same positive-exponent $N$-point
Fourier matrix using different encodings and logical operations. The standard
binary construction acts on $rs$ qubits, the
angle--momentum construction maps
$\mathcal A_d^{\otimes r}$ to $\mathcal M_d(f)^{\otimes r}$, and the compiled
momentum-code construction acts within $(\mathcal M_d^{(0)})^{\otimes r}$.

Table~\ref{tab:qft-construction-comparison} compares the operations
internal to the three encoded Fourier transforms; the qubit--rotor interface
is counted separately in
Table~\ref{tab:qft-compiled-qft-instruction-counts}.

\begin{table}[H]
  \centering
  \caption{
  Logical structure of three code-space constructions of the same
  $N=d^r=2^{rs}$ Fourier matrix. The rows list the local and cross-block
  logical operations used by each construction.
  }
  \label{tab:qft-construction-comparison}

  \renewcommand{\arraystretch}{1.2}

  \begin{tabularx}{\textwidth}{
    @{}
    L{0.19\textwidth}
    L{0.23\textwidth}
    L{0.26\textwidth}
    Y
    @{}
  }
    \toprule
    Construction
      & Encoding map
      & Local transform on each block
      & Cross-block phase operations
      \\
    \midrule

    Binary qubit construction
      &
      $\bigl(\mathbb C^2\bigr)^{\otimes rs}
       \longrightarrow
       \bigl(\mathbb C^2\bigr)^{\otimes rs}$
      &
      $r$ $s$-qubit QFTs
      &
      $s^2\binom{r}{2}$ binary controlled-phase gates
      \\

    Angle--momentum construction
      &
      $\mathcal A_d^{\otimes r}
       \longrightarrow
       \mathcal M_d(f)^{\otimes r}$
      &
      $r$ angle-to-momentum Fourier operations
      &
      $\binom{r}{2}$ mixed angle--momentum phase operations
      \\

    Compiled momentum-code construction
      &
      \mbox{$
        (\mathcal M_d^{(0)})^{\otimes r}
        \longrightarrow
        (\mathcal M_d^{(0)})^{\otimes r}
      $}
      &
      $r$ compiled one-rotor $d$-point QFTs
      &
      $\binom{r}{2}$ CPHS gates
      \\

    \bottomrule
  \end{tabularx}
\end{table}

Each QFT on an $s$-qubit block contains $s$ Hadamard gates and
$\binom{s}{2}$ controlled-phase gates. The standard binary construction
therefore uses $rs$ Hadamard gates and
\begin{equation}
  \binom{rs}{2} = r\binom{s}{2} + s^2\binom{r}{2}
\end{equation}
controlled-phase gates. The first term counts phases internal to the
$r$ blocks, and the second counts phases coupling distinct blocks.
Counting each Hadamard and controlled-phase gate once, the binary
construction has total instruction count
\begin{equation}
  N_{\mathrm{bin}}
  =
  rs+\binom{rs}{2}
  =
  \Theta((rs)^2).
\end{equation}

The angle--momentum construction uses one local angle-to-momentum Fourier
operation per rotor and one mixed angle--momentum phase per pair of rotors.
In the compiled momentum-code construction, the cross-block Fourier factor
that requires $s^2$ binary controlled-phase gates is implemented by one CPHS
gate.

Each compiled one-rotor transform uses $s$ sequential Fourier stages, the
internal binary bit reversal constructed in
Appendix~\ref{app:binary-bit-reversal-one-rotor-momentum-code}, and one reusable ancilla qubit.
The same ancilla may be reused for all $r$ local transforms. Detailed counts
are collected in Table~\ref{tab:qft-compiled-qft-instruction-counts}. Using the
logical-instruction counting convention of
Appendix~\ref{app:qft-logical-instruction-counts}, these counts give
\begin{equation}
  N_d
  =
  6s+5\left\lfloor\frac{s}{2}\right\rfloor-1,
  \qquad
  N_{\mathrm{rot}}
  =
  rN_d+\binom{r}{2}
  =
  O(rs+r^2).
  \label{eq:qft-compiled-rotor-instruction-count}
\end{equation}

For $r\geq2$, the smallest cross-register phase angle in both rotor
constructions is
\begin{equation}
  \frac{2\pi}{d^r}
  =
  \frac{2\pi}{N}.
\end{equation}
The binary QFT contains controlled-phase rotations at the same scale.
Approximate QFT circuits can reduce the binary count by truncating small-angle
rotations \cite{coppersmith2002approximate}; the same truncation principle
applies to the small conditional-phase and CPHS angles in the rotor
decomposition. Table~\ref{tab:qft-construction-comparison} compares the
transforms.

\clearpage

\section{Physical implementation considerations}
\label{sec:physical-implementation-considerations}

The logical operations used above require different physical controls in
native rotors and finite-code implementations. A native rotor supplies the
periodic angle and integer momentum directly. A finite-code implementation
reproduces selected rotor operations on a spectral subspace of another
system.

\subsection{Native-rotor and finite-code routes}

\paragraph{Native-rotor route.}
Superconducting phase--charge degrees of freedom directly realize a periodic
phase and integer-valued conjugate charge. A flux-tunable Cooper-pair-box
Hamiltonian
$H_{\mathrm{CPB}}=4E_C(\hat n-n_g)^2-E_J\cos\hat\phi$
contains the drift and cosine terms used in the universal-control construction
\cite{koch2007chargeinsensitive}. After division by $4E_C$ and removal of an
additive constant, its charging term has the form of
Eq.~\eqref{eq:biased-rotor-control-hamiltonian}, with $\hat\ell=\hat n$ and
$\xi=-2n_g$, while tuning $E_J$ supplies the cosine control. A stable offset
bias $n_g\notin\tfrac12\ZZ$ realizes the nonresonance condition $\xi\notin\ZZ$ used in Sec.~\ref{sec:universal-control-clifford-non-clifford-gates}. At a different operating point, $n_g=0$, the rescaled Hamiltonian is the Mathieu operator in Eq.~\eqref{eq:qpe-mathieu-hamiltonian}, with $\lambda=E_J/(4E_C)$. The same phase--charge variables also appear in proposals for rotor codes
\cite{vuillot2024homological}, and in recent superconducting implementations of compact $U(1)$ gauge fields coupled to matter
\cite{alcainecuervo2026compact}.

Trapped planar rotors provide a direct realization of the U(1) degree of freedom, with a periodic angle and integer-valued angular momentum \cite{urban2019coherent,glikin2025rotational,leibscher2025planar}. Molecular ions and ultracold molecules have larger rotational state spaces, but resolved angular-momentum manifolds and coherent internal-state control allow selected levels to be used as finite rotor codes \cite{chou2017preparation,lin2020quantum,hepworth2025longlived}. Coherent translation on a momentum lattice, another operation used in the constructions above, has also been demonstrated in atom-optics kicked rotors and momentum-space quantum walks \cite{moore1995atom,dadras2019momentumspace}.

\paragraph{Finite-code route.}
Oscillator platforms can reproduce selected rotor operations on finite codes. Dispersive cavity and trapped-ion controls provide number-selective qubit operations, conditional displacements, and nonlinear bosonic interactions \cite{heeres2015cavity,eickbusch2022fast,diringer2024conditional}. Trapped-ion implementations of trigonometric continuous-variable gates also include the periodic cosine control used in Sec.~\ref{sec:universal-control-clifford-non-clifford-gates} \cite{rainaldi2026trigonometric}. In this setting, the rotor is represented only on a chosen finite spectral subspace; the full oscillator need not reproduce the rotor algebra outside that code. The additional requirements are preparation of the encoded states, spectrally selective control within the code, and suppression of leakage outside it.

\subsection{Physical requirements for the gate set}

Table~\ref{tab:implementation-capabilities} summarizes the physical
capabilities associated with the logical operations used throughout the
paper.

\begin{table}[H]
  \centering
  \caption{Physical requirements and main uses of the logical operations in
  the paper.}
  \label{tab:implementation-capabilities}
  \small
  \begin{tabularx}{\linewidth}{
    @{}L{0.23\linewidth}L{0.47\linewidth}Y@{}
  }
    \toprule
    Logical operation
      & Required physical capability
      & Main use \\
    \midrule

    Momentum shifts and momentum-diagonal phase gates
      & Integer translations of rotor momentum and evolution generated by
        $\hat\ell$ or $\hat\ell^2$
      & Rotor Clifford circuits; gauge dynamics \\

    Periodic angle potentials
      & Evolution generated by $V(\hat\theta)$, including the cosine gate
        used in the universal-control construction
      & Universal control; gauge dynamics \\

    Momentum-lattice automorphisms $U_A$
      & Integer linear transformations of the rotor momentum coordinates
      & Multirotor Clifford synthesis \\

    $\Gamma_{c,j}$ and MQR
      & Qubit operations conditioned on rotor parity or on selected momentum
        values or residue classes
      & Mixed Clifford gates; coherent transfer; QFT compilation \\

    CShift and $\mathcal O_U$
      & Qubit-conditioned momentum shift or target operations conditioned
        on rotor momentum
      & Gauge hopping; phase estimation; coherent transfer \\

    $\mathrm{CPHS}_{ij}(\alpha)$
      & Evolution generated by an interaction proportional to
        $\hat\ell_i\hat\ell_j$
      & Cross-register Fourier phases \\

    Code preparation and readout
      & Preparation of finite momentum superpositions and measurement of
        rotor momentum or angle
      & Phase estimation; encoded Fourier processing \\

    \bottomrule
  \end{tabularx}
\end{table}

Momentum-resolved qubit control gives an oscillator analogue of MQR
\cite{krastanov2015universal,heeres2015cavity,eickbusch2022fast}, and
conditional displacements implement the qubit-to-mode control used by CShift
\cite{diringer2024conditional}. The two-rotor interaction required by CPHS is
proportional to $\hat\ell_i\hat\ell_j$; its oscillator analogue is a
number--number or cross-Kerr interaction
\cite{kounalakis2018tuneable,holland2015singlephotonresolved}.

\subsection{Control scales and finite-window diagnostics}

For the compiled QFT, the code size $d=2^s$ fixes both the momentum range
and the phase resolution. We use the momentum code
\begin{equation*}
  \mathcal M_d^{(0)}
  =
  \operatorname{span}\{|0\rangle_R,\ldots,|d-1\rangle_R\},
  \qquad
  d=2^s.
\end{equation*}
Every intermediate momentum reached by the circuit remains in
$\{0,\ldots,d-1\}$. The Fourier stages use qubit-controlled shifts with
distances at most $2^{s-1}$ and MQR operations that resolve individual
binary digits of the momentum label. The binary-reversal circuit also uses
MQR operations conditioned on selected two-bit patterns. The smallest local
conditional phase used in the compilation is $2\pi/d$. For $r$ rotor
registers, with $N=d^r$, the smallest cross-register CPHS angle is
\begin{equation*}
  \frac{2\pi}{N}
  =
  \frac{2\pi}{d^r}.
\end{equation*}

Table~\ref{tab:qft-finite-code-resource-sheet} gives an illustrative example
with $r=2$, $s=3$, $d=8$, and $N=64$. The instruction counts follow
Eq.~\eqref{eq:qft-compiled-rotor-instruction-count} and the convention of
Appendix~\ref{app:qft-logical-instruction-counts}.

\begin{table}[H]
  \centering
  \caption{Illustrative control requirements and instruction counts for the
  compiled two-rotor $N=64$ QFT with $r=2$, $s=3$, and $d=8$.
  Counts follow the logical-instruction convention of
  Appendix~\ref{app:qft-logical-instruction-counts}.}
  \label{tab:qft-finite-code-resource-sheet}
  \small
  \begin{tabularx}{0.86\linewidth}{@{}L{0.50\linewidth}Y@{}}
    \toprule
    Requirement & Target \\
    \midrule
    Momentum range per rotor
      & $\ell=0,\ldots,7$ \\
    Maximum CShift distance
      & $4$ \\
    Momentum selectivity (MQR)
      & Individual bits and selected two-bit patterns of
        $\ell\in\{0,\ldots,7\}$ \\
    Smallest local conditional phase used
      & $\pi/4$ \\
    Cross-register CPHS phase
      & $\pi/32$ \\
    Compiled-QFT instruction count
      & $N_d=22$ per rotor; $N_{\mathrm{rot}}=45$ total \\
    \bottomrule
  \end{tabularx}
\end{table}

The phase-estimation protocol imposes a different finite-range requirement.
For a probe with finite momentum support, the largest occupied $|\ell|$
fixes the largest controlled power $U^{-\ell}$ that must be implemented.
Increasing the support therefore increases the range of controlled powers
required by the protocol.

For the angle-code QFT construction, the relevant finite-code errors are
those associated with the realization of the mixed angle--momentum phase.
The code-space error and leakage are bounded by
Eqs.~\eqref{eq:app-qft-mixed-phase-multiplication-error}
and~\eqref{eq:app-qft-mixed-phase-multiplication-leakage}, respectively.
Both bounds are controlled by the weight of the angle-code wavefunctions
outside their assigned angular cores.

The gauge calculation uses the hard-wall electric-flux cutoff
$|E_e|\leq L$. The link-averaged boundary occupation
$\overline p_{\partial L}$ in
Eq.~\eqref{eq:gauge-boundary-occupation-diagnostic}
measures the weight on the flux states where the projected shift differs
from the full rotor shift. Its decrease with $L$ therefore provides a direct
diagnostic of the finite-flux truncation.

\section{Conclusion and outlook}

This paper has treated the $U(1)$ rotor as a computational register in its own
right, coupled directly to qubits. The Clifford theory that results is directional 
but constructive: rotor momentum parity may control qubit Pauli operations while 
the reverse is forbidden, every commutation-preserving phase-space action has an 
explicit finite circuit realization, and the mixed part of that circuit has an 
exact minimal gate count. The same analysis gives complete Clifford criteria for 
the momentum- and angle-dependent gate families used throughout the paper. Adding 
two non-Clifford controls promotes the local Clifford operations to universal
control on the full Hilbert space in the strong operator topology.

The applications use this structure in different ways. For compact $U(1)$
gauge--matter dynamics, we constructed an exact full-space realization of
gauge-covariant hopping and studied the effect of a hard-wall electric-flux
cutoff on ground-state and real-time observables. For rotor phase estimation,
we used the duality $\ZZ\longleftrightarrow\TT$ to convert phase kickback into
a translation of the measured angle distribution; this removes the need for a
coherent inverse QFT on the phase register and turns probe preparation into an
optimization problem under momentum-support or mean-energy constraints. For
Fourier processing, we gave both an angle-to-momentum code construction and an
explicit gate-set compilation within a finite momentum code, in which
quadratic rotor interactions implement the cross-register Fourier phases.

We also identified the physical capabilities these logical constructions
require. A native phase--charge rotor supplies the integer momentum, periodic
angle, quadratic drift, and cosine potential directly. A finite-code
implementation can reproduce the same logical operations through
momentum-selective qubit control, conditional displacements, and nonlinear
interactions on a chosen spectral subspace. In both routes the rotor enters
through its full phase-space structure, not only as a large code space.

Two limitations remain. Strong-operator density is an existence statement and
yields no bound on circuit length or control time. The Fourier resource counts
are logical instruction counts, not physical gate counts; they assign unit
cost to MQR, CShift, conditional-phase, CPHS, and Hadamard operations, and
they exclude state preparation, spectral selectivity, interaction time, phase
resolution, readout, and leakage.
Sec.~\ref{sec:physical-implementation-considerations} identifies where these
requirements enter the two implementation routes, but does not fix their
quantitative cost.

The main remaining problem is therefore quantitative compilation.
Finite-window synthesis bounds for CShift, MQR, and $\mathcal O_U$ would
connect the strong-density result to circuit length and control time, and
would determine whether the logical compression of the cross-register Fourier
phases survives once spectral selectivity and small interaction angles are
assigned physical costs. A second problem is algebraic: hard-wall projection
makes the rotor shift nonunitary while preserving the finite Gauss-law basis,
while cyclic completion restores unitarity at the price of changing the shift
relation at the boundary, so the Clifford and normalizer structures of these
finite models need not coincide with those of the full rotor. Understanding
that change bears directly on gate synthesis and on the boundary between
classically simulable hybrid normalizer circuits and circuits containing
additional non-Clifford operations. Beyond $U(1)$, compact Abelian groups are
the next setting for the structural theory, followed by non-Abelian compact
groups, where continuous group variables are accompanied by discrete
representation data.

\section*{Acknowledgments}

We thank Joel Bierman, Weijian Chen, Raghav Jha, Yuan Liu, and Grant Scotto for 
fruitful discussions. This work was supported by the U.S. Department of Energy, 
Advanced Scientific Computing Research, under contract number DE-SC0025384.

\clearpage
\appendix

\section{Proofs for the hybrid Clifford classification}

This appendix gives the block calculation underlying the hybrid phase-space
classification in Theorem~\ref{thm:hybrid-phase-space-classification} and computes the
one-qubit--one-rotor symplectic group and Clifford quotient. The explicit
unitary realization, the kernel of the induced Clifford action, and the
classification of mixed blocks under separate qubit and rotor Clifford
operations are proved in the main text.

\subsection{Block classification of hybrid phase-space transformations}
\label{app:block-classification-hybrid-phase-space-transformations}

Set $V=\FF_2^{2n}$, and let
$\rho_2:\ZZ^r\to\FF_2^r$ denote reduction modulo two. Let
$S:K_{n,r}\to K_{n,r}$ be a group automorphism preserving $\omega$.
With respect to the decomposition
$
  K_{n,r}=V\oplus\ZZ^r\oplus\TT^r,
$
write $S$ in block form as
\begin{equation}
  \begin{pmatrix}
    q'\\
    m'\\
    \phi'
  \end{pmatrix}
  =
  \begin{pmatrix}
    S_{VV} & S_{V\ZZ} & S_{V\TT}\\
    S_{\ZZ V} & S_{\ZZ\ZZ} & S_{\ZZ\TT}\\
    S_{\TT V} & S_{\TT\ZZ} & S_{\TT\TT}
  \end{pmatrix}
  \begin{pmatrix}
    q\\
    m\\
    \phi
  \end{pmatrix}.
\end{equation}
Each row specifies a target component, and each column specifies a source
component. We first use the group structures of $V$, $\ZZ^r$, and $\TT^r$
to determine the possible nonzero blocks, and then impose preservation of
$\omega$.

The group $\TT^r$ is divisible. The groups $V$ and $\ZZ^r$ contain no
nonzero divisible subgroups. In addition, $V$ is torsion and $\ZZ^r$ is
torsion-free. Hence
\begin{equation}
  \Hom(\TT^r,V)=0, \qquad \Hom(\TT^r,\ZZ^r)=0, \qquad \Hom(V,\ZZ^r)=0.
\end{equation}
Every homomorphism $\ZZ^r\to V$ factors through $\rho_2$; every
homomorphism $V\to\TT^r$ has image in $\{0,\pi\}^r$. The remaining
blocks therefore have the form
\begin{equation}
  \begin{pmatrix}
    q'\\
    m'\\
    \phi'
  \end{pmatrix}
  =
  \begin{pmatrix}
    M & C\rho_2 & 0\\
    0 & A & 0\\
    \pi G & B & \alpha
  \end{pmatrix}
  \begin{pmatrix}
    q\\
    m\\
    \phi
  \end{pmatrix},
\end{equation}
or, in component form,
\begin{equation}
  \begin{split}
    q'&=Mq+C\bar m,\\
    m'&=Am,\\
    \phi'&=\alpha(\phi)+Bm+\pi Gq.
  \end{split}
  \label{eq:hybrid-parameter-component-form}
\end{equation}
Here $M$, $C$, and $G$ are matrices over $\FF_2$, $A$ is an integer
matrix, and $B$ has entries in $\TT$.

The subgroup $\TT^r$ is the unique maximal divisible subgroup of
$K_{n,r}$ and is therefore preserved by $S$. In the quotient
$K_{n,r}/\TT^r$, the subgroup $V$ is precisely the torsion subgroup and
is also preserved. It follows that $S$ induces automorphisms on the
three diagonal components, so
\begin{equation}
  M\in\GL(2n,\FF_2), \qquad A\in\GL(r,\ZZ), \qquad \alpha\in\operatorname{Aut}(\TT^r).
\end{equation}

We now impose preservation of $\omega$ on four types of inputs.

\paragraph{Qubit--qubit inputs.}
For $q,p\in V$, preservation of $\omega$ on
$(q,0,0)$ and $(p,0,0)$ gives
\begin{equation}
  M^{\transpose}J_nM=J_n.
\end{equation}
Thus $M\in\Sp(2n,\FF_2)$.

\paragraph{Torus--integer inputs.}
For $\phi\in\TT^r$ and $k\in\ZZ^r$, preservation of $\omega$ on
$(0,0,\phi)$ and $(0,k,0)$ gives
$\exp\!\left(
    i\alpha(\phi)^{\transpose}Ak
  \right)
  =
  \exp\!\left(
    i\phi^{\transpose}k
  \right) $ for all $k\in\ZZ^r$.
Taking $k$ to be the standard basis vectors of $\ZZ^r$ gives
$A^{\transpose}\alpha(\phi)=\phi$ in $\TT^r$.
Therefore
\begin{equation}
  \alpha(\phi)=A^{-\transpose}\phi.
  \label{eq:hybrid-block-torus-condition}
\end{equation}

\paragraph{Qubit--integer inputs.}
For $q\in V$ and $k\in\ZZ^r$,
$
  \omega\bigl((q,0,0),(0,k,0)\bigr)=1.
$
Preservation of $\omega$ therefore requires
$
  M^{\transpose}J_nC+G^{\transpose}\bar A=0.
$
Since $\bar A$ is invertible over $\FF_2$, this gives
\begin{equation}
  G = \bar A^{-\transpose} C^{\transpose}J_nM.
  \label{eq:hybrid-block-reciprocal-condition}
\end{equation}

\paragraph{Integer--integer inputs.}
For $m,k\in\ZZ^r$, preservation of $\omega$ on
$(0,m,0)$ and $(0,k,0)$ gives
\begin{equation}
  A^{\transpose}B-B^{\transpose}A
  =\pi C^{\transpose}J_nC.
  \label{eq:hybrid-block-shear-condition}
\end{equation}
Define $\Theta=A^{\transpose}B$.  Then $B=A^{-\transpose}\Theta$ and
Eq.~\eqref{eq:hybrid-block-shear-condition} becomes
Eq.~\eqref{eq:theta-constraint}.

For every $v\in\FF_2^r$, $ A^{-\transpose}(\pi v) = \pi\bar A^{-\transpose}v$ in $\TT^r.$
Substitution of
Eqs.~\eqref{eq:hybrid-block-torus-condition},
\eqref{eq:hybrid-block-reciprocal-condition}, and
\eqref{eq:hybrid-block-shear-condition} into
Eq.~\eqref{eq:hybrid-parameter-component-form} gives
Eq.~\eqref{eq:general-hybrid-map}.

The pairing is already trivial on
$V\times\TT^r$ and $\TT^r\times\TT^r$, both before and after applying
$S$. Reversing the order of the two inputs only inverts the resulting
phase. The four cases above therefore cover all pairs of components.  
Conversely, direct substitution shows that $M$ preserves the qubit--qubit term,
$\alpha=A^{-T}$ preserves the torus--integer term, the block $G$ cancels the
qubit--integer terms, and Eq.~\eqref{eq:theta-constraint} cancels the extra 
integer--integer term. Hence every map in Eq.~\eqref{eq:general-hybrid-map} 
satisfying Eq.~\eqref{eq:theta-constraint} preserves $\omega$.

Such a map is invertible. From $(q',m',\phi')$, the input is recovered
successively as
\begin{equation}
  m=A^{-1}m', \qquad q=M^{-1}(q'-C\bar m), \qquad
  \phi=A^{\transpose}\phi'-\Theta m -\pi C^{\transpose}J_nMq.
\end{equation}
All maps whose domains are $V$ or $\ZZ^r$ are continuous because these
groups are discrete. The formula $\alpha=A^{-\transpose}$ and the inverse
above show that $S$ and $S^{-1}$ are continuous.

Finally, the action of $S$ on inputs supported in each of the three
components determines $M$, $C$, $A$, $B$, $G$, and $\alpha$ uniquely,
and $\Theta=A^{\transpose}B$. This proves
Theorem~\ref{thm:hybrid-phase-space-classification}.

\paragraph{Composition and inverse.}
With the convention $S_1S_2=S_1\circ S_2$, direct substitution gives,
for $S_i=S_{M_i,A_i,C_i,\Theta_i}$,
\begin{equation}
  \begin{array}{r@{}l@{\qquad}r@{}l}
    M_{12}&=M_1M_2,
    &A_{12}&=A_1A_2,\\
    C_{12}&=M_1C_2+C_1\bar A_2,
    &\Theta_{12}
      &=\Theta_2+A_2^{\transpose}\Theta_1A_2
        +\pi\bar A_2^{\transpose}
        C_1^{\transpose}J_nM_1C_2.
  \end{array}
  \label{eq:hybrid-symplectic-map-composition}
\end{equation}
The inverse has parameters
\begin{equation}
  \begin{array}{r@{}l@{\qquad}r@{}l}
    \widetilde M&=M^{-1},
    &\widetilde A&=A^{-1},\\
    \widetilde C&=M^{-1}C\bar A^{-1},
    &\widetilde\Theta
      &=-A^{-\transpose}\Theta A^{-1}
        +\pi\widetilde C^{\transpose}J_n\widetilde C.
  \end{array}
\end{equation}

\subsection{The one-qubit--one-rotor quotient}
\label{app:one-qubit-one-rotor-quotient}

\begin{proposition}[One-qubit--one-rotor symplectic group]
For one qubit and one rotor,
\begin{equation}
  \Sp(K_{1,1},\omega)
  \cong
  \TT\times S_4\times\ZZ_2.
  \label{eq:one-qubit-one-rotor-symplectic-group}
\end{equation}
Consequently,
$
  \Cl_{1,1}/\Pauli_{1,1}
  \cong
  \TT\times S_4\times\ZZ_2.
$
\end{proposition}

\begin{proof}
For $n=r=1$, the parameters in
Theorem~\ref{thm:hybrid-phase-space-classification} reduce to
$
  M\in\Sp(2,\FF_2)=\GL(2,\FF_2),
  c\in\FF_2^2,
  \epsilon\in\{\pm1\},
  \theta\in\TT.
$
Since the form defined by $J_1$ is alternating,
$c^{\transpose}J_1c=0$. The constraint on $\theta$ is therefore automatically
satisfied for every $\theta\in\TT$. Using the convention
$S_1S_2=S_1\circ S_2$, the composition law is
\begin{equation}
  (M,c,\epsilon,\theta) (N,d,\delta,\eta)
  = \left( MN, c+Md, \epsilon\delta,
  \theta+\eta+\pi c^{\transpose}J_1Md
  \right).
  \label{eq:one-qubit-one-rotor-product}
\end{equation}

Define $\nu:\FF_2^2\to\FF_2$ by
$
  \nu(0)=0,
  \nu(c)=1 (c\ne0).
$
Every $M\in\GL(2,\FF_2)$ permutes the three nonzero vectors, and direct
calculation gives
\begin{equation}
  \nu(Mc)=\nu(c),
  \qquad
  \nu(c)+\nu(d)+\nu(c+d)=c^{\transpose}J_1d
  \quad\text{in }\FF_2.
\end{equation}
Define the shifted torus coordinate
\nopagebreak[4]
$
  \widetilde\theta=\theta+\pi\nu(c).
$
With this redefined parameter, the final phase term in
Eq.~\eqref{eq:one-qubit-one-rotor-product} is absorbed, and the product
becomes
\begin{equation}
  (M,c,\epsilon,\widetilde\theta)(N,d,\delta,\widetilde\eta)= \left( MN, c+Md, \epsilon\delta,
  \widetilde\theta+\widetilde\eta \right).
  \label{eq:one-qubit-one-rotor-untwisted-product}
\end{equation}
The pair $(M,c)$ has the product law of
$\operatorname{AGL}(2,\FF_2)$.  Its action on the four points of $\FF_2^2$
is faithful, and
$
  |\operatorname{AGL}(2,\FF_2)|
  =4|\GL(2,\FF_2)|=24.
$
Hence $\operatorname{AGL}(2,\FF_2)\cong S_4$.  Equation
\eqref{eq:one-qubit-one-rotor-untwisted-product} proves
Eq.~\eqref{eq:one-qubit-one-rotor-symplectic-group}, and
Corollary~\ref{cor:clifford-symplectic-exact-sequence} gives the Clifford quotient.
\end{proof}

\section{Proofs of the gate-family criteria}

This appendix proves the classification results for momentum-diagonal phase
gates and momentum-dependent qubit gates. Both proofs use conjugation of the
elementary momentum shift. In the first case, this gives a scalar finite-difference
equation; in the second, it gives an operator-valued finite-difference
equation. The remaining Clifford criteria in
Sec.~\ref{sec:gate-families} are proved directly in the main text.

\subsection{Momentum-diagonal phase gates}
\label{app:momentum-diagonal-phase-gate-proof}

\begin{proof}[Proof of Proposition~\ref{prop:momentum-diagonal-phases}]
Let $e_j$ be the $j$th standard basis vector of $\ZZ^r$. Conjugating the
corresponding momentum shift gives
\begin{equation}
  G_uX_R(e_j)G_u^\dagger|\ell\rangle
  =
  \frac{u(\ell+e_j)}{u(\ell)}
  |\ell+e_j\rangle.
\end{equation}
If $G_u$ is Clifford, the coefficient multiplying the shift must, up to a
constant phase, be a character of $\ZZ^r$. Hence there exist
$\rho_j\in\TT$ and $r_j\in\TT^r$ such that
\begin{equation}
  \frac{u(\ell+e_j)}{u(\ell)}
  =
  e^{i\rho_j+ir_j^{\transpose}\ell}.
  \label{eq:momentum-diagonal-phase-gate-recurrence}
\end{equation}

We now compare the value of $u(\ell+e_i+e_j)$ obtained by applying
Eq.~\eqref{eq:momentum-diagonal-phase-gate-recurrence} first in the
$e_i$-direction and then in the $e_j$-direction with the value obtained in
the reverse order. This gives
$
  (r_j)_i=(r_i)_j
$ in $\TT$.
Thus the matrix
$
  R=(r_1,\ldots,r_r)
$
is symmetric.

From the definition of $f_R$ in
Eq.~\eqref{eq:rotor-quadratic-phase},
\begin{equation}
  f_R(\ell+e_j)-f_R(\ell)
  =
  r_j^{\transpose}\ell+R_{jj}.
  \label{eq:quadratic-phase-finite-difference}
\end{equation}
Set
$
  \lambda_j=\rho_j-R_{jj}
$
and define
\begin{equation}
  h(\ell)
  =
  u(\ell)
  e^{-i\lambda^{\transpose}\ell-if_R(\ell)}.
\end{equation}
Eqs.~\eqref{eq:momentum-diagonal-phase-gate-recurrence} and
\eqref{eq:quadratic-phase-finite-difference} imply
$ h(\ell+e_j)=h(\ell)$ for every $\ell\in\ZZ^r $ and every $j.$
Replacing $\ell$ by $\ell-e_j$ gives the same equality for a shift by
$-e_j$. Since any two points of $\ZZ^r$ are connected by such unit shifts,
$h$ is constant on $\ZZ^r$. Writing $h(0)=e^{i\gamma}$ gives
\begin{equation}
  u(\ell)
  =
  e^{i\gamma+i\lambda^{\transpose}\ell+if_R(\ell)}.
  \label{eq:momentum-diagonal-phase-gate-lattice-solution}
\end{equation}

Conversely, if $u$ has the form in
Eq.~\eqref{eq:momentum-diagonal-phase-gate-lattice-solution}, then
$
  G_u
  =
  e^{i\gamma}Z_R(\lambda)Q_R.
$
The gates $Z_R(\lambda)$ and $Q_R$ are pure-rotor Clifford gates by
Sec.~\ref{subsec:pure-rotor-mixed-clifford-gates}. Therefore
$G_u\in\Cl_{0,r}$.
\end{proof}

\subsection{Momentum-dependent qubit gates}
\label{app:momentum-dependent-qubit-gate-proof}

\begin{proof}[Proof of Proposition~\ref{prop:momentum-dependent-qubit-gates}]
For brevity, write $U=\mathcal V$. Conjugating the elementary rotor
shift gives
\begin{equation}
  U(I\otimes X_R(1))U^\dagger = \sum_{\ell\in\ZZ} V_{\ell+1}V_\ell^\dagger \otimes|\ell+1\rangle\langle\ell|.
\end{equation}
If $U$ is Clifford, the right-hand side must be proportional to a hybrid
Weyl operator with momentum shift $1$. Its qubit part must therefore be a fixed
Pauli operator. Its scalar dependence on $\ell$ must be a character of $\ZZ$.
Hence there exist
$c\in\FF_2^{2n}$ and $\mu,\alpha\in\TT$ such that
\begin{equation}
  V_{\ell+1}V_\ell^\dagger
  =
  e^{i\mu+i\alpha\ell}P_c.
  \label{eq:momentum-dependent-qubit-relative-recurrence}
\end{equation}

Applying Eq.~\eqref{eq:momentum-dependent-qubit-relative-recurrence}
successively gives
\begin{equation}
  V_\ell
  =
  e^{i\mu\ell+i\alpha\ell(\ell-1)/2}
  P_c^{\,\ell}V_0.
\end{equation}
For negative $\ell$, the same formula follows by solving
Eq.~\eqref{eq:momentum-dependent-qubit-relative-recurrence} for
$V_\ell$ in terms of $V_{\ell+1}$.

Set $\lambda=\mu-\alpha$.
Then
\begin{equation}
  V_\ell
  =
  e^{i[\lambda\ell+\alpha\ell(\ell+1)/2]}
  P_c^{\,\ell}V_0.
  \label{eq:momentum-dependent-qubit-gates-solution}
\end{equation}
It remains to show that $V_0$ is a qubit Clifford operator.

For every qubit
Pauli operator $P_b$,
\begin{equation}
  U(P_b\otimes I_R)U^\dagger = \sum_{\ell\in\ZZ} V_\ell P_bV_\ell^\dagger \otimes|\ell\rangle\langle\ell|.
\end{equation}
Since $U$ is Clifford, this operator must be proportional to a hybrid Weyl
operator. It preserves each momentum subspace, so its rotor-shift component
must be zero. Evaluating the qubit block at $\ell=0$ shows that
$
  V_0P_bV_0^\dagger
$
is proportional to a Pauli operator. This holds for every $P_b$, and hence
$V_0$ normalizes the qubit Pauli group. Therefore
$V_0\in\Cl_{n,0}$.

Conversely, suppose that
Eq.~\eqref{eq:momentum-dependent-qubit-gates-solution} holds with
$V_0\in\Cl_{n,0}$. Let $Q_{(\alpha)}$ be the one-rotor quadratic gate
corresponding to the $1\times1$ matrix $R=(\alpha)$. Acting on each momentum
subspace gives
\begin{equation}
  \mathcal V = \bigl[I_Q\otimes Z_R(\lambda)Q_{(\alpha)}\bigr]
  \Gamma_{c,1}(V_0\otimes I_R).
\end{equation}
Here
$Z_R(\lambda)Q_{(\alpha)}$ is a pure-rotor Clifford gate,
$\Gamma_{c,1}$ is a mixed Clifford gate, and
$V_0$ is a qubit Clifford gate. Therefore
$\mathcal V\in\Cl_{n,1}$.
\end{proof}

\section{Proofs for universal control}
\label{app:proofs-universal-control}

As in Sec.~\ref{subsec:strong-topology-finite-window}, all strong
closures below are taken in $\UU(\bH)$.  Multiplication and inversion are
continuous in the strong operator topology on $\UU(\bH)$, so the strong
closure of a unitary subgroup is again a subgroup.

\subsection{Finite-window approximation and strong density}

Let $P_L$ and $\bH_L$ be as in
Lemma~\ref{lem:finite-window-strong-density}.

\begin{lemma}[Finite-window unitary approximation]
\label{lem:finite-window-interpolation}
For every $V\in\UU(\bH)$, every finite set of vectors
$\psi_1,\ldots,\psi_s\in\bH$, and every $\delta>0$, there exist
$L$ and $V_L\in\UU(\bH_L)$ such that
\begin{equation}
  \|V_LP_L\psi_j-P_LV\psi_j\|<\delta,
  \qquad
  j=1,\ldots,s,
\end{equation}
and
\begin{equation}
  \|(I-P_L)\psi_j\|<\delta,
  \qquad
  \|(I-P_L)V\psi_j\|<\delta.
  \label{eq:finite-window-tail-bounds}
\end{equation}
\end{lemma}

\begin{proof}
Let
$
  \mathcal S=\operatorname{span}\{\psi_1,\ldots,\psi_s\}.
$
Since $\mathcal S$ and $V\mathcal S$ are finite-dimensional and $P_L\to I$ strongly, the
convergence is uniform on the unit spheres of both subspaces. Hence, for
sufficiently large $L$, the restrictions of $P_L$ to $\mathcal S$ and
$V\mathcal S$ are
injective.

Define
\begin{equation}
  T_L:P_L\mathcal S\longrightarrow P_LV\mathcal S,
  \qquad
  T_L(P_Lx)=P_LVx.
\end{equation}
This map is well defined because the restriction of $P_L$ to $\mathcal S$ is
injective. The uniform convergence on $\mathcal S$ and $V\mathcal S$ implies
\begin{equation}
  \left\|
    T_L^\dagger T_L-I_{P_L\mathcal S}
  \right\|
  \longrightarrow0.
\end{equation}

Let
$
  T_L=Q_L|T_L|
$
be the polar decomposition of $T_L$. Since $T_L$ is an invertible map
from $P_L\mathcal S$ onto $P_LV\mathcal S$, the operator $Q_L$ is unitary between these
two subspaces. Moreover,
\begin{equation}
  \|Q_L-T_L\|
  =
  \bigl\||T_L|-I_{P_L\mathcal S}\bigr\|
  \longrightarrow0.
\end{equation}

The subspaces $P_L\mathcal S$ and $P_LV\mathcal S$ have the same dimension. Their
orthogonal complements in $\bH_L$ therefore also have the same dimension,
so $Q_L$ extends to a unitary
$
  V_L\in\UU(\bH_L).
$
For each $j$,
$
  \|V_LP_L\psi_j-P_LV\psi_j\|
  =
  \|(Q_L-T_L)P_L\psi_j\|,
$
which is smaller than $\delta$ for sufficiently large $L$. Increasing
$L$ further if necessary also gives
Eq.~\eqref{eq:finite-window-tail-bounds}.
\end{proof}

\begin{proof}[Proof of Lemma~\ref{lem:finite-window-strong-density}]
Fix $V\in\UU(\bH)$, vectors $\psi_1,\ldots,\psi_s\in\bH$, and
$\epsilon>0$. Set
\begin{equation}
  M=\max\{1,\|\psi_1\|,\ldots,\|\psi_s\|\},
  \qquad
  \delta=\frac{\epsilon}{8M}.
\end{equation}
Choose $L$ and $V_L\in\UU(\bH_L)$ from
Lemma~\ref{lem:finite-window-interpolation}. By the hypothesis of
Lemma~\ref{lem:finite-window-strong-density}, there exists
$W\in\overline G^{\,\mathrm s}$ preserving $\bH_L$ and satisfying
$\|W|_{\bH_L}-V_L\|<\delta$. Since $W$ and $V$ are unitary,
\begin{equation}
\begin{aligned}
  \|(W-V)\psi_j\|
  &\leq \|(I-P_L)\psi_j\|
  +\|WP_L\psi_j-V_LP_L\psi_j\|
  +\|V_LP_L\psi_j-P_LV\psi_j\|
  +\|(I-P_L)V\psi_j\|\\
  &<3\delta+\delta\|\psi_j\|
  \leq4\delta M
  =\frac{\epsilon}{2}.
\end{aligned}
\end{equation}
Because $W\in\overline G^{\,\mathrm s}$, there exists $g\in G$ such that
$\|(g-W)\psi_j\|<\epsilon/2$ for every $j$. Therefore
$\|(g-V)\psi_j\|<\epsilon$ for every $j$, proving that $G$ is strongly
dense in $\UU(\bH)$.
\end{proof}

\subsection{Spectral conditions for single-rotor universal control}

For $H_{\mathrm b}$ in
Eq.~\eqref{eq:biased-rotor-control-hamiltonian}, the domain in the
momentum basis is
\begin{equation}
  \mathcal D(H_{\mathrm b}) = \left\{ \psi=\sum_{\ell\in\ZZ}\psi_\ell|\ell\rangle: \sum_{\ell\in\ZZ}
  |\ell^2+\xi\ell|^2|\psi_\ell|^2<\infty \right\}.
\end{equation}
The control operator $\cos\hat\theta$ is bounded and self-adjoint, with matrix
elements
\begin{equation}
  \langle k|\cos\hat\theta|\ell\rangle
  =\frac{1}{2}\left(\delta_{k,\ell+1}+\delta_{k,\ell-1}\right).
\end{equation}
Thus the momentum states form a complete eigenbasis of the drift, and
the control couples only neighboring momentum states.

To match the notation of
Ref.~\cite{boscain2012spectral}, set
$
  A_{\mathrm c}=-iH_{\mathrm b},
  B_{\mathrm c}=-i\cos\hat\theta,
$
and take the control values in a fixed interval $[0,\delta]$.
Since $\cos\hat\theta$ is bounded and self-adjoint,
$H_{\mathrm b}+u\cos\hat\theta$ is self-adjoint on
$\mathcal D(H_{\mathrm b})$ for every $u\in[0,\delta]$.
Its restriction to finite linear combinations of momentum states has
this operator as its self-adjoint closure. These facts verify the basis
and domain assumptions of the cited controllability theorem
\cite{reed1980functional,boscain2012spectral}.

The drift eigenvalues are
$
  \varepsilon_\ell=\ell^2+\xi\ell.
$
If $\varepsilon_\ell=\varepsilon_k$, then
$
  (\ell-k)(\ell+k+\xi)=0.
$
For distinct integers $\ell$ and $k$, the second factor cannot vanish
because $\xi\notin\ZZ$. Hence the spectrum is nondegenerate.

The control couples the transitions
$
  |\ell\rangle
  \longleftrightarrow
  |\ell+1\rangle,
$
whose transition frequencies are
$
  |\varepsilon_{\ell+1}-\varepsilon_\ell|
  =
  |2\ell+1+\xi|.
$
Suppose that two such frequencies are equal:
$
  |2\ell+1+\xi|
  =
  |2k+1+\xi|.
$
Either
$
  2\ell+1+\xi
  =
  2k+1+\xi,
$
which gives $\ell=k$, or
$
  2\ell+1+\xi
  =
  -(2k+1+\xi),
$
which requires
$
  \ell+k+1+\xi=0.
$
The latter is impossible because $\xi\notin\ZZ$. Therefore the
nearest-neighbor transition frequencies are pairwise distinct.

These nearest-neighbor transitions connect all momentum states
$|\ell\rangle$, $\ell\in\ZZ$, and hence form a nonresonant connectedness
chain in the sense of
Ref.~\cite[Definition~2.5]{boscain2012spectral}. The nondegenerate
spectrum also makes the condition on degenerate energy levels automatic.
Thus the hypotheses of Theorem~2.11 in
Ref.~\cite{boscain2012spectral} are satisfied. The conclusion of that
theorem is approximate simultaneous controllability in the sense of
Ref.~\cite[Definition~2.10]{boscain2012spectral}, namely approximation of
a prescribed unitary on any finite set of states. The phase-adjustment
argument uses the unbounded-spectrum case of
Ref.~\cite[Lemma~6.3]{boscain2012spectral}, and the present drift spectrum
$\varepsilon_\ell=\ell^2+\xi\ell$ is unbounded.

\section{Gauge-model calculations}

\subsection{Gauge invariance and the Pauli form of hopping}
\label{app:gauge-invariance-pauli-hopping}

\paragraph{Gauge invariance.}
For the oriented link $e=(a,b)$, temporarily omit the factor
$\zeta_{ab}$ and define the forward hopping operator
\begin{equation}
  T_{ab}
  =
  \sigma_a^+S_{ab}U_{ab}\sigma_b^-.
\end{equation}
The required commutators are
\begin{equation}
  [\hat n_a,\sigma_a^+]=\sigma_a^+, \qquad [\hat n_b,\sigma_b^-]=-\sigma_b^-, \qquad [E_{ab},U_{ab}]=U_{ab}.
\end{equation}
By construction, the Jordan--Wigner string $S_{ab}$ contains neither
endpoint and therefore commutes with $\hat n_a$ and $\hat n_b$. The
background charge $b_n$ and external charge $q_n^{\mathrm{ext}}$ are scalar
terms and do not contribute to the commutators.

The link $e=(a,b)$ enters the Gauss-law generator with sign $+1$ at $a$
and sign $-1$ at $b$. Hence
\begin{equation}
  \begin{split}
    [G_a,T_{ab}]
    &=
    [E_{ab}-\hat n_a,T_{ab}]
    =
    T_{ab}-T_{ab}
    =
    0,
    \\
    [G_b,T_{ab}]
    &=
    [-E_{ab}-\hat n_b,T_{ab}]
    =
    -T_{ab}+T_{ab}
    =
    0.
  \end{split}
\end{equation}
Every other Gauss-law generator commutes with $T_{ab}$, and the same
calculation applies to $T_{ab}^\dagger$. Therefore $[G_n,H_{K,ab}]=0$ for every $n$.

The electric and mass terms are functions of the mutually commuting
operators $E_e$ and $\hat n_n$, and hence commute with every $G_n$. For
the magnetic term, a plaquette operator changes the link fluxes by the
oriented boundary of the plaquette. The incidence and boundary matrices
satisfy $DB^{\transpose}=0$,
so the plaquette operators also commute with every Gauss-law generator.
Consequently,
$
  [G_n,H]=0
$
for every $n.$

\paragraph{Pauli form of the hopping term.}
With the occupation convention in
Eq.~\eqref{eq:gauge-matter-occupation},
$
  \sigma^+
  =
  \frac{X+iY}{2},
  \sigma^-
  =
  \frac{X-iY}{2}.
$
Substituting these expressions together with
$\zeta_{ab}=e^{i\delta_{ab}}$ and
$U_{ab}=e^{i\hat\theta_{ab}}$ gives
\begin{equation}
  \zeta_{ab}T_{ab} +\zeta_{ab}^*T_{ab}^\dagger = \frac{1}{2}\Big[ \cos(\hat\theta_{ab}+\delta_{ab})
  (X_aS_{ab}X_b+Y_aS_{ab}Y_b) + \sin(\hat\theta_{ab}+\delta_{ab}) (X_aS_{ab}Y_b-Y_aS_{ab}X_b) \Big].
\end{equation}
Multiplying by $\kappa/2$ gives
Eq.~\eqref{eq:gauge-hopping-pauli-form}.

\subsection{Full-space hopping realization}
\label{app:gauge-exact-full-space-hopping-realization}

Consider first the controlled shift using the qubit at site $a$ as the
control:
$
  V_{ab}^{(a)}
  =
  |0\rangle\langle0|_a\otimes I_e
  +
  |1\rangle\langle1|_a\otimes U_e^\dagger.
$
With the occupation convention used here,
$
  \sigma_a^+=|0\rangle\langle1|_a,
  \sigma_a^-=|1\rangle\langle0|_a.
$
Direct multiplication gives
\begin{equation}
  V_{ab}^{(a)} \sigma_a^+ V_{ab}^{(a)\dagger} = \sigma_a^+\otimes U_e, \qquad V_{ab}^{(a)} \sigma_a^-
  V_{ab}^{(a)\dagger} = \sigma_a^-\otimes U_e^\dagger.
\end{equation}

Alternatively, using the qubit at site $b$ as the control gives
$
  V_{ab}^{(b)}
  =
  |0\rangle\langle0|_b\otimes I_e
  +
  |1\rangle\langle1|_b\otimes U_e,
$
and
\begin{equation}
  V_{ab}^{(b)} \sigma_b^- V_{ab}^{(b)\dagger} = \sigma_b^-\otimes U_e, \qquad V_{ab}^{(b)} \sigma_b^+
  V_{ab}^{(b)\dagger} = \sigma_b^+\otimes U_e^\dagger.
\end{equation}

The Jordan--Wigner string $S_{ab}$ acts only on the sites between $a$ and
$b$ in the chosen ordering. It therefore commutes with both
$V_{ab}^{(a)}$ and $V_{ab}^{(b)}$. Substitution into
Eq.~\eqref{eq:gauge-qubit-exchange-hamiltonian} gives
Eq.~\eqref{eq:gauge-full-space-hopping-conjugation} for either choice of
control qubit.

For $\delta_{ab}=0$, define
$
  P_1=X_aS_{ab}X_b
$
and
$  
  P_2=Y_aS_{ab}Y_b.
$
At sites $a$ and $b$, the Pauli operators $X$ and $Y$ anticommute. The
two minus signs cancel, and the Pauli factors on all other sites are
identical. Hence
$
  [P_1,P_2]=0.
$
The qubit exchange Hamiltonian therefore becomes
$
  H_{ab}^{(q)}(0)
  =
  \frac{\kappa}{4}(P_1+P_2).
$

Let
$
  R_a(\delta)
  =
  e^{i\delta Z_a/2}.
$
Its action on the qubit raising and lowering operators is
\begin{equation}
  R_a(\delta) \sigma_a^+ R_a(\delta)^\dagger = e^{i\delta}\sigma_a^+, \qquad R_a(\delta) \sigma_a^-
  R_a(\delta)^\dagger = e^{-i\delta}\sigma_a^-.
\end{equation}
Consequently,
$
  H_{ab}^{(q)}(\delta)
  =
  R_a(\delta)
  H_{ab}^{(q)}(0)
  R_a(\delta)^\dagger.
$
Since $P_1$ and $P_2$ commute,
\begin{equation}
  e^{-i\tau H_{ab}^{(q)}(\delta_{ab})} = R_a(\delta_{ab}) e^{-i\tau\kappa P_1/4} e^{-i\tau\kappa P_2/4}
  R_a(\delta_{ab})^\dagger,
\end{equation}
which proves
Eq.~\eqref{eq:gauge-qubit-exchange-decomposition}.

\subsection{Projected and cyclic shifts}
\label{app:gauge-projected-and-cyclic-shifts}

For the projector $\Pi_L$ in
Eq.~\eqref{eq:gauge-flux-projector}, the projected shift acts as
\begin{equation}
  U_{e,L}|\ell\rangle
  =
  \begin{cases}
    |\ell+1\rangle,
      & -L\leq\ell<L,\\
    0,
      & \ell=L.
  \end{cases}
\end{equation}
Its adjoint acts as
\begin{equation}
  U_{e,L}^\dagger|\ell\rangle
  =
  \begin{cases}
    |\ell-1\rangle,
      & -L<\ell\leq L,\\
    0,
      & \ell=-L.
  \end{cases}
\end{equation}
Computing
$U_{e,L}^\dagger U_{e,L}$ and
$U_{e,L}U_{e,L}^\dagger$
gives Eq.~\eqref{eq:gauge-projected-shift-defect}.

The operator $\widetilde U_{e,L}$ in
Eq.~\eqref{eq:gauge-cyclic-completion} cyclically permutes the
$2L+1$ momentum states and is therefore unitary. For every nonboundary
momentum state,
\begin{equation}
  [E_{e,L},\widetilde U_{e,L}]|\ell\rangle = \widetilde U_{e,L}|\ell\rangle, \qquad -L\leq\ell<L.
\end{equation}
For the boundary state $|L\rangle$,
\begin{equation}
  [E_{e,L},\widetilde U_{e,L}]|L\rangle
  =
  -2L|-L\rangle.
\end{equation}
Since
$
  \left(
    \widetilde U_{e,L}
    -(2L+1)|-L\rangle\langle L|
  \right)|L\rangle
  =
  -2L|-L\rangle,
$
these relations give
Eq.~\eqref{eq:gauge-cyclic-commutator-defect}.

A finite-dimensional controlled shift defined using
$\widetilde U_{e,L}$ agrees with the full rotor shift for
$-L\leq\ell<L$, but maps the boundary state according to
$
  |L\rangle\longmapsto|-L\rangle.
$

\subsection{Finite Gauss-law basis and numerical details}
\label{app:gauge-finite-gauss-law-basis-numerical-details}

For the geometry in Fig.~\ref{fig:gauge-two-plaquette-geometry}, the sites are
numbered as $1=(0,0)$, $2=(1,0)$, $3=(2,0)$, $4=(0,1)$, $5=(1,1)$, and
$6=(2,1)$. The oriented links, listed from tail to head, are
$e_1=(1,2)$, $e_2=(2,3)$, $e_3=(4,5)$, $e_4=(5,6)$, $e_5=(1,4)$,
$e_6=(2,5)$, and $e_7=(3,6)$. The Jordan--Wigner snake order is
$1\prec2\prec3\prec6\prec5\prec4$. With sites and links in these orders,
the incidence matrix is
\begin{equation}
  D=
  \begin{pmatrix}
     1& 0& 0& 0& 1& 0& 0\\
    -1& 1& 0& 0& 0& 1& 0\\
     0&-1& 0& 0& 0& 0& 1\\
     0& 0& 1& 0&-1& 0& 0\\
     0& 0&-1& 1& 0&-1& 0\\
     0& 0& 0&-1& 0& 0&-1
  \end{pmatrix}.
\end{equation}
Ordering the plaquettes as $(p_1,p_2)$ gives
\begin{equation}
  B=
  \begin{pmatrix}
    1&0&-1&0&-1&1&0\\
    0&1&0&-1&0&-1&1
  \end{pmatrix}
  =
  \begin{pmatrix}
    b_{p_1}^{\transpose}\\
    b_{p_2}^{\transpose}
  \end{pmatrix},
  \qquad
  DB^{\transpose}=0.
\end{equation}

Set $q_n^{\mathrm{ext}}=0$. For occupation eigenvalues
$\mathbf o=(o_n)_n$, defined by
$
  \hat n_n|\mathbf o;E\rangle
  =
  o_n|\mathbf o;E\rangle,
$
the finite physical basis is
\begin{equation}
  \mathcal B_L = \left\{ |\mathbf o;E\rangle: o_n\in\{0,1\},\quad |E_e|\leq L,\quad DE=\mathbf o-\mathbf b
  \right\}.
\end{equation}
For each allowed $\mathbf o$, choose one integer solution
$E_0(\mathbf o)$ of $DE_0=\mathbf o-\mathbf b$.  On the open
$2\times1$ strip, every solution can be written as
\begin{equation}
  E=E_0(\mathbf o)+B^{\transpose}h, \qquad h\in\ZZ^2, \qquad |E_e|\leq L.
\end{equation}
The two integer components of $h$ are enumerated subject to the final flux
bounds.

The diagonal terms act by
\begin{equation}
  H_E|\mathbf o;E\rangle
  =\frac{g^2}{2}\left(\sum_eE_e^2\right)|\mathbf o;E\rangle,
\end{equation}
\begin{equation}
  H_M|\mathbf o;E\rangle
  =m_0\left(\sum_n\eta_no_n\right)|\mathbf o;E\rangle.
\end{equation}
Let $b_p=(B_{pe})_e$ denote the oriented flux vector around plaquette $p$.
The magnetic term acts as
\begin{equation}
  H_B|\mathbf o;E\rangle ={}\frac{N_p}{g^2a_{\mathrm{lat}}^2}|\mathbf o;E\rangle -\frac{1}{2g^2a_{\mathrm{lat}}^2}\sum_p \left(
  |\mathbf o;E+b_p\rangle +|\mathbf o;E-b_p\rangle \right),
  \label{eq:gauge-magnetic-matrix-action}
\end{equation}
where a state is omitted when any final flux lies outside the window.

For $e=(a,b)$, let $I_{ab}$ be the sites strictly between the two endpoints
in the Jordan--Wigner ordering and set
\begin{equation}
  s_{ab}(\mathbf o)
  =(-1)^{\sum_{m\in I_{ab}}o_m}.
\end{equation}
When $o_a=0$ and $o_b=1$, the forward hopping rule at
$\zeta_{ab}=1$ is
\begin{equation}
  H_{K,ab}|\mathbf o;E\rangle
  \supset
  \frac{\kappa}{2}s_{ab}(\mathbf o)
  |\mathbf o^{a\leftarrow b};E+\delta_e\rangle.
\end{equation}
Here $\delta_e$ is the unit flux vector on link $e$, and
$\mathbf o^{a\leftarrow b}$ moves one matter occupation from $b$ to $a$.
For $o_a=1$ and $o_b=0$, the reverse rule is
\begin{equation}
  H_{K,ab}|\mathbf o;E\rangle
  \supset
  \frac{\kappa}{2}s_{ab}(\mathbf o)
  |\mathbf o^{b\leftarrow a};E-\delta_e\rangle.
\end{equation}
Terms whose final flux leaves the window are omitted.  The two rules are
adjoints and preserve the finite Gauss-law basis.

The finite-flux cutoff calculation uses the parameters and cutoffs listed in
Sec.~\ref{subsec:finite-flux-benchmarks}.  For
$L=1,2,3,4,6$, the dimensions of the physical basis are
$90,302,634,1086,$ and $2350$.  The $L=1$ matrix was diagonalized with
\texttt{scipy.linalg.eigh}; for $L\geq2$ we used the shift-invert solver
\texttt{scipy.sparse.linalg.eigsh} with tolerance $2\times10^{-11}$.  Over
the full scan, the largest residual of the two computed eigenpairs satisfies
\begin{equation}
  \max_k\lVert H\psi_k-\mathcal E_k\psi_k\rVert_2
  \leq 6.51\times10^{-10}.
\end{equation}
The numerical environment used Python 3.11.4, NumPy 2.4.6, SciPy 1.17.1,
and Matplotlib 3.10.9.

The processed data also record the link-averaged boundary occupation
\begin{equation}
  \overline p_{\partial L}
  =\frac{1}{N_\ell}
   \sum_e\left\langle\mathbf 1_{\{|E_e|=L\}}\right\rangle.
  \label{eq:gauge-boundary-occupation-diagnostic}
\end{equation}
Its maximum over the scan decreases from $1.05\times10^{-2}$ at $L=4$ to
$6.72\times10^{-4}$ at $L=6$.  Taking $L=6$ as the largest available cutoff,
the maximum absolute differences between $L=4$ and $L=6$ are
$2.25\times10^{-2}$ for $\overline W$, $1.57\times10^{-1}$ for $\Delta$,
and $5.79\times10^{-1}$ for
$\langle E^2\rangle_{\mathrm{link}}$; all three maxima occur at
  $g^{-2}=10$.  On the sampled interval $g^{-2}\leq1$, the corresponding
  maxima are $1.34\times10^{-8}$, $4.69\times10^{-10}$, and
  $7.23\times10^{-9}$.

The plaquette expectation values use the same projected shifts as
Eq.~\eqref{eq:gauge-magnetic-matrix-action}; the remaining observables are
diagonal in the flux basis.  These data give the three panels in
Fig.~\ref{fig:gauge-ground-state-cutoff}.

For Fig.~\ref{fig:gauge-matter-dynamics}, the initial state is
$|\mathbf o=\mathbf b;E=\mathbf0\rangle$. The projected Hamiltonian was
exponentiated on the finite Gauss-law basis for $L=1,2,3,4$, $g^{-2}=1$,
and $0\leq t\leq8$ using $241$ equally spaced time points. Across all
reported calculations, the maximum norm error was $5.67\times10^{-13}$ and
the maximum total-matter-number error was $1.70\times10^{-12}$.

\begin{table}[H]
  \centering
  \caption{Real-time cutoff comparison on the $241$ equally spaced sample
  times in $0\leq t\leq8$. The first six rows compare the successive cutoffs
  $L=2,3,4$; the final two rows report the largest link-averaged boundary
  occupation at $L=3$ and $L=4$.}
  \begin{tabular}{lc}
    \toprule
    Diagnostic & Maximum over the sampled times \\
    \midrule
    $|\overline W_{2}(t)-\overline W_{3}(t)|$
      & $1.33\times10^{-2}$ \\
    $|C_{12,2}(t)-C_{12,3}(t)|$
      & $1.18\times10^{-2}$ \\
    $|I_{M,2}(t)-I_{M,3}(t)|$
      & $4.72\times10^{-3}$ \\
    $|\overline W_{3}(t)-\overline W_{4}(t)|$
      & $9.99\times10^{-5}$ \\
    $|C_{12,3}(t)-C_{12,4}(t)|$
      & $2.28\times10^{-4}$ \\
    $|I_{M,3}(t)-I_{M,4}(t)|$
      & $3.54\times10^{-5}$ \\
    $\overline p_{\partial L}(t)$ at $L=3$
      & $7.60\times10^{-5}$ \\
    $\overline p_{\partial L}(t)$ at $L=4$
      & $2.84\times10^{-7}$ \\
    \bottomrule
  \end{tabular}
\end{table}

\clearpage
\section{Calculations for rotor quantum phase estimation}

\subsection{Fej\'er kernel and weak convergence for direct-angle readout}
\label{app:qpe-fejer-kernel-weak-convergence-direct-angle-readout}

For $\Delta\notin2\pi\ZZ$, the finite geometric series satisfies
\begin{equation}
  \sum_{\ell=0}^{d-1}e^{-i\ell\Delta} = e^{-i(d-1)\Delta/2} \frac{\sin(d\Delta/2)} {\sin(\Delta/2)}.
  \label{eq:app-qpe-finite-geometric-sum}
\end{equation}

The angle-state convention in
Sec.~\ref{subsec:hybrid-register-phase-space} gives
\begin{equation}
  \langle\theta|\chi_d(\varphi)\rangle_R = \frac{1}{\sqrt{2\pi d}} \sum_{\ell=0}^{d-1}
  e^{i\ell(\theta-\varphi)}.
\end{equation}
Taking the squared modulus and using
Eq.~\eqref{eq:app-qpe-finite-geometric-sum} gives
Eq.~\eqref{eq:qpe-direct-angle-readout-distribution}.

The density is normalized because
\begin{equation}
  \int_0^{2\pi}
  e^{i(\ell-k)\theta}\,d\theta
  =
  2\pi\delta_{\ell k}.
\end{equation}
Therefore
\begin{equation}
  \int_0^{2\pi} p_d(\theta\mid\varphi)\,d\theta = \frac{1}{d} \sum_{\ell,k=0}^{d-1}
  \delta_{\ell k} = 1.
\end{equation}

The numerator of
Eq.~\eqref{eq:qpe-direct-angle-readout-distribution} vanishes when
$
  \theta-\varphi
  =
  \frac{2\pi m}{d}
  \pmod{2\pi},
  m\in\ZZ.
$
The closest nonzero choices are $m=\pm1$, so the first zeros occur at
circular distance
$
  \frac{2\pi}{d}
$
from the central maximum at $\theta=\varphi$.

Moreover,
\begin{equation}
  2\pi
  p_d(\theta\mid\varphi)
  =
  K_{d-1}(\theta-\varphi),
\end{equation}
where $K_{d-1}$ is the periodic Fej\'er kernel of order $d-1$.
A standard property of the Fej\'er kernels gives, for every continuous
$2\pi$-periodic function $f$,
\begin{equation}
  \lim_{d\to\infty} \int_0^{2\pi} f(\theta) p_d(\theta\mid\varphi)\,d\theta = f(\varphi).
\end{equation}
Thus the probability distribution becomes concentrated at
$\theta=\varphi$ as $d$ increases.

\subsection{Finite-support probe optimization}
\label{app:qpe-finite-support-probe-optimization}

Let
\begin{equation}
  \Pi_L:=\sum_{\ell=-L}^{L}|\ell\rangle\langle\ell|,
  \qquad
  X_{R,L}:=\Pi_LX_R(1)\Pi_L.
\end{equation}
Conjugation by $Z_R(\phi)$ gives
\begin{equation}
  Z_R(\phi)^\dagger X_{R,L}Z_R(\phi)=e^{-i\phi}X_{R,L}.
\end{equation}
The phase of $\langle\eta|X_{R,L}|\eta\rangle$ can therefore be removed
without changing the momentum support. With both maxima taken over normalized
states in $\Pi_L\bH_{0,1}$,
\begin{equation}
  \max_{\|\eta\|=1}
  \left|\langle\eta|X_{R,L}|\eta\rangle\right|
  =
  \max_{\|\eta\|=1}
  \left\langle\eta\left|
    \frac{X_{R,L}+X_{R,L}^\dagger}{2}
  \right|\eta\right\rangle.
\end{equation}
In the ordered basis $|-L\rangle,\ldots,|L\rangle$, the Hermitian operator
on the right is the path matrix with $1/2$ on its first off-diagonals. Its
largest eigenvalue and a normalized eigenvector are
\begin{equation}
  \cos\!\left(\frac{\pi}{2(L+1)}\right),
  \qquad
  \frac{1}{\sqrt{L+1}}
  \sum_{\ell=-L}^{L}
  \cos\!\left(\frac{\pi\ell}{2(L+1)}\right)|\ell\rangle.
\end{equation}
This proves the fixed-support optimality of
Eq.~\eqref{eq:qpe-cosine-window-probe} and the value of
$R_{\mathrm{cos}}$ used in Eq.~\eqref{eq:qpe-cosine-circular-variance}.

Set $\alpha=\pi/[2(L+1)]$. The energy of the cosine-window probe is
\begin{equation}
  E_{\mathrm{cos}}
  =
  \frac{1}{L+1}
  \sum_{\ell=-L}^{L}\ell^2\cos^2(\alpha\ell).
\end{equation}
Using $\cos^2t=(1+\cos2t)/2$ and differentiating the Dirichlet kernel gives
\begin{equation}
  \sum_{\ell=-L}^{L}\ell^2
  =
  \frac{L(L+1)(2L+1)}{3},
  \qquad
  \sum_{\ell=-L}^{L}\ell^2\cos(2\alpha\ell)
  =
  (L+1)\bigl[L+1-\csc^2\alpha\bigr].
\end{equation}
Substitution gives Eq.~\eqref{eq:qpe-cosine-energy-exact}. The expansion
\begin{equation}
  \csc^2\alpha
  =
  \alpha^{-2}+\frac{1}{3}+\frac{\alpha^2}{15}+O(\alpha^4)
\end{equation}
then gives Eq.~\eqref{eq:qpe-cosine-energy-asymptotic}, and combining this
result with Eq.~\eqref{eq:qpe-cosine-circular-variance} proves
Eq.~\eqref{eq:qpe-cosine-energy-scaling}.

\subsection{Fixed-energy probe asymptotics}
\label{app:qpe-fixed-energy-probe-asymptotics}

For large $\lambda$, the ground state of
Eq.~\eqref{eq:qpe-mathieu-hamiltonian} is localized near $\theta=0$. Expanding
the potential gives
\begin{equation}
  -\lambda\cos\theta
  =
  -\lambda+\frac{\lambda}{2}\theta^2+O(\lambda\theta^4),
\end{equation}
so the leading local Hamiltonian is
\begin{equation}
  H_\lambda+\lambda
  \simeq
  \hat\ell^2+\frac{\lambda}{2}\theta^2.
\end{equation}
The ground-state variances in this local quadratic approximation satisfy
\begin{equation}
  \langle\hat\ell^2\rangle
  \sim
  \sqrt{\frac{\lambda}{8}},
  \qquad
  \langle\theta^2\rangle
  \sim
  \frac{1}{\sqrt{2\lambda}}.
\end{equation}
Writing $E_R=\langle\hat\ell^2\rangle$ gives
$\lambda\sim8E_R^2$ and $\langle\theta^2\rangle\sim1/(4E_R)$. Since
$1-\cos\theta\sim\theta^2/2$ on the localized ground state,
\begin{equation}
  V_\star(E_R)
  \sim
  \frac{1}{8E_R},
  \qquad
  \Delta_\star
  \sim
  \frac{1}{2\sqrt{E_R}}.
\end{equation}
The localization width tends to zero as $\lambda$ grows, so periodic-boundary
corrections do not change the leading terms.

\clearpage

\section{Angle--momentum construction of the rotor QFT}

\subsection{Code-space realization of the mixed angle--momentum phase}
\label{app:qft-mixed-angle-momentum-realization}

Let
$\overline\psi_x(\theta):=\langle\theta|\overline x\rangle$ and
$\Pi_A:=\sum_{x=0}^{d-1}|\overline x\rangle
\langle\overline x|$. Choose pairwise disjoint measurable angular cores
$C_x$ and a bounded real periodic function $g$ satisfying
$g(\theta)=x$ on $C_x$. Write $M_g$ for multiplication by $g$ and define
\begin{equation}
  \eta_x:=\int_{\mathbb T\setminus C_x}|\overline\psi_x(\theta)|^2\,d\theta,
  \qquad
  \Delta_x:=\operatorname*{ess\,sup}_{\theta\notin C_x}|g(\theta)-x|.
\end{equation}
The $x$th column of
$(M_g-D_A)\Pi_A$ is
$(g-x)\overline\psi_x$. The operator norm is bounded by the
Hilbert--Schmidt norm, and hence
\begin{equation}
  \delta_g
  :=\bigl\|(M_g-D_A)\Pi_A\bigr\|
  \leq\left(\sum_{x=0}^{d-1}\Delta_x^2\eta_x\right)^{1/2}.
\end{equation}
Thus the error is controlled by the weight of the orthonormalized
codewords outside their assigned cores. If $0\leq g\leq d-1$, then
$\delta_g\leq(d-1)(\sum_x\eta_x)^{1/2}$.

Let
$\Pi:=\Pi_{M,j}\otimes\Pi_{A,k}$.
The Duhamel formula and
$\|D_M\|=d-1$ give
\begin{equation}
  \bigl\|[e^{i\beta D_{M,j}\otimes M_{g,k}}
  -\Lambda_{jk}(\beta)]\Pi\bigr\|
  \leq |\beta|(d-1)\delta_g.
  \label{eq:app-qft-mixed-phase-multiplication-error}
\end{equation}
Since the logical operation preserves the code space, the same estimate
bounds leakage:
\begin{equation}
  \bigl\|(I-\Pi)e^{i\beta D_{M,j}\otimes M_{g,k}}\Pi\bigr\|
  \leq |\beta|(d-1)\delta_g.
  \label{eq:app-qft-mixed-phase-multiplication-leakage}
\end{equation}

For the CPHS realization in
Eq.~\eqref{eq:qft-mixed-angle-momentum-cphs-reduction}, take
$f_n=\delta_{n0}$. Then
\begin{equation}
  \mathrm{CPHS}_{jk}(\beta)|p\rangle_{R_j}|x\rangle_{R_k}
  =e^{i\beta px}|p\rangle_{R_j}|x\rangle_{R_k}.
\end{equation}
Conjugating by $T_k$ proves
Eq.~\eqref{eq:qft-mixed-angle-momentum-cphs-reduction}. For the general
residue-class code,
\begin{equation}
  \mathrm{CPHS}_{jk}(\beta)|\overline p\rangle_{R_j}|\overline x\rangle_{R_k}
  =\sum_{n,m\in\mathbb Z}f_nf_m e^{i\beta(p+nd)(x+md)}
  |p+nd\rangle_{R_j}|x+md\rangle_{R_k}.
\end{equation}
The $n,m$-dependent phases generally prevent the codeword from acquiring a
uniform logical phase $e^{i\beta px}$.

\subsection{Base-\texorpdfstring{$d$}{d} Fourier phase identity}

For digits $x_k,p_j\in\{0,\ldots,d-1\}$, write
$\mathbf x=(x_1,\ldots,x_r)$ and
$\mathbf p=(p_1,\ldots,p_r)$, and define
\begin{equation}
  X(\mathbf x) = \sum_{k=1}^{r}x_kd^{r-k}, \qquad P_{\mathrm{rev}}(\mathbf p) =
  \sum_{j=1}^{r}p_jd^{j-1}.
  \label{eq:app-base-d-integers}
\end{equation}
Expanding the product in
Eq.~\eqref{eq:app-base-d-integers} gives
\begin{equation}
  \frac{X(\mathbf x)P_{\mathrm{rev}}(\mathbf p)}{d^r} - \sum_{1\leq j\leq k\leq r}
  \frac{p_jx_k}{d^{k-j+1}} = \sum_{j>k}p_jx_kd^{j-k-1} \in \ZZ.
\end{equation}
Consequently,
\begin{equation}
  \exp\!\left[ \frac{2\pi iX(\mathbf x)P_{\mathrm{rev}}(\mathbf p)}{d^r} \right] = \exp\!\left[ 2\pi i
  \sum_{1\leq j\leq k\leq r} \frac{p_jx_k}{d^{k-j+1}} \right].
  \label{eq:app-base-d-phase-identity}
\end{equation}
If
$y_j=p_{r-j+1}$, then
$P_{\mathrm{rev}}(\mathbf p)=\sum_{j=1}^{r}y_jd^{r-j}$, and
\begin{equation}
  \bigotimes_{j=1}^{r} |\overline{p_j}\rangle_{R_j} = \mathrm{REV}_r \bigotimes_{j=1}^{r}
  |\overline{y_j}\rangle_{R_j}.
\end{equation}
Hence $\mathrm{REV}_r$ relates the physical and standard digit orders.

\subsection{Proof of the angle--momentum factorization}
\label{app:qft-angle-momentum-factorization}

Consider the angle-code basis state
\begin{equation}
  |\overline{\mathbf x}\rangle = \bigotimes_{k=1}^{r} |\overline{x_k}\rangle_{R_k}, \qquad \mathbf x =
  \bigl( x_1,\ldots,x_r \bigr).
\end{equation}
Write $X(\mathbf x)$ for the associated input integer in
Eq.~\eqref{eq:app-base-d-integers}.

When processing register \(R_j\), the local angle-to-momentum transform
gives
\begin{equation}
  F_{d,j}^{A\to M} |\overline{x_j}\rangle_{R_j} = \frac{1}{\sqrt d} \sum_{p_j=0}^{d-1}
  \exp\!\left( \frac{2\pi i}{d} p_jx_j \right) |\overline{p_j}\rangle_{R_j}.
\end{equation}
For each \(k>j\), the operation
$\Lambda_{jk}(2\pi/d^{k-j+1})$ contributes
\begin{equation}
  \exp\!\left(
    \frac{2\pi i}{d^{k-j+1}}
    p_jx_k
  \right).
\end{equation}
The total contribution associated with \(p_j\) is therefore
\begin{equation}
  \exp\!\left[ 2\pi i p_j \sum_{k=j}^{r} \frac{x_k}{d^{k-j+1}} \right].
\end{equation}

Multiplying the contributions from all \(r\) register steps gives
\begin{equation}
  U_{\mathrm{AM}} |\overline{\mathbf x}\rangle
  = \frac{1}{\sqrt{d^r}} \sum_{p_1,\ldots,p_r=0}^{d-1}
  \exp\!\left[ 2\pi i \sum_{1\leq j\leq k\leq r}
  \frac{p_jx_k}{d^{k-j+1}} \right]\times
  \bigotimes_{j=1}^{r} |\overline{p_j}\rangle_{R_j}.
  \label{eq:app-qft-am-circuit-output}
\end{equation}

The local transforms give the terms $j=k$ in the exponent of
Eq.~\eqref{eq:app-qft-am-circuit-output}, and the mixed
angle--momentum phases give the terms $j<k$.
Eq.~\eqref{eq:app-base-d-phase-identity} identifies their product
with the standard $d^r$-point Fourier phase and gives the register-order
reversal. Comparing Eq.~\eqref{eq:app-qft-am-circuit-output} with the
definition of
\(F_N^{A\to M}\) now gives
\begin{equation}
  U_{\mathrm{AM}} |\overline{\mathbf x}\rangle = \mathrm{REV}_r
  F_N^{A\to M} |\overline{\mathbf x}\rangle.
  \label{eq:app-qft-am-basis-state-identity}
\end{equation}
Since the states
\(
|\overline{\mathbf x}\rangle
\)
form a basis of
\(\mathcal A_d^{\otimes r}\),
Eq.~\eqref{eq:app-qft-am-basis-state-identity} proves
Eq.~\eqref{eq:angle-momentum-base-d-qft}.

\clearpage
\section{Gate-set compilation of rotor QFTs}
\label{app:qft-gate-set-compilation}

\subsection{Verification of the one-rotor momentum-code QFT}
\label{app:verification-one-rotor-momentum-code-qft}

To verify Eq.~\eqref{eq:compiled-one-rotor-qft-action}, write
$p=\sum_{k=0}^{s-1}2^kp_k$, with $p_k\in\{0,1\}$, and consider the
$j$th stage
$F_j=B_jS_j^\dagger C_j(H_a\otimes I_R)S_jB_j^\dagger$.
Immediately before this stage, the downward-induction invariant is
\begin{equation}
  \ell_j = \sum_{k=0}^{j}2^kp_k + \sum_{k=j+1}^{s-1}2^kc_k.
\end{equation}
The positions $0,\ldots,j$ retain the corresponding input bits, while the
higher positions contain the output bits already generated. Removing bit
$p_j$ gives $u_j=\ell_j-2^jp_j$, whose $j$th binary digit is zero.

The gates $B_j^\dagger$ and $S_j$ extract $p_j$ into the ancilla and remove
that bit from the rotor momentum. The Hadamard and the conditional phase
produce the superposition over $c_j$, after which $S_j^\dagger$ writes
$c_j$ into the vacant position and $B_j$ resets the ancilla. Their combined
action is
\begin{equation}
  \begin{split}
    F_j
    \bigl(
      |0\rangle_a\otimes|\overline{\ell_j}\rangle_R
    \bigr)
    &=
    \frac{|0\rangle_a}{\sqrt2}
    \otimes
    \sum_{c_j=0}^{1}
    (-1)^{c_jp_j}
    e^{i\alpha_jc_ju_j}
    |\overline{u_j+2^jc_j}\rangle_R .
  \end{split}
  \label{eq:rotor-code-one-stage-action}
\end{equation}
For $j=0$, the integer $u_0$ is even and $\alpha_0=\pi$, so $C_0$ acts
trivially on every state reached in that stage and may be omitted.

The invariant holds initially because $\ell_{s-1}=p$. Equation
\eqref{eq:rotor-code-one-stage-action} replaces $p_j$ by $c_j$ and leaves
all other binary positions unchanged, producing the invariant for stage
$j-1$. Moreover,
\begin{equation}
  (-1)^{c_jp_j}e^{i\alpha_jc_ju_j} = \exp\!\left[ 2\pi i c_j \sum_{k=0}^{j} \frac{p_k}{2^{j-k+1}} \right],
\end{equation}
because the terms containing previously generated bits are integer
multiples of $2\pi$ in the exponent. Downward induction therefore gives
\begin{equation}
  F_0F_1\cdots F_{s-1} \bigl( |0\rangle_a\otimes|\overline p\rangle_R \bigr)
  = \frac{|0\rangle_a}{\sqrt d} \otimes \sum_{c_0,\ldots,c_{s-1}=0}^{1}
  \exp\!\left[ 2\pi i \sum_{j=0}^{s-1} c_j
  \sum_{k=0}^{j} \frac{p_k}{2^{j-k+1}} \right]
  \left| \overline{\sum_{j=0}^{s-1}2^jc_j} \right\rangle_R .
\end{equation}

After the internal binary reversal, set
$y=\sum_{j=0}^{s-1}2^{s-1-j}c_j$. Expanding $py/d$ shows that the omitted
terms are integers and hence
\begin{equation}
  e^{2\pi i py/d} = \exp\!\left[ 2\pi i \sum_{j=0}^{s-1} c_j \sum_{k=0}^{j} \frac{p_k}{2^{j-k+1}} \right].
\end{equation}
The map from $(c_0,\ldots,c_{s-1})$ to $y$ is a bijection, so
\begin{equation}
  \widetilde{\mathrm{BREV}}_s F_0F_1\cdots F_{s-1} \bigl( |0\rangle_a\otimes|\overline p\rangle_R \bigr) =
  \frac{|0\rangle_a}{\sqrt d} \otimes \sum_{y=0}^{d-1} e^{2\pi i py/d} |\overline y\rangle_R =
  |0\rangle_a\otimes F_d^M|\overline p\rangle_R .
\end{equation}
This proves Eq.~\eqref{eq:compiled-one-rotor-qft-action}.

The Fourier-stage sequence remains inside $\mathcal M_d^{(0)}$. Each $S_j$
removes one occupied binary digit and each $S_j^\dagger$ writes into the
vacant position, so no borrow or carry occurs.

\subsection{Binary bit reversal within one rotor momentum code}
\label{app:binary-bit-reversal-one-rotor-momentum-code}

For $0\leq\mu<\nu<s$, let
\begin{equation}
  P_{10}^{\mu\nu}
  =
  \sum_{\substack{0\leq p<d\\p_\mu=1,\;p_\nu=0}}
  |\overline p\rangle_R\langle\overline p|_R,
  \qquad
  P_{01}^{\mu\nu}
  =
  \sum_{\substack{0\leq p<d\\p_\mu=0,\;p_\nu=1}}
  |\overline p\rangle_R\langle\overline p|_R .
\end{equation}
Define the momentum-dependent ancilla rotations
\begin{equation}
  R_{10}^{\mu\nu}
  =(-iY_a)\otimes P_{10}^{\mu\nu}
  +I_a\otimes(I_R-P_{10}^{\mu\nu}),
  \qquad
  R_{01}^{\mu\nu}
  =(iY_a)\otimes P_{01}^{\mu\nu}
  +I_a\otimes(I_R-P_{01}^{\mu\nu}).
\end{equation}
These are MQR operations used in the gate-set compilation. With
$\Delta_{\mu\nu}=2^\nu-2^\mu$, set
\begin{equation}
  S_{\mu\nu} = R_{10}^{\mu\nu} \mathrm{CShift}_{a\to R}(-\Delta_{\mu\nu}) R_{01}^{\mu\nu}
  \mathrm{CShift}_{a\to R}(\Delta_{\mu\nu}) R_{10}^{\mu\nu}.
\end{equation}
A direct check of the four patterns $00,01,10,11$ shows that
$S_{\mu\nu}$ exchanges bits $\mu$ and $\nu$, returns the ancilla to
$|0\rangle_a$, and never leaves the code.

The complete binary reversal exchanges the disjoint pairs
$(0,s-1),(1,s-2),\ldots$. Define
\begin{equation}
  \widetilde{\mathrm{BREV}}_s
  =
  \prod_{\mu=0}^{\lfloor s/2\rfloor-1}
  S_{\mu,s-1-\mu}.
\end{equation}
If
$z=\sum_{\mu=0}^{s-1}2^\mu z_\mu$, then
$\operatorname{rev}_s(z)=
\sum_{\mu=0}^{s-1}2^{s-1-\mu}z_\mu$, and
\begin{equation}
  \widetilde{\mathrm{BREV}}_s \bigl( |0\rangle_a\otimes|\overline z\rangle_R \bigr) = |0\rangle_a \otimes
  |\overline{\operatorname{rev}_s(z)}\rangle_R = |0\rangle_a \otimes
  \mathrm{BREV}_s|\overline z\rangle_R .
\end{equation}
For $s=1$, the product is empty and acts as the identity on
$|0\rangle_a\otimes\mathcal M_2^{(0)}$. The operation count is recorded in
Appendix~\ref{app:qft-logical-instruction-counts}.

\subsection{Verification of the compiled base-\texorpdfstring{$d$}{d} QFT}
\label{app:verification-compiled-base-d-qft}

We suppress the reusable ancilla from the notation, since it is initialized
in $|0\rangle_a$ and returned to $|0\rangle_a$ after each compiled
one-rotor transform.

For $
  \mathbf x
  =
  \bigl(x_1,\ldots,x_r\bigr),
  \qquad
  \mathbf p
  =
  \bigl(p_1,\ldots,p_r\bigr),$
write
\begin{equation}
  |\overline{\mathbf x}\rangle_R := \bigotimes_{k=1}^{r} |\overline{x_k}\rangle_{R_k}, \quad
  |\overline{\mathbf p}\rangle_R := \bigotimes_{j=1}^{r} |\overline{p_j}\rangle_{R_j}.
\end{equation}

When register $R_j$ is processed, the local transform
$F_{d,j}^M$ contributes the phase $\exp\!\left(
    \frac{2\pi i p_jx_j}{d}
  \right).$
At that point, every register $R_k$ with $k>j$ still carries its input
digit $x_k$. Hence $
  \mathrm{CPHS}_{jk}
  \left(
    \frac{2\pi}{d^{k-j+1}}
  \right)$
contributes
$
  \exp\!\left(
    \frac{2\pi i p_jx_k}{d^{k-j+1}}
  \right).$

  Multiplying all local and cross-register phases gives
\begin{equation}
  U_{\mathrm{QFT}} |\overline{\mathbf x}\rangle_R = \frac{1}{\sqrt{d^r}}
  \sum_{\mathbf p\in\{0,\ldots,d-1\}^r} \exp\!\left[ 2\pi i \sum_{1\leq j\leq k\leq r}
  \frac{p_jx_k}{d^{k-j+1}} \right] |\overline{\mathbf p}\rangle_R.
  \label{eq:base-d-qft-output}
\end{equation}
The terms with $j=k$ are generated by the local transforms
$F_{d,j}^M$, and the terms with $j<k$ are generated by the CPHS gates.

By Eq.~\eqref{eq:app-base-d-phase-identity}, the phase in
Eq.~\eqref{eq:base-d-qft-output} is
\begin{equation}
  \exp\!\left[ 2\pi i \sum_{1\leq j\leq k\leq r} \frac{p_jx_k}{d^{k-j+1}} \right] = \exp\!\left[
  \frac{2\pi i X(\mathbf x)P_{\mathrm{rev}}(\mathbf p)}{d^r} \right].
\end{equation}
The same derivation shows that the physical ordering of the digits
$p_1,\ldots,p_r$ differs from the standard base-$d$ output
ordering by $\mathrm{REV}_r$. Since
\(\mathbf p\longmapsto P_{\mathrm{rev}}(\mathbf p)\) is a bijection from
$\{0,\ldots,d-1\}^r$ to $\{0,\ldots,d^r-1\}$, it follows that
\begin{equation}
  U_{\mathrm{QFT}} |\overline{\mathbf x}\rangle_R = \mathrm{REV}_rF_N^M
  |\overline{\mathbf x}\rangle_R.
\end{equation}
The states $|\overline{\mathbf x}\rangle_R$ form a basis of
$(\mathcal M_d^{(0)})^{\otimes r}$. Linearity therefore gives
Eq.~\eqref{eq:compiled-base-d-qft-reversal}.

\subsection{Logical instruction counts}
\label{app:qft-logical-instruction-counts}

Each CShift by an arbitrary integer amount, MQR, conditional-phase, CPHS,
and Hadamard gate is counted as one logical instruction.
Set $f_s=\lfloor s/2\rfloor$.

At stage $j=0$, the reduced momentum is even and $\alpha_0=\pi$.
The gate $C_0=C_{c_a,1}(\pi)$ therefore acts trivially on every state
reached in that stage, leaving $s-1$ nontrivial conditional-phase gates.

\begin{table}[H]
  \centering
  \caption{Logical instruction counts for the explicit momentum-code QFT
  compilation. Each CShift of arbitrary size, MQR, conditional-phase, CPHS,
  and Hadamard gate counts as one instruction.}
  \label{tab:qft-compiled-qft-instruction-counts}
  \renewcommand{\arraystretch}{1.2}
  \begin{tabularx}{\textwidth}{@{}L{0.37\textwidth}Y@{}}
    \toprule
    Operation
      & Logical instruction count
      \\
    \midrule
    One compiled $d$-point rotor QFT $F_d^M$
      & $(2s+3f_s)$ MQR $+(2s+2f_s)$ CShift
        $+s$ Hadamard gates $+(s-1)$ conditional-phase gates
        $C_j$
      \\
    Compiled base-$d$ QFT on $r$ rotor momentum codes
      & $r$ applications of $F_d^M$
        $+\binom{r}{2}$ CPHS
      \\
    Qubit--momentum-code transfer, one direction, $r$ registers
      & $rs$ MQR $+rs$ CShift
      \\
    \bottomrule
  \end{tabularx}
\end{table}

\printbibliography[heading=bibintoc]
\end{document}